%% file: ArXiv-VM-YL-main-singlecolumn.tex
\documentclass{article}

\usepackage{times}
\usepackage{fullpage}
\usepackage{multirow}
\usepackage{graphicx,subfigure}
\usepackage{algorithm}
\usepackage[noend]{algorithmic}
\usepackage{amsthm}
\usepackage{arydshln}
\usepackage{pmat}
\usepackage[cmex10]{amsmath}
\usepackage{amssymb,amsfonts}
\usepackage{color}
\newtheorem{remark}{Remark}
\newtheorem{dfr}{Definition}
\newtheorem{theoremm}{Theorem}
\newtheorem{lemmaa}{Lemma}
\newtheorem{proposition}{Proposition}

\newcommand{\argmax}{\operatornamewithlimits{argmax}}

\graphicspath{{./fig2/}}
\usepackage{cite}
\usepackage{url}
\usepackage{textcomp}

\begin{document}

\title{\LARGE Influence Diffusion Dynamics and Influence Maximization
    in Social Networks with Friend and Foe Relationships}

\author{Yanhua Li$^\dag$, Wei Chen$^\S$, Yajun Wang$^\S$ and Zhi-Li Zhang$^\dag$ \\
{$^\dag$Dept. of Computer Science \& Engineering, Univ. of Minnesota, Twin Cities, Minneapolis, MN, US}\\
{$^\S$ Microsoft Research Asia, Beijing, China}\\
{\{yanhua,zhzhang\}@cs.umn.edu,\{weic,yajunw\}@microsoft.com}}
\date{}
\maketitle \sloppy

\input{ArXiv-sec-02-abs}
\input{ArXiv-sec-12-intro}
\input{ArXiv-sec-2motivation}

\input{ArXiv-sec-3Dynamics}
\input{ArXiv-sec-4Maximization}

\input{ArXiv-sec-5EvaluationNew}
\input{ArXiv-sec-6con}

\section{Acknowledgement} We would like to thank
Christian Borgs and Jennifer T. Chayes for pointing out the
relations between the signed digraph voter model and concepts in
physics, such as Ising model and Gauge transformations.
%% epidemic spreading models
%% with two absorbing states, e.g. Ising model and invasion process.
We also thank Zhenming Liu for many useful discussions on this work.
This work was mostly done while the first author was working as
full-time
intern at Microsoft Research Asia.% and
%the anonymous reviewers for their thoughtful comments.
\bibliographystyle{abbrv}
\bibliography{Yanhua-long-bak}
%\clearpage
\appendix
\input{ArXiv-sec-7app}

\end{document}

%% file: ArXiv-sec-02-abs.tex
\begin{abstract}
{\color{black} Influence diffusion and influence maximization in
large-scale online social networks (OSNs) have been extensively
studied because of their impacts on enabling effective online viral
marketing. Existing studies focus on social networks with only
friendship relations, whereas the foe or enemy relations that
commonly exist in many OSNs, e.g., Epinions and Slashdot, are
completely ignored. In this paper, we make the first attempt to
investigate the influence diffusion and influence maximization in
OSNs with both friend and foe relations, which are modeled using
positive and negative edges on signed networks. In particular, we
extend the classic voter model to signed networks and analyze the
dynamics of influence diffusion of two opposite opinions.
{\color{black} We first provide systematic characterization of both
short-term and long-term dynamics of influence diffusion in this
model, %. The analytical results naturally lead to structural balance
%theory, namely,
%and illustrate that the steady state behaviors of the dynamics
%depend on three types of graph structures, which we refer to as
%balanced graphs, anti-balanced graphs, and strictly unbalanced
%graphs.
and illustrate that the steady state behaviors of the dynamics
depend on three types of graph structures, which we refer to as
balanced graphs, anti-balanced graphs, and strictly unbalanced
graphs.} We then apply our results to solve the influence
maximization problem and develop efficient algorithms to select
initial seeds of one opinion that maximize either its short-term
influence coverage or long-term steady state influence coverage.
Extensive simulation results on both synthetic and real-world
networks, such as Epinions and Slashdot, confirm our theoretical
analysis on influence diffusion dynamics, and demonstrate the
efficacy of our influence maximization algorithm
over other heuristic algorithms. %Moreover, our identification of the
%class of anti-balanced graphs is new and could be of independent
%interests to social network analysis.
}

\end{abstract}

%\begin{IEEEkeywords}
\textbf{Keywords:} Signed networks, voter model, influence
maximization, social networks
%\end{IEEEkeywords}

%\category{C.2.2}{Computer-Communication Networks}{Miscellaneous}
%
%\terms{Theorem, Algorithm, Performance}

%\begin{IEEEkeywords} Dynamic programming, Forwarder list, Wireless
%routing, Opportunistic Routing.
%\end{IEEEkeywords}
%
%\bibliographystyle{abbrv}
%\bibliography{Yanhua}             % 9)References
%\end{document}

%% file: ArXiv-sec-12-intro.tex
\section{Introduction}

As the popularity of online social networks (OSNs) such as Facebook
and
    Twitter continuously increases, OSNs have become an important
    platform for the dissemination of news, ideas, opinions, etc.
The openness of the OSN platforms and the richness of contents and
    user interaction information enable intelligent online recommendation
    systems and viral marketing techniques.
For example, if a company wants
    to promote a new product, it may identify a set of influential users in the online social network
    and provide them with free sample products.
They hope that these influential users could influence their
friends,
    and friends of friends in the network and so on, generating a large
    influence cascade so that many users adopt their product
    as a result of such word-of-mouth effect.
The question is how to select the initial users given a limited
budget on free samples, so as to influence the largest number of
people to purchase the product through this ``word-of-mouth''
process. Similar situations could apply to the promotion of ideas
and opinions,
    such as political
    candidates trying to find early supporters for their political proposals
    and agendas, government authorities or companies trying to win
    public support by finding and convincing an initial set of
    early adopters to their ideas.

The above problem is referred to as the
    {\em influence maximization} problem in the literature, which
    has been extensively studied in recent
    years~\cite{KDD03JonKl,KS06,JL07,even2007note,NN08,CWY09,chen2010scalable,chen2010scalableIDCM,icdm11laks,pvldb11Laks,PBS10}.
In these studies, several influence diffusion models are proposed
    to
% Yajun
% mimic
%
% model
% Yanhua
formulate
the
    underlying influence propagation processes, including linear threshold
    (LT) model, independent cascade (IC) model, voter model, etc.
A number of approximation algorithms and scalable heuristics are
    designed under these models to solve the influence maximization problem.

{\color{black} However, all existing studies only look at networks
with positive (i.e., friend, altruism, or trust) relationships,
where in reality, relationships also include negative ones, such as
foe, spite or distrust relationships. In Ebay, users develop trust
and distrust in agents in the network;
%% Amazon allows members to
%% express their likes and dislikes toward the products;
In online review and news forums, such as Epinions and Slashdot,
readers approve or denounce reviews
and articles of each other. %; certain nodes in a computer network may
%be detected as byzantine by a sub-network, and communication through
%a byzantine
%node is considered unreliable. %In all such cases, it is helpful to
%think of the edges between entities as being signed either
%positively or negatively.
%The study of signed networks dates back to the early 1950s where
%dislikes and distrusts were modeled as negative weight edges in a
%graph. A formal model encompassing di?erent types of interpersonal
%relationships was proposed and notion of balance in signed graphs
%was de?ned.
%A number of OSNs explicitly record such negative relationships, e.g.
%    Slashdot.org and Epinions.com, and
{\color{black}Some recent
studies~\cite{WWW10Leskovec,leskovec2010signed,chiang2011CIKM}
already look into the network structures with both positive and
negative relationships. As a common sense exploited in many existing
social influence
studies~\cite{KDD03JonKl,even2007note,CWY09,chen2010scalable,chen2010scalableIDCM},
positive relationships carry the influence in a positive manner,
i.e., you would \emph{more likely} trust and adopt your friends'
opinions. In contrast, we consider that negative relationships often
carry influence in a reverse direction
--- if your foe chooses one opinion or votes for one candidate, you
would \emph{more likely} be influenced to do the
opposite. %~\footnote{Here, we consider that these common senses hold
%statistically, namely, one would more likely to adopt (or reject)
%his/her neighbor's opinion, due to the trust (or distrust) to that
%neighbor.}.
This echoes the principles that ``the friend of my enemy is my
enemy'' and ``the enemy of my enemy is my friend''. Structural
balance theory has been developed based on these assumptions in
social science (see Chapter 5 of~\cite{EK10} and the references
therein).
%``enemy of my enemy is my friend'' --- that people
%    found in practice and has been studied by structural balance theory
%    in social science (see Chapter 5 of~\cite{EK10} and the references therein).
We acknowledge that in real social networks, people's reactions to
the influence from their friends or foes could be complicated, i.e.,
one could take the opposite opinion of what her foe suggests for one
situation or topic, but may adopt the suggestion from the same
person for a different topic, because she trusts her foe's expertise
in that particular topic. %For the mathematical modeling work
In this study, %for simplicity,
we consider the influence diffusion for a single topic, where one
always takes the opposite opinion of what her foe suggests. This is
our first attempt to model influence diffusion in signed networks,
and such topic-dependent simplification is commonly employed in
prior influence diffusion studies on unsigned
networks~\cite{KDD03JonKl,even2007note,CWY09,chen2010scalable,chen2010scalableIDCM,icdm11laks}.
Our work aims at providing a mathematical analysis on the influence
diffusion dynamic incorporated with negative relationship and
applying our analysis to the algorithmic problem of influence
maximization.}
%
%Our work aims at providing mathematical and algorithmic analysis on
%the incorporation of negative relationship into the influence
%diffusion model and its impact on the influence maximization
%problem.}

%% Then, how should we incorporate such negative
%% relationship into the influence diffusion models, and how would it
%% affect the influence maximization problem? These are the problems we
%% address in this paper.
}

%% Unfortunately, all of existing studies only focus on the positive
%% social influence, i.e., every user has two states, as ``agree'' and
%% ``unknown'', and the influence from user $i$ to user $j$ means that
%% $j$ is convinced (or influenced) by $i$ to purchase (or trust) an
%% item, or adopt $i$'s opinions. However, in reality, these influences
%% between users could both positive and negative, governed by
%% ``trust'' and ``distrust'' (or ``like'' and ``dislike'') relations
%% on OSNs, such as \texttt{Epinions}~\cite{},
%% \texttt{Slashdot}~\cite{}. Every user in this type of ``signed''
%% online social networks %has three states, namely, ``agree'',
%% %``disagree'' and ``unknown''.
%% could be influenced either positively or negatively. If an influence
%% occurs from user $i$ to user $j$, user $j$ would adopt $i$'s opinion
%% if $j$ trusts $i$, and $j$ would adopt an opposite opinion from
%% $i$'s, if if $j$ distrusts $i$. The existing results cannot capture
%% this fundamental structural difference of signed social networks,
%% and it is still unclear how various influence models would perform
%% and how the initial seeds should be selected on Signed OSNs.

\subsection{Our contributions}

In this paper, we extend the classic voter
    model~\cite{clifford1973model,holley1975ergodic} to incorporate
    negative relationships for modeling the diffusion of opinions in a social
    network.
Given an unsigned directed graph (digraph), the basic voter model
works as follows. At each step, every node in the graph randomly
picks one of its \emph{outgoing}
    neighbors and adopts the opinion of this neighbor.
% Yajun:
%in the previous step.
% Yanhua
%
Thus, the voter model is suitable to interpret and model opinion
diffusions where
    people's opinions may switch back and forth based on their
    interactions
    with other people in the network.
To incorporate negative relationships, we consider signed digraphs
in which
    every directed edge is either positive or negative, and
    we consider the diffusion of two opposite opinions, e.g.,
    black and white colors.
We extend the voter model to signed digraphs, such that at each
step, every node randomly picks one of its outgoing neighbors, and
if the edge
    to this neighbor is positive, the node adopts the neighbor's opinion, but
    if the edge is negative, the node adopts the opposite of the
    neighbor's opinion (Section~\ref{sec:motivation}).

We provide detailed mathematical analysis on the voter model
dynamics for signed networks (Section~\ref{sec:VotermodelDynamics}).
For short-term dynamics, we derive the exact formula for opinion
distribution
    at each step.
For long-term dynamics, we provide closed-form formulas for the
steady
    state distribution of opinions.
We show that the steady state distribution depends on the graph
structure:
    we divide signed digraphs into three classes of graph structures
    --- balanced graphs, anti-balanced graphs, and
    strictly unbalanced graphs,
    each of which leads to a different type of steady state distributions
    of opinions.
While balanced and unbalanced graphs have been extensively studied
by
    structural balance theory in social science~\cite{EK10}, the anti-balanced
    graphs form a new class that has not been covered before, to the
    best of our knowledge.
Moreover, our long-term dynamics not only cover strongly connected
and aperiodic digraphs that most of such studies focus on,
    but also weakly connected and disconnected digraphs, making our
    study
    more comprehensive.

We then study the influence maximization problem under the voter
model
    for signed digraphs (Section~\ref{sec:max}).
The problem here is to select at most $k$ initial white nodes
    while all others are black, so that either in short term or long term
    the expected number of white nodes is maximized.
This corresponds to the scenario where one opinion is dominating the
public
    and an alternative opinion (e.g. a competing political agenda, or a new innovation) tries to win over supporters as much as
possible
    by selecting some initial seeds to influence
    on.
We provide efficient algorithms that find optimal solutions for both
short-term and long-term cases. In particular, for long-term
influence maximization, our algorithm
    provides a comprehensive solution covering weakly connected and
    disconnected signed digraphs, with nontrivial computations on influence coverage
    of seed nodes.

Finally, we conduct extensive simulations on both real-world and
    synthetic networks to verify our analysis and
    to show the effectiveness of our influence maximization algorithm
    (Section~\ref{sec:evaluation}).
The simulation results demonstrate that our influence maximization
    algorithms perform much better than other heuristic algorithms.

To the best of our knowledge, we are the first to study influence
    diffusion and influence maximization in signed networks, and the first
    to apply the voter model to this case and provide %exact and
    efficient algorithms for influence maximization under voter model for signed networks.

\subsection{Related work}
%To the best of our knowledge, we are the first to study the
%influence diffusion and influence maximization on signed digraphs.
In this subsection, we discuss the topics that are closely related
to our problem, such as: (1) influence maximization and voter model,
(2) signed networks, and (3) competitive influence diffusion.

\noindent{\textbf{Influence maximization and voter model.}}
Influence maximization has been extensively studied in the literature. %. ~\cite{KDD03JonKl,KS06,JL07,even2007note,NN08,CWY09,chen2010scalable,chen2010scalableIDCM}.
The initial work~\cite{KDD03JonKl} proposes several influence
diffusion models and provides the greedy approximation algorithm for
influence maximization. More recent
works~\cite{KS06,JL07,even2007note,NN08,CWY09,chen2010scalable,chen2010scalableIDCM,icdm11laks}
study efficient optimizations and scalable heuristics for the
influence maximization problem. In particular, the voter model is
proposed in~\cite{clifford1973model,holley1975ergodic}, and is
suitable for modeling opinion diffusions in which people may switch
opinions back and forth from time to time due to the interactions
with other people in the network. Even-Dar and
Shapira~\cite{even2007note} study the influence maximization problem
in the voter model on simple unsigned and undirected graphs, and
they show that the best seeds for long-term influence maximization
are simply the highest degree nodes. As a contrast, we show in this
paper that seed selection for signed digraphs
    are more sophisticated, especially for weakly connected or
    disconnected signed
    digraphs.
%Our previous
%study~\cite{Chen2011} investigates the independent cascade
%(IC) model with extensions to the scenarios where the negative
%opinions are considered.
%Pathak et al.~\cite{PBS10} study the dynamics of a generalized
%    voter model in unsigned and undirected networks, but they do not study
%    influence maximization in their model.
More voter model related research is conducted in physics domain,
where the voter model, the zero-temperature Glauber dynamics for the
Ising model, invasion process, and other related models of
population dynamics belong to the class of models with two absorbing
states and epidemic spreading dynamics~\cite{sood2008voter,% zschaler2011early
ageles2009conservation, masuda2009evolutionary}. However, none of
these works study the influence diffusion and influence maximization
of voter model under signed networks.
%like the contact
%process [11] or the Susceptible-Infected-Susceptible model [12],
%usually have a single absorbing state with no conservation law.
%voter model is studied as an interacting particle system on complex
%network~\cite{sood2008voter, zschaler2011early,
%ageles2009conservation}, where they investigate the relation between
%the degree distribution of the directed (or undirected) graph and
%the convergence time (i.e., time needed to get a consensus of all
%voters).

%% All in all, these studies are all conducted on unsigned graphs,
%% where in this paper, we will systematically investigate how the
%% voter propagation model would perform on \emph{signed networks},
%% where both positive and negative influences may happen, and develop
%% efficient algorithms to solve the influence maximization problems of
%% voter model on signed digraphs.

\noindent{\textbf{Signed networks.}} The signed networks with both
positive and negative links have gained attentions
    recently~\cite{kunegis:signed-kernels,leskovec2010signed,WWW10Leskovec,borgs2010novel}.
In~\cite{leskovec2010signed,WWW10Leskovec}, the authors
    empirically study the structure of real-world social networks with
    negative relationships based on two social science theories,
    i.e., balance theory and status theory.
Kunegis et al.~\cite{kunegis:signed-kernels} study the spectral
properties of the signed undirected graphs, with applications in
link predictions, spectral clustering, etc. Borgs et
al.~\cite{borgs2010novel} proposes a generalized PageRank algorithm
for signed networks with application to online recommendations,
where the distrust relations are considered as adversarial or
arbitrary user behaviors, thus the outgoing relations of distrusted
users are ignored while ranking nodes. Our algorithm can also be
considered as an influence ranking algorithm that generalizes the
PageRank algorithm, but we treat distrust links as generating
negative influence rather than ignoring distrusted users' opinions,
and thus our ranking method is different from~\cite{borgs2010novel}.
None of the above work studies influence diffusion and influence
maximization
    in signed networks.
%However, the influence propagation models and are

\noindent{\textbf{Competitive influence diffusion.}} A number of
recent studies focus on competitive influence
    diffusion and maximization~\cite{BKS07,BFO10,WWW11,Chen2011,Ma:2008:MSN:1458082.1458115,he2011influence}, in which
    two or more competitive opinions or innovations are diffusing in
    the network.
Although they consider two or more competitive or opposing influence
    diffusions, they are all on unsigned networks, different from
    our study here on diffusion with both positive and negative relationships.

%Hence, none of the above work studies the influence diffusion and
%influence maximization on social networks with both friend and foe
%relationships. We in this paper systematically investigate this
%problem

%% Wei deleted due to space
%% \subsection{Paper organization}
%% In Section~\ref{sec:motivation},
%% we  introduce the voter model for signed digraphs and notations
%%     used in the paper.
%% %problem settings and the introduction of the basic voter model on
%% %signed digraphs. We present the structural analysis of the signed
%% %digraphs in Section~\ref{sec:StrAna}, and
%% In Section~\ref{sec:VotermodelDynamics}, we fully characterize the dynamics of
%% the voter model on signed digraphs.
%% In Section~\ref{sec:max}, we develop efficient
%% algorithms to solve the short-term and short-term influence
%% maximization problem on signed digraphs. %In Section~\ref{sec:ext},
%% %we generalize our results to weakly connected and disconnected
%% %signed digraphs, respectively.
%% In Section~\ref{sec:evaluation}, we demonstrate the efficacy of our
%% seed selection scheme using both synthetic datasets and large scale
%% real datasets. We conclude our paper
%% in Section~\ref{sec:con}.

%% %we study the dynamics of the voter model on signed digraphs, i.e.,
%% %derive exactly how the influence evolve over step, with respect to
%% %various network structures.

%% file: ArXiv-sec-2motivation.tex
\section{Voter model for signed networks}
    \label{sec:motivation}

We consider a weighted directed graph (digraph) $G=(V,E,A)$, where $V$ is
    the set of vertices, $E$ is the set of directed edges, and $A$ is the weighted
    adjacency matrix with $A_{ij} \ne 0$ if and only if $(i,j) \in E$,
    with $A_{ij}$ as the weight of edge $(i,j)$.
The voter model was first introduced for unsigned graphs, with
nonnegative adjacency matrices $A$'s. In this model, each node holds
one of two opposite opinions, represented by
    black and white colors. Initially each node has
either black or white color. At each step $t\ge 1$, every node $i$
randomly picks one outgoing neighbor
    $j$ with the probability proportional to the weight of $(i,j)$, namely
    $A_{ij}/\sum_\ell{A_{i\ell}}$, and changes its color to
    $j$'s color.
The voter model also has a random walk interpretation. If a random
walk starts from $i$ and stops at node $j$ at step $t$, then $i$'s
color at step $t$ is
    $j$'s color at step $0$.

In this paper, we extend the voter model to signed digraphs, in
which the adjacency
    matrix $A$ may contain negative entries.
A positive entry $A_{ij}$ represents that $i$ considers $j$ as a friend
    or $i$ trusts $j$, and a negative $A_{ij}$ means that $i$ considers
    $j$ as a foe or $i$ distrusts $j$.
The absolute value $|A_{ij}|$ represents the strength of this trust or
    distrust relationship.
The voter model is thus extended naturally such that one always
takes the same opinion from his/her friend, and the opposite opinion
of his/her foe. Technically, at each step $t\ge 1$, $i$ randomly
picks one
    outgoing neighbor $j$ with probability
    $|A_{ij}|/\sum_\ell{|A_{i\ell}|}$, and if $A_{ij}>0$ (or edge $(i,j)$
    is positive) then $i$ changes its color to $j$'s color, but if $A_{ij}<0$ (or edge $(i,j)$
    is negative) then $i$ changes its color to the opposite of
    $j$'s color.
The random walk interpretation can also be extended for signed networks:
    if the $t$-step random walk from $i$ to $j$ passes an even number of
    negative edges, then $i$'s color at step $t$ is the same as $j$'s
    color at step $0$; while if it passes an odd number of negative
    edges, then $i$'s color at step $t$ is the opposite of $j$'s
    color at step $0$.

\begin{table}[t]
\centering \caption{Notations and
terminologies}\label{tab:terminology} {\small
%%\begin{tabular}{|p{2cm}|p{14.5cm}|}
\begin{tabular}{|p{2.5cm}|p{12.5cm}|}
%\begin{tabular}{|p{2cm}|p{5.5cm}|}
\hline
$G=(V,E,A)$, $\bar{G}=(V,E,\bar{A})$ & $G$ is a signed digraph, with signed adjacency matrix $A$ and $\bar{G}$ is the unsigned version of $G$, with adjacency matrix $\bar{A}$\\
\hline \vspace{0.00mm} $A^+$, $A^-$ & $A^+$ (resp. $A^-$) is the
non-negative adjacency matrix
    representing positive (resp. negative) edges of $G$, with $A=A^+-A^-$ and
    $\bar{A} = A^+ + A^-$.\\
\hline \vspace{0.02mm} $\textbf{1}$, $\pi$, $x_0$, $x_t$, $x$,
$x_e$, $x_o$ & Vector forms. All vectors are $|V|$-dimensional
    column vectors by default;
    $\textbf{1}$ is all one vector, $\pi$ is the stationary
    distribution of ergodic digraph $\bar G$;
    $x_0$ (resp. $x_t$) is the white color distribution at the beginning (resp.
    at step $t$);
    %% $x_t$ is the color distribution after step $t$;
    $x$ is the steady state white color distribution;
    $x_e$ (resp. $x_o$) is the steady state white color distribution for even
    (resp. odd) steps.
    %i.e., $x_e = \lim_{t\rightarrow \infty}x_{2t}$ and $x_o = \lim_{t\rightarrow \infty}x_{2t+1}$.
    \\
\hline \vspace{0.02mm} $d$, $d^+$, $d^-$, $D$ & $d$, $d^+$, and
$d^-$ are weighted out-degree vectors of $G$, where
$d=\bar{A}\textbf{1}$,
    $d^+=A^+\textbf{1}$, and $d^-=A^-\textbf{1}$;
    $D=diag[d]$ is the diagonal degree matrix filled with entries of $d$.\\
%\hline $g^-$ & The negative generative vector, with g^-(i) \\
\hline \vspace{0.02mm} $P$, $\bar{P}$ & $P=D^{-1} A$ is the signed
transition matrix of $G$ and
$\bar{P}=D^{-1} \bar{A}$ is the transition probability matrix of $\bar{G}$.\\
\hline \vspace{0.02mm} ${v}_{Z}$,  $\hat{v}_{S}$, $\hat{v}_{Z,S_Z}$&
Given a vector $v$, a node set $Z\subseteq V$, ${v}_{Z}$ is the
projection of $v$ on $Z$. Given a partition $S,\bar{S}$ of $V$,
$\hat{v}_{S}$ is signed such that $\hat{v}_{S}(i)=v(i)$ if $i\in S$,
and $\hat{v}_{S}(i)=-v(i)$ if $i\not\in S$. Given a partition
$S_Z,\bar{S}_Z$ of $Z$,
    $\hat{v}_{Z,S_Z}$ is taking the projection of $v$ on $Z$ first,
    then negating the signs for entries in $\bar{S}_Z$.
%% $\textbf{1}$ is unit column vector,
%% $\hat{\textbf{1}}_S=\hat{I}_S{\textbf{1}}$ is signed unit vector,
%% ${\textbf{1}}_{X}$ is the (column) unit vector for component
%% $G_X\subseteq G$. Given a partition $\{S_X, \bar{S}_X\}$ of $G_X$,
%% $\hat{\textbf{1}}_{X,S_X}$ is the signed unit vector on
%% $G_X$.
    \\
\hline \vspace{0.00mm} $I$, $\hat{I}_S$, $B_Z$ & $I$ is the identity
matrix. $\hat{I}_S=diag[\hat{\textbf{1}}_S]$ is the signed identity
matrix.
$B_Z$ is the projection of a matrix $B$ to $Z\subseteq V$.\\
%% \hline
%% \vspace{0.05mm}
%% $I$, $\hat{I}_S$, $I_X$, $\hat{I}_{X,S_X}$ & $I$ is identity
%% matrix. $\hat{I}_S=diag[\hat{I}_S(i)]$ is signed identity matrix,
%% with $\hat{I}_S(i)=1$, if $i\in S$, and $\hat{I}_S(i)=0$, otherwise.
%% $I_X$ is the identity matrix for component $G_X\subseteq G$. Given a
%% partition $\{S_X, \bar{S}_X\}$ of $G_X$,
%% $\hat{I}_{X,S_X}$ is the signed identity matrix on $G_X$.\\
%% \hline
%% \vspace{0.05mm}
%% $\pi$, $\hat{\pi}_S$, ${\pi}_X$, $\hat{\pi}_{X,S_X}$ &
%% ${\pi}$ is the (column) stationary distribution vector of $\bar{G}$.
%% $\hat{\pi}_S=\hat{I}_S\pi$ is signed version. ${\pi}_X$ is (column)
%% the stationary distribution vector for nodes in $G_X\subseteq G$.
%% Given a partition $\{S_X, \bar{S}_X\}$ of $G_X$,
%% $\hat{\pi}_{X,S_X}$ is the signed column stationary vector on$G_X$.\\
\hline
\end{tabular}}\vspace*{-0.2cm}
\end{table}

Given a signed digraph $G=(V,E,A)$, let $G^+=(V,E^+,A^+)$ and
    $G^-=(V,E^-,A^-)$ denote the unsigned subgraphs consisting of all
    positive edges $E^+$ and all negative edges $E^-$, respectively,
    where $A^+$ and $A^-$ are the corresponding
    non-negative adjacency matrices.
Thus we have $A=A^+-A^-$.
%% We will study both the short-term and long-term
%%  dynamics of voter model in signed digraphs.
%% For the long-term dynamics concerning the steady state behavior of the
%%  voter model, we first consider strongly connected digraphs and
%%  then extend the result to weakly connected and disconnected graphs.
Similar to unsigned digraphs, $G$ is {\em aperiodic} if the greatest
    common divisor of the lengths of all cycles in $G$ is $1$, and
    $G$ is {\em ergodic} if it is strongly connected and aperiodic.
A {\em sink component} of a signed digraph is a strongly connected
component that has no outgoing edges to any nodes outside the
component. When studying the long-term dynamics of the voter model,
we assume that all signed strongly connected components are ergodic.
We first study the case of ergodic graphs, and then extend it to the
more general case of weakly connected or disconnected graphs with
ergodic sink components. Table~\ref{tab:terminology} provides
notations and terminologies used in the paper.
%% In the following sections we will intensively use matrix and vector operations,
%%     and we provide some notational conventions here.
%% In the context of a digraph $G=(V,E,A)$ with $|V|=n$,
%%     all vectors are by default column vectors of dimension $n$, and
%%     all matrices are by default square matrices of dimension $n\times n$.
%% For a set $Z\subseteq V$, the projection of a vector $x$ onto set $Z$ is
%%     denoted as $x_Z$; similarly the projection of a matrix $A$ onto
%%     set $Z$ is denoted as $A_Z$.
%% For a set $S\subseteq V$ and an $n$-dimensional vector $x$,
%%     we use $\hat{x}_S$ to denote another $n$-dimensional vector such
%%     that $\hat{x}_S(i)=x(i)$ for all $i\in S$, and
%%     $\hat{x}_S(i)=-x(i)$ for all $i\in V\setminus S$.
%% Similarly, for a diagnal matrix $D$, we use $\hat{D}_S$ to denote
%%     the diagnal matrix such that $\hat{D}_S(i,i)=D(i,i)$ for all $i\in S$,
%%     and $\hat{D}_S(i,i)=-D(i,i)$ for all $i\in V\setminus S$.
%% Moreover, we use $\hat{x}_{Z,S}$ for $S\subseteq Z \subseteq V$ to
%%     denote the $|Z|$-dimension vector obtained by first projecting
%%     vector $x$ to $Z$, and then negating the signs for elements
%%     in $S\setminus Z$.
%% We use $I$ to represent the indentity matrix and $\textbf{1}$ to
%%     represent the all-one vector. A summary
%% of main notations used in this paper can be found in
%% Table~\ref{tab:terminology}.
Note that one basic fact we often use in studying long-term
convergence behavior is: If matrix $P$ satisfies
$\lim_{t\rightarrow\infty}P^t=\mathbf{0}$,
    then $I-P$ is invertible and $(I-P)^{-1}=
    \lim_{t\rightarrow\infty}\sum_{i=0}^{t}P^i$.

%% file: ArXiv-sec-3Dynamics.tex
\newcommand{\zlimit}{\textbf{1}_{SZ}\pi_{SZ}^T}
\newcommand{\limt}{\lim_{t\rightarrow \infty}}
\newcommand{\limm}{\lim_{m\rightarrow \infty}}

\section{Analysis of voter model dynamics on signed digraphs}\label{sec:VotermodelDynamics}

In this section, we study the short-term and long-term dynamics
of the voter model on signed digraphs. In particular, we answer the
following two questions.

\noindent\textbf{(i) Short-term dynamics:} Given an initial
  distribution of black and white nodes, what is the distribution of black
    and white nodes at step $t>0$?

\noindent\textbf{(ii) Convergence of voter model:} Given an initial
  distribution of black and white nodes, would the distribution
  converge, % or diverge,
  and what is the steady state distribution of
    black and white nodes?

\subsection{Short-term dynamics}\label{sec:STVotermodelDynamics}

%% In Section~\ref{sec:motivation}, the dynamics of the voter model in {\em
%%   unsigned digraph} is closely related to the {\em transition
%%   probability matrix} of the underlying graph, i.e., in Eq.(cite the
%% equation in sec 2.). In {\em signed
%%   digraphs}, we define the {\em signed transition matrix} as the
%% analogy of the transitiom probability matrix in unsigned case.

To study voter model dynamics on signed digraphs, we first define
    the {\em signed transition matrix} as follows.

\begin{dfr}[Signed transition matrix]
Given a signed digraph $G=(V,E,A)$, we define the {\em signed
transition matrix} of $G$ as
$P_{}=D^{-1}A$,
where $D=diag[d_i]$ is the diagonal matrix and $d_i=\sum_{j\in V}|A_{ij}|$ is
    the weighted out-degree of node $i$.
\end{dfr}

Next proposition characterizes the dynamics of the voter model at
each step using the {\em signed transition matrix}.

%% {\bf Yajun: what is the ``negative link generative vector''? Do we
%%   define here? Is it a term used in the literature?}

\begin{proposition}
\label{thm:stdyn} Let $G=(V,E,A)$ be a signed digraph and denote the
initial white color distribution vector as $x_0$, i.e., $x_0(i)$
represents the probability that node $i$ is white initially. Then,
the white color distribution at step $t$, denoted by $x_t$ can be
computed as
\begin{align}
x_{t}&={P}^tx_0+(\sum_{i=0}^{t-1}{P}^i)g^-\label{eq:signed2},
\end{align}
where $g^- = D^{-1} A^{-} \textbf{1}$, i.e. $g^-(i)$ is the
    weighted fraction of outgoing negative edges of node $i$.
\end{proposition}
\begin{proof}
Based on the signed digraph voter model defined in
Section~\ref{sec:motivation}, $x_t$ can be iteratively computed as
\begin{align}
x_{t}(i)&=\sum_{j\in V}\frac{A^+_{ij}}{d_i}x_{t-1}(j)+
    \sum_{j\in V}\frac{A^-_{ij}}{d_i}(1-x_{t-1}(j)). \label{eq:voterscalar2}%\\
%&=\sum_{j:i\rightarrow j\in
%E}{P}(i,j)x_{t-1}(j)+\sum_{j:i\rightarrow j:-}\frac{A^-_{ij}}{d_i},
%\nonumber
\end{align}
In matrix form, we have
\begin{align}
x_{t}&={D}^{-1}Ax_{t-1} + {D}^{-1} A^{-} \textbf{1}={P} x_{t-1} +
g^-, \label{eq:voters2}
\end{align}
which yields Eq.(\ref{eq:signed2}) by repeatedly applying
Eq.(\ref{eq:voters2}).\end{proof}

\subsection{Convergence of signed transition matrix with relation
    to structural balance of signed digraphs} \label{sec:balance}

%In this section, we perform structural analysis on signed social
%digraphs, where we will discuss the structural properties and
%convergence on various signed digraphs, such as positive balanced,
%negative balanced, and strictly balanced digraphs.

Eq.(\ref{eq:signed2}) infers that the long-term dynamics,
    i.e., the vector $x_t$ when $t$ goes to infinity, depends
    critically on the limit of $P^t$ and
    $\sum_{i=0}^{t-1}P^i$.
We show below that the limiting behavior of the two matrix sequences
is fundamentally determined by the structural balance of signed
digraph $G$, which connects to the social balance
    theory well studied in the social science literature (cf.~\cite{EK10}).
We now define three types of signed digraphs based on their balance
structures.

\begin{dfr}[Structural balance of signed digraphs]\label{def:BalDig2}
Let $G=(V,E,A)$ be a signed digraph.
\begin{enumerate}
\setlength{\itemsep}{-1ex}
\item {\em \bf Balanced digraph}. $G$ is {\em balanced} if
    there exists a partition $S, \bar S$ of nodes in $V$, such that
  all edges within $S$ and $\bar{S}$ are positive and all edges across
  $S$ and $\bar S$ are negative.
\item {\em \bf Anti-balanced digraph}. $G$ is {\em anti-balanced}
    if there exists a partition $S, \bar S$ of nodes in $V$, such that
  all edges within $S$ and $\bar{S}$ are negative and all edges across
  $S$ and $\bar S$ are positive.
\item {\em \bf Strictly unbalanced digraph}. $G$ is {\em strictly
    unbalanced} if $G$ is neither balanced nor
    anti-balanced.
\end{enumerate}
\end{dfr}

The balanced digraphs defined above correspond to the balanced
graphs originally defined in social balance theory. It is known that
a balanced graph can be equivalently defined by
    the condition that
    all circles in $G$ without considering edge directions contain
    an even number of negative edges~\cite{EK10}.
On the other hand, the concept of anti-balanced digraphs seems not
appearing in the social balance theory. % and is new to us.
Note that
balanced digraphs and anti-balanced digraphs are not mutually
exclusive. For example, a four node circle with one pair of
non-adjacent edges being
    positive and the other pair being negative is both balanced and
    anti-balanced.
However, for studying long-term dynamics, we only need the above
    categorization for aperiodic digraphs, for which we show below that
    balanced digraphs and anti-balanced digraphs are mutually exclusive.

\begin{proposition}
An aperiodic digraph $G$ cannot be both balanced and
anti-balanced.
\end{proposition}
\begin{proof}
% Removed by Yanhua for techreport 11-15-11
% This can be seen by looking at the parity of the number
%    of negative and positive edges respectively on any cycle in a digraph
%    that is both balanced and anti-balanced.

% Wei deleted due to space % Recovered by Yanhua for technical report
{\color{black} Suppose, for a contradiction, that an aperiodic
digraph $G$ is both
     balanced and anti-balanced.
 By the equivalent condition of balanced graphs, we know that all cycles
     of $G$ have an even number of negative edges.
 Since an anti-balanced graph will become balanced if we negate the signs
     of all its edges, we know that all cycles of $G$ also have an even
     number of positive edges.
 Therefore, all cycles of $G$ must have an even number of edges, which
     means their lengths have a common divisor $2$, contradicting to
     the assumption that $G$ is aperiodic.}
\end{proof}

With the above proposition, we know that balanced
    graphs, anti-balanced graphs, and strictly unbalanced graphs indeed
    form a classification of aperiodic digraphs, where anti-balanced graphs and
    strictly unbalanced graphs together correspond to unbalanced graphs
    in the social balance theory.
We identify anti-balanced graphs as a special category because it has
    a unique long-term dynamic behavior different from other graphs.
%An example of anti-balanced graphs is a graph with only negative
%edges. In general, anti-balanced graph could be viewed as an extreme
%in which
%    many hostility exist among individuals.
%We leave it as future work to study the social implications of
%anti-balanced
%    graphs.
An example of anti-balanced graphs is a graph with only negative
edges. In general, anti-balanced graphs could be viewed as an
extreme in which many hostility exist among individuals, e.g.,
networks formed by bidders in
auctions~\cite{morgan2003spite,brandt2007spiteful}.

{\color{black}\noindent{\textbf{Case of ergodic signed digraphs.}}
Now, we discuss the limiting behavior of $P^t$ of ergodic signed
digraphs with three balance structures. A signed digraph $G=(V,E,A)$
is ergodic if and only if for any node $i$, there always exists a
signed path to any other node in $G$ and the common divisor of all
cycle path lengths of $i$ is $1$. Here, a signed path $R$ in a
signed graph $G$ is a sequence of nodes with the edges being
directed from each node to the following one, where the length of
the path, denoted as $|R|$, is the total number of directed edges in
$R$. The sign of a path is positive, if there is an even number of
negative edges along the path; otherwise the sign of a path is
negative. Below, we first introduce
Proposition~\ref{pro:UnbalancedDig} presenting that the balance
structures of ergodic signed digraphs can be interpreted and
distinguished in terms of the path lengths and path signs in $G$. As
a result, Lemma~\ref{thm:BalDig2} introduces the various limiting
behaviors of $P^t$ of ergodic signed digraphs with respect to three
balance structures.

%Given a signed digraph
%$G=(V,E,A)$, let $\bar G =(V,E,\bar{A})$ corresponds to its
%    unsigned version ($\bar{A}_{ij} = |A_{ij}|$ for all $i,j\in V$).
%When $\bar G$ is ergodic, random walk on $\bar G$
%    has a unique stationary distribution, denoted as $\pi$.
%That is, $\pi^T = \pi^T \bar{P}$, where $\bar{P}=D^{-1}\bar{A}$ is
%    the transition probability matrix for $\bar{G}$.
%Henceforth, we always use $S,\bar{S}$ to denote the corresponding
%partition for either balanced graphs or anti-balanced graphs.
 }

% \subsection{Proposition~\ref{pro:UnbalancedDig}}\label{sec:proposition3}
% \begin{theorem}
% \label{thm:exponentialconverge}
% Let $B \in {\mathbb{C}}^{n\times n}$ be a complex matrix.
% If $\lim_{i\rightarrow \infty} B^i = \bf 0$,
% there exist $\lambda \in (0,1)$ and
% constant $T$, such that, for any $i >T$,
% $$||A^i||_\infty \leq \lambda^i.$$
% \end{theorem}
% \begin{proof}
% Let $\rho = \rho(A)$ be the largest absolute value of the eigenvalues
% of $A$. $\lim_{i\rightarrow \infty} =\bf 0$ implies $\rho <1$.

% Consider the standard Jordan form of $A =
% Q^{-1}JQ$. Since $A^i = Q^{-1}J^iQ$, $||A^i||_\infty \leq
% n^2||J^i||_\infty$. Notice that,
% $$||J^i||_\infty < i^n \rho^{i-n}.$$
% Since $\rho^{i-n}$ decreases exponentially while $i^n$ grows
% polynomially, the claim follows.
% \end{proof}

% \noindent\textbf{Proof of Lemma~\ref{lem:unbalancedP}}
{\color{black}
\begin{proposition}[]\label{pro:UnbalancedDig}
Let $G=(V,E,A)$ be an ergodic strictly unbalanced digraph. There
exist two nodes $i$ and $j$, and two directed paths from $i$ to $j$
with the same length but different signs.
\end{proposition}
\begin{proof}
%We prove this proposition by contradiction. In particular,
Given the following three statements, we prove $\texttt{Statement 1}
\Rightarrow \texttt{Statement 2} \Rightarrow \texttt{Statement 3}$,
which in turn proves this proposition, i.e., $\neg\texttt{Statement
3} \Rightarrow \neg\texttt{Statement 1}$. We assume that $G$ is a
signed ergodic digraph.

\noindent\textbf{Statement 1:}  For any two nodes $i$ and $j$, all
paths from $i$ to $j$ with the same length have  same signs.

\noindent\textbf{Statement 2:} For any two nodes $i$ and $j$, all
paths from $i$ to $j$ with even length have same signs.

\noindent\textbf{Statement 3:} $G$ is either balanced or
anti-balanced.

\noindent(1) Proof by contradiction for $\texttt{Statement 1}
\Rightarrow \texttt{Statement 2}$. We assume that in $G$, there
exist two even length paths $R_{e1}$ and $R_{e2}$ from $i$ to $j$
with different signs. Since $G$ is ergodic,
    by Proposition~\ref{pro:evenoddpath} in
    Appendix~\ref{sec:exponentialofbarP},
    there must exist a path, denoted by $R_{o}$, from $j$ to
$i$ with odd length (no matter what sign it carries). Denote the
length of these three paths as $|R_{e1}|$, $|R_{e2}|$ and $|R_o|$,
respectively.

Then, $R_{c1}=R_{e1}+R_o$ forms a cycle at node $i$ with odd length
$|R_{e1}|+|R_o|$ and $R_{c2}=R_{e2}+R_o$ forms another cycle at $i$
with odd length $|R_{e2}|+|R_o|$. Clearly, two cycles $R_{c1}$ and
$R_{c2}$ carry different signs. Then, let
$R'_{c1}=R_{c1}^{|R_{c2}|}$ denote a cycle of node $i$, by
continuing $R_{c1}$ for $|R_{c2}|$ times, which has the same sign
with $R_{c1}$ since $|R_{c2}|$ is odd. Similarly, we construct a
cycle $R'_{c2}=R_{c2}^{|R_{c1}|}$  by continuing $R_{c2}$ for
$|R_{c1}|$ times, which has the same sign as $R_{c2}$. Thus
$R'_{c1}$ and $R'_{c2}$ have the same length of $|R_{c1}||R_{c2}|$
but different signs, which contradicts to Statement 1.

\noindent(2) Proof for $\texttt{Statement 2} \Rightarrow
\texttt{Statement 3}$.
By Proposition~\ref{pro:evenoddpath} in
    Appendix~\ref{sec:exponentialofbarP}, we know that between any two nodes
    there must exist even-length paths.
By Statement 2, we partition $V$ into $S$
and $\bar S$, based on the signs of even length paths originated
from a particular node $i\in V$. More specifically, $S$ contains the
nodes to which all even length paths from $i$ have positive signs,
and $\bar S$ contains the other set of nodes
    (note that $i$ may not be in $S$).

We argue that  (a) within $S$ and $\bar{S}$, all edges have  same
signs; and (b) all edges between $S$ and $\bar S$ have same signs.
Since $G$ contains both negative and positive edges, it must be
either balanced or anti-balanced.
% all edges have an opposite sign, which infers that $I$ and $\bar{I}$
% represent a
% positive / negative balanced partition.

For (a), assume to the contrary that
% Suppose that the above argument is not the case, i.e.,
there exist two directed edges $R_{ab} = a\rightarrow b$ and $R_{cd}
= c\rightarrow d$, which both reside in the same set, e.g., $S$ with
different signs. (The case for $\bar S$ is similar.)

We construct two even length paths from $i$ to $c$ and $i$ to $d$ as
follows.
\begin{align}
R_e(i,c)&=R_e(i,b)+R_e(b,c),\nonumber \\
R_e(i,d)&=R_e(i,a)+R_{ab}+R_e(b,c)+R_{cd}\nonumber
\end{align}
where $R_e(x,y)$ represents the constructed even length path from
node $x$ to node $y$.

% First of all, if two edges $e_{ab}$ and $e_{cd}$ with different signs reside in the
% same subgraph, i.e., $e_{ab},e_{cd}\in I$ or $e_{ab},e_{cd}\in
% \bar{I}$,
Since both $c,d \in S$, by construction, then $R_e(i,c)$ and
$R_e(i,d)$ have same signs
\begin{align}
sgn(R_e(i,c))&=sgn(R_e(i,d)).\label{eq:sgn1}
\end{align}
On the other hand, since $a$ and $b$ are in the same group as $c$
and $d$, $sgn(R_e(i,a))=sgn(R_e(i,b))$. Then, we have
\begin{align}
sgn(R_e(i,c))&=sgn(R_e(i,b))sgn(R_e(b,c)),\label{eq:sgn2} \\
sgn(R_e(i,d))&=sgn(R_e(i,a))sgn(R_{ab})sgn(R_e(b,c))sgn(R_{cd})\nonumber \\
&=- sgn(R_e(i,b))sgn(R_e(b,c)).\label{eq:sgn3}
\end{align}
Eq.(\ref{eq:sgn3}) comes from the assumption that $R_{ab}$ and
$R_{cd}$ have different signs. Eq.(\ref{eq:sgn1}) contradicts with
Eq.(\ref{eq:sgn2}) and Eq.(\ref{eq:sgn3}).

For (b),  assume that there exist two edges $R_{ab}$ and $R_{cd}$
with different signs between $S$ and $\bar S$. Still consider the
two even length paths $R_e(i,c)$ and $R_e(i,d)$ constructed before.
Since $c$ and $d$ are not in the same side, $R_e(i,c)$ and
$R_e(i,d)$ have opposite signs by the construction, i.e.,
\begin{align}
sgn(R_e(i,c))&=-sgn(R_e(i,d)).\label{eq:sgn11}
\end{align}
On the other hand, since $a$ and $b$ are in the different groups as
well, $sgn(R_e(i,a))=-sgn(R_e(i,b))$. Then, we have
\begin{align}
sgn(R_e(i,c)) %&=sgn(p_e(i,b))sgn(p_e(b,c)) %\nonumber \\
&=sgn(R_e(i,b))\cdot sgn(R_e(b,c)),\label{eq:sgn22} \\
sgn(R_e(i,d))&=sgn(R_e(i,a))sgn(R_{ab})sgn(R_e(b,c))sgn(R_{cd})\nonumber \\
&= sgn(R_e(i,b))\cdot sgn(R_e(b,c)).\label{eq:sgn33}
\end{align}
However, Eq.(\ref{eq:sgn11}) contradicts with Eq.(\ref{eq:sgn22})
and Eq.(\ref{eq:sgn33}). This completes the proof.
\end{proof}
}

The next lemma characterizes the limiting behavior of $P^t$ of
ergodic signed digraphs with all three balance structures. Given a
signed digraph $G=(V,E,A)$, let $\bar G =(V,E,\bar{A})$ corresponds
to its unsigned version ($\bar{A}_{ij} = |A_{ij}|$ for all $i,j\in
V$). When $\bar G$ is ergodic, a random walk on $\bar G$ has a
unique stationary distribution, denoted as $\pi$. That is, $\pi^T =
\pi^T \bar{P}$, where $\bar{P}=D^{-1}\bar{A}$ is
    the transition probability matrix for $\bar{G}$.
Henceforth, we always use $S,\bar{S}$ to denote the corresponding
    partition for either balanced graphs or anti-balanced graphs.
%{\color{red}For a square matrix $M=[M_{ij}]\in \mathbb{R}^{m\times
%m}$, the sequence $M^t$, $t\geq 0$, is said to be
%\emph{exponentially convergent}, if there exists a limiting matrix
%$M'$ and an exponential convergence rate $0<\rho< 1$ such that
%$\|M^t-M'\|_{\infty}\leq C \cdot \rho^t$ holds, for sufficiently
%large $t$, where $C$ is a constant and $\|\cdot\|_{\infty}$
%represents the induced infinity matrix norm, i.e.,}
We define the infinity norm of matrix $M\in \mathbb{R}^{m\times m}$
as: $\|M\|_{\infty}:=\max_{1\leq i\leq m}\sum_{j=1}^m|M_{ij}|$.

\begin{lemmaa}[]\label{thm:BalDig2}
Given an ergodic signed digraph $G=(V,E,A)$, let ${\bar G} = (V,E,
\bar{A})$ be the {\em unsigned} digraph.  When $G$ is balanced or
strictly unbalanced, $P^t$ converges, and when $G$ is anti-balanced,
the odd and even subsequences of $P_{}^t$ converge to opposite
matrices. %{\color{red}These matrix sequences are all exponentially
%convergent.}

\begin{eqnarray}
\mbox{Balanced $G$:} &  \mbox{ $\lim_{t\rightarrow
\infty}P_{}^t=\hat{\textbf{1}}_S\hat{\pi}_S^T$};
\nonumber \\
\mbox{Strictly unbalanced $G$:} & \mbox{ $\lim_{t\rightarrow
\infty}P_{}^t=\textbf{0}$};
\nonumber\\
\mbox{Anti-balanced $G$:} &  \mbox{ $\lim_{t\rightarrow
\infty}P_{}^{2t}=\hat{\textbf{1}}_S\hat{\pi}_S^T$}, \mbox{ }
\mbox{$\lim_{t\rightarrow
\infty}P_{}^{2t+1}=-\hat{\textbf{1}}_S\hat{\pi}_S^T$}.\nonumber
\end{eqnarray}
\end{lemmaa}
\begin{proof}
\noindent (1) When $G$ is balanced, the signed transition matrix
$P_{}$ can be written as $P_{}=\hat{I}_S\bar{P}\hat{I}_S$. Since
$\bar{G}$ is ergodic, % for its transition probability matrix $\bar{P}$,
we have $\lim_{t\rightarrow\infty}\bar{P}^t =
\textbf{1}\pi^T$. Thus,
\begin{align}
    \lim_{t\rightarrow \infty}{{P}}^t &= \lim_{t\rightarrow
    \infty}(\hat{I}_S\bar{P}\hat{I}_S)^t=\hat{\textbf{1}}_S \hat{\pi}_S^T,\nonumber
\end{align}
where we use simple facts $\hat{I}_S^2=I$,
    $\hat{I}_S\textbf{1}=\hat{\textbf{1}}_S$, and
    $\pi^T \hat{I}_S = \hat{\pi}_S^T$.

\noindent (2) When  $G$ is anti-balanced, we have ${P}=
-\hat{I}_S\bar{P}\hat{I}_S$. Thus,
\begin{align}
    \lim_{t\rightarrow \infty}{{P}}^{2t} &= \lim_{t\rightarrow
    \infty}(-\hat{I}_S\bar{P}\hat{I}_S)^{2t}=%  \lim_{t\rightarrow
  % \infty}(-1)^{2t}\hat{I}_S\bar{P}^{2t}\hat{I}_S\nonumber \\
  % &= \hat{I}_S \textbf{1}\pi^T \hat{I}_S=
    \hat{\textbf{1}}_S \hat{\pi}_S^T \nonumber \\
    \lim_{t\rightarrow \infty}{{P}}^{2t+1} &= \lim_{t\rightarrow
    \infty}(-\hat{I}_S\bar{P}\hat{I}_S)^{2t+1}= -\hat{\textbf{1}}_S \hat{\pi}_S^T \nonumber.
% \lim_{t\rightarrow
% \infty}(-1)^{2t+1}\hat{I}_S\bar{P}^{2t+1}\hat{I}_S\nonumber \\
%     &= -\hat{I}_S \textbf{1}\pi^T \hat{I}_S= -\hat{\textbf{1}}_S \hat{\pi}_S^T \nonumber
\end{align}

%{\color{black}In addition, the above proof also infers that the
%convergence time of $P^t$, when $G$ is ergodic balanced or strictly
%unbalanced, is identical to that of $\bar{P}^t$, which is proven to
%be exponentially fast (See Lemma~\ref{pro:exponentialofbarP} in
%Appendix~\ref{sec:exponentialofbarP}).}

{\color{black}\noindent (3) %See~\cite{VMreport} for the proof on strictly
%unbalanced digraphs.
%We defer the proof of result on strictly unbalanced digraph to our
%full technical report~\cite{VMreport} due to the limited space.
By Proposition~\ref{pro:UnbalancedDig}, % in
%Appendix~\ref{sec:proposition3},
given a signed strictly unbalanced
digraph $G$, there exist a pair of nodes $i$ and $j$, such that two
paths $R_1$ and $R_2$ from $i$ to $j$ have the same length $\ell(i)$
and opposite signs. Consider a random walk from $i$. Let $p_1$
(resp. $p_2$) be the probability that the walk exactly follows $R_1$
(resp. $R_2$) in the first $\ell(i)$ steps. Let $R^{\ell(i)}_{i,k}$
be the set of all paths from $i$ to $k$ with length $\ell(i)$. Then,
for a unit vector $e_i$ with $i$-th entry equal to $1$ and other
entries as $0$, we have
$$||e_i^TP^{\ell(i)}||_1 =\sum_{k \in V} \left|\sum_{R\in R^{\ell(i)}_{i,k}}
\mathrm{Prob}[R] sgn(R)\right| \leq 1-\min(p_1,p_2) = \rho_i.$$
%\noindent Notice that the $L_1$ norm always shrinks.

For any other node $i'$, there must exist a path $R'$ from
$i'\rightarrow i$, due to the ergodicity of $G$, thus two paths
$R'_1=R'+R_1$ and $R'_2=R'+R_2$ from $i'$ to $j$ have the same
length, but opposite signs. With similar arguments as that for node
$i$, $||e_{i'}^TP^{\ell(i')}||_1 \leq\rho_{i'}$ holds for any $i'\in
V$. Let $\rho = \max_i \rho_i <1$ and $\ell = \max_i \ell(i)$, we
conclude for any $i\in V$, $||e_{i}^TP^{\ell}||_1 \leq\rho$ holds.
Hence, when $t\geq T=2\ell$, the following inequality holds
$$||e_{i}^TP^{t}||_1 = ||e_{i}^TP^{\frac{t}{\ell}\ell}||_1 \leq  \rho^{\lfloor\frac{t}{\ell}\rfloor}\leq\rho^{\frac{t}{T}}.$$
Hence $\lim_{t\rightarrow \infty}\|P^t\|_{\infty}=0$, i.e.,
$\lim_{t\rightarrow \infty}P^t=\textbf{0}$.
%{\color{red}which implies $\|P^t\|_{\infty}\leq\rho^{\frac{t}{T}}$,
%i.e., $P^t$ exponentially converges to $\textbf{0}$. %,
%$\lim_{t\rightarrow \infty} P^t ={\bf
%  0}$.
%}
%$$||e_{i}^TP^{\ell}||_1 \leq\rho \mbox{ and } \lim_{t\rightarrow \infty} e_i^TP^t ={\bf 0}.$$
%This implies $\lim_{t\rightarrow \infty} P^t ={\bf
%  0}$.% since its null space is $\mathbf{R}^{|V|}$.
}
\end{proof}

The above lemma clearly shows different convergence behaviors of
$P^t$ for three
    types of graphs.
In particular, $P^t$ of anti-balanced graphs exhibits a bounded
oscillating behavior in the
    long term.

{\color{black}\noindent{\textbf{Case of weakly connected signed
digraphs.}}} Now, we consider a weakly connected signed digraph
$G=(V,E,A)$ with one ergodic sink component $G_Z$ with node set $Z$,
    which only has incoming edges from
    the rest of the signed digraph $G_X$ with node set $X=V\setminus Z$.
Then, the signed transition matrix $P$ has the following block form.
\begin{align}
    {P}&=\left[
            \begin{array}{c:c}
            {P}_X & {P}_Y \\ \hdashline
            \textbf{0} & {P}_Z
            \end{array}\right],\label{eq:blockPs}
\end{align}
%\begin{align}
%    {P}&=\begin{pmat}[{|}]
%    P_X & P_Y \cr\-
%    \textbf{0} & P_Z \cr
%    \end{pmat},\label{eq:blockPs}
%\end{align}
\noindent{where} ${P}_X$ and ${P}_Z$ are the block matrices for
component ${G}_X$ and ${G}_Z$, and ${P}_Y$ represent the one-way
connections from ${G}_X$ to ${G}_Z$.
%% The left bottom part is $\textbf{0}$, since
%% the sink component ${G}_Z$ only has incoming edges.
Then, the
$t$-step transition matrix ${P}^t$ can be expressed as
\begin{align}
    {P}^t&=\left[
            \begin{array}{c:c}
            P_X^{(t)} & P_Y^{(t)} \\ \hdashline
            \textbf{0} & P_Z^{(t)}
            \end{array}\right]
,\label{eq:blockPst2}
\end{align}
%\begin{align}
%    {P}^t&=\begin{pmat}[{|}]
%    P_X^{(t)} & P_Y^{(t)} \cr\-
%    \textbf{0} & P_Z^{(t)} \cr
%    \end{pmat},\label{eq:blockPst2}
%\end{align}
where $P_X^{(t)}=P_X^t$, $P_Z^{(t)}=P_Z^t$ and $P_Y^{(t)}=\sum_{i=0}^{t-1} P_X^iP_YP_Z^{t-1-i}$. %, \label{eq:blockPyt}
When $G_Z$ is balanced or anti-balanced, we use $S_Z,\bar{S}_Z$ to denote
    the partition of $Z$ defining its balance or anti-balance structure.
Then, we denote column vectors
\begin{align}
u_b&=(I_X-P_X)^{-1}P_Y\hat{\textbf{1}}_{Z,S_Z}, \label{eq:u_b}\\
\mbox{and   } u_u&=(I_X+P_X)^{-1}P_Y\hat{\textbf{1}}_{Z,S_Z}.\label{eq:u_u}
\end{align}
The reason that $I_X-P_X$ is invertible is because
    $\lim_{t\rightarrow\infty}P_X^t = \mathbf{0}$, which is in turn
    because there is a path from any node $i$ in $G_X$ to nodes in $Z$ (since $Z$
    is the single sink), and thus informally a random walk from $i$ eventually
    reaches and then stays in $G_Z$.
The same reason applies to $I_X+P_X$. Lemma~\ref{thm:convergWeakly2}
provides the formal proof of the fact
    $\lim_{t\rightarrow\infty}P_X^t = \mathbf{0}$.

Let $\pi_Z$ denote the stationary distribution of nodes in $G_Z$,
    and $\hat{\pi}_{Z,S_Z}$ is signed, with $\hat{\pi}_{Z,S_Z}(i)=\pi_Z(i)$ for $i\in S_Z$, and
    $\hat{\pi}_{Z,S_Z}(i)=-\pi_Z(i)$ for $i\in Z\setminus S_Z$.
Lemma~\ref{thm:convergWeakly2} discloses the convergence of $P^t$
given various balance structures of $G_Z$.

\begin{lemmaa}[]\label{thm:convergWeakly2}
For weakly connected signed digraph $G=(V,E,A)$ with one ergodic
    sink components, with signed transition matrix
    given in Eq.(\ref{eq:blockPst2}), we have
\begin{eqnarray}
\mbox{Balanced $G_Z$:} &  \mbox{ $ \lim_{\substack{t\rightarrow
\infty}}P^t=\left[\begin{array}{c:c}
     \textbf{0} & u_b\hat{\pi}_{Z,S_Z}^T \\ \hdashline
    \textbf{0} & \hat{\textbf{1}}_{Z,S_Z}\hat{\pi}_{Z,S_Z}^T
    \end{array}\right]$} \nonumber\\
\mbox{Strictly unbalanced $G_Z$:} &  \mbox{ $\lim_{\substack{t\rightarrow \infty}}P^t=\textbf{0}$}\nonumber\\
\mbox{Anti-balanced $G_Z$:} &  \mbox{ $ \lim_{t\rightarrow
\infty}P^{2t}=\left[
            \begin{array}{c:c}
            \textbf{0} & -u_u\hat{\pi}_{Z,S_Z}^T \\ \hdashline
            \textbf{0} & \hat{\textbf{1}}_{Z,S_Z}\hat{\pi}_{Z,S_Z}^T
            \end{array}\right]$, }%\nonumber \\
   \mbox{    $
\lim_{t\rightarrow \infty}P^{2t+1}=\left[
            \begin{array}{c:c}
            \textbf{0} & u_u\hat{\pi}_{Z,S_Z}^T \\ \hdashline
            \textbf{0} & -\hat{\textbf{1}}_{Z,S_Z}\hat{\pi}_{Z,S_Z}^T
            \end{array}\right]$
}\nonumber
\end{eqnarray}
%
%\[
%W = \left(
%\begin{array}{c:c}
%\begin{matrix}
%S_n & U_n  \\
%V_n & S_n^\mathrm{\,t}
%\end{matrix} & \text{\Large{\:0}} \\
%\hdashline %
%\text{\rule{0pt}{17pt}\Large{0}}  &  \textit{\Large{I}}
%\end{array}
%\right)
%\]
%When {\emph{$G_Z$ is strictly unbalanced}},
%$\lim_{\substack{t\rightarrow \infty}}P^t=\textbf{0}$. When
%{\emph{$G_Z$ is positively balanced}},
%\begin{align}
%%    {P}_s^{\infty}&=\left[
%    \lim_{\substack{t\rightarrow \infty}}P^t&=\left[
%            \begin{array}{c|c}
%            \textbf{0} & (I-P_X)^{-1}P_Y\hat{\textbf{1}}_{Z,S_Z}\hat{\pi}_{Z,S_Z}^T \\ \hline
%            \textbf{0} & \hat{\textbf{1}}_{Z,S_Z}\hat{\pi}_{Z,S_Z}^T
%            \end{array}\right].\nonumber
%\end{align}
%When {\emph{$G_Z$ is negatively balanced}},
%\begin{align}
%    \lim_{\substack{t\rightarrow \infty\\t \text{is odd}}}P^t&=\left[
%            \begin{array}{c|c}
%            \textbf{0} & (I+P_X)^{-1}P_Y\hat{\textbf{1}}_{Z,S_Z}\hat{\pi}_{Z,S_Z}^T \\ \hline
%            \textbf{0} & -\hat{\textbf{1}}_{Z,S_Z}\hat{\pi}_{Z,S_Z}^T
%            \end{array}\right];\nonumber \\
%   \lim_{\substack{t\rightarrow \infty\\t \text{is even}}}P^t&=\left[
%            \begin{array}{c|c}
%            \textbf{0} & -(I+P_X)^{-1}P_Y\hat{\textbf{1}}_{Z,S_Z}\hat{\pi}_{Z,S_Z}^T \\ \hline
%            \textbf{0} & \hat{\textbf{1}}_{Z,S_Z}\hat{\pi}_{Z,S_Z}^T
%            \end{array}\right];\nonumber
%\end{align}
\end{lemmaa}

{\color{black}
\begin{proof}
We discuss the convergence of $P_X^t$, $P_Z^t$, and $P_Y^{(t)}$ in
Eq.(\ref{eq:blockPst2}), respectively.

\noindent (1) We first prove that $P_X^t$ converges to $\textbf{0}$, i.e., %exponentially fast as
$\lim_{t\rightarrow \infty}P_X^t=\textbf{0}$.

Since $G_X$ does not contain sink components, any node $i\in X$ has
a path to component $G_Z$.
Let $R_{iZ}$ be the shortest path from $i$ to some node in $Z$,
    and $Prob[R_{iZ}]$ denote the probability that a random walk starting
    from $i$ takes the path $R_{iZ}$.
Hence we denote
$$p=\min_{i\in X}Prob[R_{iZ}], \mbox{ and  } m=\max_{i\in X}|R_{iZ}|,$$
which implies that starting from any node $i\in X$, after $m$ steps
of random walk, there is at least probability $p$ that it reaches
component $G_Z$. Hence, we have $\|P_X^{m}\|_{\infty}\leq
(1-p) < 1$. Let $T=2m$, then for any $t>T$, we have
$$\|P_X^t\|_{\infty}=\|P_X^{\frac{t}{m}m}\|_{\infty}\leq (1-p)^{\lfloor\frac{t}{m}\rfloor}
    \leq (1-p)^{\frac{t}{T}},$$
which implies $\lim_{t\rightarrow\infty}\|P_X^t\|_{\infty}=0$, i.e.,
$\lim_{t\rightarrow\infty}P_X^t=\textbf{0}$.

%Hence, we have $u=P_X^{m}\textbf{1}$ with
%$\|u\|_{\infty} \leq 1-p<1$. Let $T=2m$, then for any $t>T$, we have
%$$\|P_X^t\textbf{1}\|_{\infty}=\|P_X^{\frac{t}{m}m}\textbf{1}\|_{\infty}\leq (1-p)^{\lfloor\frac{t}{m}\rfloor}\leq(1-p)^{\frac{t}{T}},$$
%which implies
%$\lim_{t\rightarrow\infty}\|P_X^t\textbf{1}\|_{\infty}=0$, thus
%$\lim_{t\rightarrow\infty}P_X^t=\textbf{0}$ converges exponentially.

\noindent (2) For subgraph $G_Z$, Lemma~\ref{thm:BalDig2} directly
yields
\begin{align}
\lim_{t\rightarrow\infty}P_Z^{t}&=\left\{
            \begin{array}{l l}
            \textbf{0}, & \text{Strictly unbalanced $G_Z$};\\
            \textbf{1}_{Z,S_Z}\pi_{Z,S_Z}^T, & \text{Balanced $G_Z$};\\
            \textbf{1}_{Z,S_Z}\pi_{Z,S_Z}^T, & \text{Anti-balanced $G_Z$, even $t$};\\
            -\textbf{1}_{Z,S_Z}\pi_{Z,S_Z}^T, & \text{Anti-balanced $G_Z$, odd $t$}.\\
            \end{array}\right.\label{eq:PZinfty}
\end{align}

\noindent (3) %{\color{red} I currently wrote down detailed proof
%here for your ease of reading, where we can discuss how to simplify
%this part.}
Below, we focus on proving the results on
$\lim_{t\rightarrow\infty} P_Y^{(t)}$ using
Proposition~\ref{thm:matrixlimit} %-(iii)
in Appendix~\ref{sec:Matrixform}.

\noindent{When \emph{$G_Z$ is strictly unbalanced,}} %From
%Lemma~\ref{thm:BalDig2} and (1) in this proof,
%$\lim_{t\rightarrow\infty} P_X^{t}=\textbf{0}$ and
%$\lim_{t\rightarrow\infty} P_Z^{t}=\textbf{0}$ hold, thus
%$\lim_{t\rightarrow\infty}P_Y^{(t)}=\textbf{0}$.
from Lemma~\ref{thm:BalDig2} and (1) in this proof,
$\lim_{t\rightarrow\infty} P_X^{t}=\textbf{0}$ and
$\lim_{t\rightarrow\infty} P_Z^{t}=\textbf{0}$ hold, thus by
    Proposition~\ref{thm:matrixlimit} %-(iii)
    in Appendix~\ref{sec:Matrixform}
$\lim_{t\rightarrow\infty}P_Y^{(t)}=\textbf{0}$.

\noindent{When \emph{$G_Z$ is balanced,}} {\color{black}
Lemma~\ref{thm:BalDig2} and Proposition~\ref{pro:exponentialofbarP}
in Appendix~\ref{sec:exponentialofbarP} directly yield
%$\lim_{t\rightarrow\infty}
%P_Z^t-\textbf{1}_{Z,S_Z}\pi^T_{Z,S_Z}=\textbf{0}$ holds, which implies\footnote{It holds due to the fact
%$(P_Z-\textbf{1}_{Z,S_Z}\pi^T_{Z,S_Z})^t=P_Z^t-\textbf{1}_{Z,S_Z}\pi^T_{Z,S_Z}$
%($t\geq 1$)}
$(P_Z-\textbf{1}_{Z,S_Z}\pi^T_{Z,S_Z})^t=P_Z^t-\textbf{1}_{Z,S_Z}\pi^T_{Z,S_Z}$
for any integer $t>0$, and $\lim_{t\rightarrow\infty}
(P_Z-\textbf{1}_{Z,S_Z}\pi^T_{Z,S_Z})^t=\textbf{0}$, thus}
\begin{align}
& \lim_{t\rightarrow\infty}P_Y^{(t)}=
\lim_{t\rightarrow\infty}\sum_{i=0}^{t-1}P_X^iP_Y(P_Z^{t-1-i}-\textbf{1}_{Z,S_Z}\pi^T_{Z,S_Z}+\textbf{1}_{Z,S_Z}\pi^T_{Z,S_Z})\nonumber
\\
&
=\lim_{t\rightarrow\infty}\sum_{i=0}^{t-1}P_X^iP_Y(P_Z-\textbf{1}_{Z,S_Z}\pi^T_{Z,S_Z})^{t-1-i}+\lim_{t\rightarrow\infty}\sum_{i=0}^{t-2}P_X^iP_Y\textbf{1}_{Z,S_Z}\pi^T_{Z,S_Z}\nonumber
\\
&
=(I_X-P_X)^{-1}P_Y\textbf{1}_{Z,S_Z}\pi^T_{Z,S_Z}=u_b\pi^T_{Z,S_Z},\nonumber
\end{align}
where the first term in the second line being $\textbf{0}$
    is due to
    Proposition~\ref{thm:matrixlimit}~(ii) in Appendix~\ref{sec:Matrixform}.

%\begin{align}
%& \lim_{t\rightarrow\infty}P_Y^{(t)}=
%\lim_{t\rightarrow\infty}\sum_{i=0}^{t-1}P_X^iP_Y(P_Z^{t-1-i}-\textbf{1}_{Z,S_Z}\pi^T_{Z,S_Z}+\textbf{1}_{Z,S_Z}\pi^T_{Z,S_Z})\nonumber
%\\
%&
%=\lim_{t\rightarrow\infty}\sum_{i=0}^{t-1}P_X^iP_Y(P_Z-\textbf{1}_{Z,S_Z}\pi^T_{Z,S_Z})^{t-1-i}\nonumber\\
%&-\lim_{t\rightarrow\infty}P_X^{t-1}P_Y\textbf{1}_{Z,S_Z}\pi^T_{Z,S_Z}+\lim_{t\rightarrow\infty}\sum_{i=0}^{t-1}P_X^iP_Y\textbf{1}_{Z,S_Z}\pi^T_{Z,S_Z}\nonumber
%\\
%&
%=(I_X-P_X)^{-1}P_Y\textbf{1}_{Z,S_Z}\pi^T_{Z,S_Z}=u_b\pi^T_{Z,S_Z}\nonumber
%\end{align}

\noindent{When \emph{$G_Z$ is anti-balanced,}} %From
%Lemma~\ref{thm:BalDig2} $\lim_{t\rightarrow\infty}
%P_Z^{t}-(-1)^t\textbf{1}_{Z,S_Z}\pi^T_{Z,S_Z}=\textbf{0}$, which
%implies
{\color{black}{applying Lemma~\ref{thm:BalDig2} and
Proposition~\ref{pro:exponentialofbarP} in
Appendix~\ref{sec:exponentialofbarP}, we have for any integer $t>0$,
$(P_Z+\textbf{1}_{Z,S_Z}\pi^T_{Z,S_Z})^t=P_Z^t-(-1)^t\textbf{1}_{Z,S_Z}\pi^T_{Z,S_Z}$,
    and $\lim_{t\rightarrow\infty}
    (P_Z+\textbf{1}_{Z,S_Z}\pi^T_{Z,S_Z})^t=\textbf{0}$ hold true,
    thus}}
\begin{align}
& \lim_{t\rightarrow\infty}P_Y^{(t)}=
\lim_{t\rightarrow\infty}\sum_{i=0}^{t-1}P_X^iP_Y(P_Z^{t-1-i}-(-1)^{t-1-i}(\textbf{1}_{Z,S_Z}\pi^T_{Z,S_Z}-\textbf{1}_{Z,S_Z}\pi^T_{Z,S_Z}))\nonumber
\\
&
=\lim_{t\rightarrow\infty}\sum_{i=0}^{t-1}P_X^iP_Y(P_Z+\textbf{1}_{Z,S_Z}\pi^T_{Z,S_Z})^{t-1-i}+\lim_{t\rightarrow\infty}\sum_{i=0}^{t-2}(-1)^{t-1-i}P_X^iP_Y\textbf{1}_{Z,S_Z}\pi^T_{Z,S_Z}\nonumber
\\
&
=(-1)^{t-1}\lim_{t\rightarrow\infty}\sum_{i=0}^{t-2}(-P_X)^iP_Y\textbf{1}_{Z,S_Z}\pi^T_{Z,S_Z}=(-1)^{t-1}(I_X+P_X)^{-1}P_Y\textbf{1}_{Z,S_Z}\pi^T_{Z,S_Z}\nonumber
\\
&=(-1)^{t-1}u_u\pi^T_{Z,S_Z}.\nonumber
\end{align}
Hence, we have for anti-balanced $G_Z$: $\lim_{t\rightarrow\infty}P_Y^{(2t)}
    =-u_u\hat{\pi}_{Z,S_Z}^T$, and
    $\lim_{t\rightarrow\infty}P_Y^{(2t+1)} =u_u\hat{\pi}_{Z,S_Z}^T$.
\end{proof}
}

{\color{black} \noindent\textbf{{Multiple sink components and
disconnected signed digraphs.}} When there exist $m>1$ ergodic sink
components, i.e., ${G}_{Z1}, {G}_{Z2}, \cdots, {G}_{Zm}$, the rest
of the graph ${G}$ is considered as ${G}_X$. Then the signed
transition matrix $P$ and $P^t$ can be written as
\begin{align}
\small P&=\left[
            \begin{array}{c:c:c:c}
            {P}_X & {P}_{Y1} & \cdots & {P}_{Ym} \\ \hdashline
            \textbf{0} & {P}_{Z1} & \textbf{0} & \textbf{0} \\ \hdashline
            \textbf{0} & \textbf{0} & \ddots & \textbf{0} \\ \hdashline
            \textbf{0} & \textbf{0} & \textbf{0} & {P}_{Zm}
            \end{array}\right],
P^t=\left[
            \begin{array}{c:c:c:c}
            {P}_X^t & {P}_{Y1}^{(t)} & \cdots & {P}_{Ym}^{(t)} \\ \hdashline
            \textbf{0} & {P}_{Z1}^t & \textbf{0} & \textbf{0} \\ \hdashline
            \textbf{0} & \textbf{0} & \ddots & \textbf{0} \\ \hdashline
            \textbf{0} & \textbf{0} & \textbf{0} & {P}_{Zm}^t
            \end{array}\right],\label{eq:Pmatmultisink}
\end{align}
%\[
%%  \pmatset{1}{0.12pt}
%%  \pmatset{0}{2\pmatget{1}}
%%  \pmatset{2}{20pt}
%%  \pmatset{3}{8pt}
%%   \pmatset{5}{15pt}
%      \pmatset{6}{2pt}
%    P=\begin{pmat}[{|||}]
%    {P}_X & {P}_{Y1} & \cdots & {P}_{Ym} \cr\-
%    \textbf{0} & {P}_{Z1} & \textbf{0} & \textbf{0} \cr\-
%    \textbf{0} & \textbf{0} & \ddots & \textbf{0} \cr\-
%    \textbf{0} & \textbf{0} & \textbf{0} & {P}_{Zm} \cr
%  \end{pmat},
%  P^t=\begin{pmat}[{|||}]
%    {P}_X^t & {P}_{Y1}^{(t)} & \cdots & {P}_{Ym}^{(t)} \cr\-
%    \textbf{0} & {P}_{Z1}^t & \textbf{0} & \textbf{0} \cr\-
%    \textbf{0} & \textbf{0} & \ddots & \textbf{0} \cr\-
%    \textbf{0} & \textbf{0} & \textbf{0} & {P}_{Zm}^t \cr
%  \end{pmat}\label{eq:Pmatmultisink}
%\]
%\[
%%  \pmatset{1}{0.72pt}
%%  \pmatset{0}{2\pmatget{1}}
%%  \pmatset{0}{2\pmatget{1}}
%  \pmatset{5}{5pt}
%  \pmatset{6}{5pt}
%    P=\begin{pmat}[{|||}]
%    {P}_X & {P}_{Y1} & \cdots & {P}_{Yk} \cr\-
%    \textbf{0} & {P}_{Z1} & \textbf{0} & \textbf{0} \cr\-
%    \textbf{0} & \textbf{0} & \ddots & \textbf{0} \cr\-
%    \textbf{0} & \textbf{0} & \textbf{0} & {P}_{Zk} \cr
%  \end{pmat},
%  P^t=\begin{pmat}[{|||}]
%    {P}_X^t & {P}_{Y1}^{(t)} & \cdots & {P}_{Yk}^{(t)} \cr\-
%    \textbf{0} & {P}_{Z1}^t & \textbf{0} & \textbf{0} \cr\-
%    \textbf{0} & \textbf{0} & \ddots & \textbf{0} \cr\-
%    \textbf{0} & \textbf{0} & \textbf{0} & {P}_{Zk}^t \cr
%  \end{pmat}
%\]
where ${P}_{Yi}^{(t)}=\sum_{j=0}^{t-1} P_X^jP_{Yi}P_{Zi}^{t-1-j}$,
$1\leq i\leq m$. Hence, each sink ergodic component ${P}_{Zi}$ along
with ${P}_X$ and $P_{Yi}$ independently follows Lemma~\ref{thm:convergWeakly2}.
For disconnected signed digraph, with $m\geq1$ ergodic or weakly
connected components, each of which satisfies
Lemma~\ref{thm:BalDig2} or Lemma~\ref{thm:convergWeakly2},
respectively. For brevity, we omit the details here.
%% For each balanced sink component $G_{Zi}$,
%% denote $u_{bi}=(I_X-P_X)^{-1}P_{Yi}\hat{\textbf{1}}_{Zi,S_{Zi}}$,
%% where $P_{Yi}$ is the block sub-matrix in $P$ across $G_X$ and
%% $G_{Zi}$.

%Weakly connected digraphs with multiple ergodic sinks or
%disconnected
%    digraphs can be similarly analyzed (see~\cite{VMreport}).
}

\subsection{Long-term dynamics}\label{sec:LTVotermodelDynamics}

Based on the structural balance classification and
    the convergence of signed transition matrix discussed above,
    we are ready now to analyze the long-term dynamics of the voter
    model on signed digraphs.
Formally, we are interested in characterizing $x_t$ with
$t\rightarrow\infty$,
%% \footnote{Technically, $\{x_t\}$ either converges, or
%%   its odd sequence and even sequence converge respectively. To simply
%%   the notation, we write it as $\limt x_t$ without distinguish the two
%%   cases.
%% }
i.e.,
\begin{align}
x&=\lim_{t\rightarrow \infty}x_t=\lim_{t\rightarrow
\infty}({P}^tx_0+(\sum_{i=0}^{t-1}{P}^i)g^-).\label{eq:signedlim2}
\end{align}

If the even and odd subsequences of $x_t$ converge separately, we
denote $x_e = \lim_{t\rightarrow \infty}x_{2t}, x_o =
\lim_{t\rightarrow \infty}x_{2t+1}$.

Before presenting the results on long-term dynamics of voter model,
we first introduce the following useful lemma connecting a signed
digraph $G$ with another graph $G'$ where all edge signs in $G$ are
negated.

\begin{lemmaa}[]\label{lem:GGp}
Given a signed digraph $G=(V,E,A)$, let $G'=(V,E,-A)$ be a signed
digraph with all edge signs negated from $G$. Then, for any initial
color distribution $x_0$, at any $2t$ steps ($t>0$), the color
distributions $x_{2t}(G)$ on $G$ and $x_{2t}(G')$ on $G'$ are
identical.
\end{lemmaa}
% Wei deleted due to space % Yanhua resumed the following for technical report 11-09-11
{\color{black} \begin{proof}
 Let $P'=-P$ denote the signed transition matrix of $G'$, and denote
 the vector $g^-={D^{-1}A^-}\textbf{1}$ and
 $g'^-={D^{-1}(-A)^-}\textbf{1}={D^{-1}A^+}\textbf{1}$. Thus
 $g'^-=\textbf{1}-g^-$. By Eq.(\ref{eq:signed2}), after two steps, we
 have
 \begin{align}
 x_2(G') & = P'^2 x_0 + P' g'^- + g'^- = P^2 x_0 - P
 (\textbf{1}-g^-) + \textbf{1} - g^- = P^2 x_0 + P g^- + g^- = x_2(G),\nonumber
 \end{align}
 where the last equality uses facts $\textbf{1}=D^{-1}\bar{A}\textbf{1}$ and
     $P=D^{-1}A$.
 Since the lemma holds for two steps, then clearly it holds for all even steps.
 \end{proof}}

Next theorem discusses the case of ergodic signed digraphs.

\begin{theoremm}[]\label{thm:ltdyn}
Let $G=(V,E,A)$ be an ergodic signed digraph, we have
\begin{eqnarray}
\mbox{Balanced $G$:} &  \mbox{ $x
=\hat{\textbf{1}}_S\hat{\pi}_S^T(x_0-\frac{1}{2}\textbf{1})+\frac{1}{2}\textbf{1}$}
\label{eq:bx}\\
\mbox{Strictly unbalanced $G$:} &  \mbox{ $x =\frac {1}{2} {\textbf 1}$}\label{eq:stubx}\\
\mbox{Anti-balanced $G$:} &  \mbox{ $x_{e}=\hat{\textbf{1}}_S\hat{\pi}_S^T(x_0-\frac{1}{2}\textbf{1})+\frac{1}{2}\textbf{1} $}\label{eq:evenx}\\
 &  \mbox{ $x_{o}
=-\hat{\textbf{1}}_S\hat{\pi}_S^T(x_0-\frac{1}{2}\textbf{1})+\frac{1}{2}\textbf{1}$
}\label{eq:oddx}
\end{eqnarray}
%
%\begin{itemize}
%  \item \textbf{\emph{$G$ is Positive Balanced}}
%    \begin{align}
%%        x &=\hat{\textbf{1}}_S\hat{\pi}_S^Tx_0+(I-{P}+\hat{\textbf{1}}_S\hat{\pi}_S^T)^{-1}g^-.
%%        \label{eq:bx1}\\
%        x &=\hat{\textbf{1}}_S\hat{\pi}_S^Tx_0-\frac{1}{2}\hat{\textbf{1}}_S\hat{\pi}_S^T\textbf{1}+\frac{1}{2}\textbf{1} \label{eq:bx}
%    \end{align}
%  \item \textbf{\emph{$G$ is Negative balanced}}
%  ${x_t}$ is a bounded oscillating vector sequence.
% In particular, we have,
%    \begin{align}
%        \lim_{t\rightarrow \infty} x_{2t} %&=
%%        \hat{\textbf{1}}_S\hat{\pi}_S^Tx_0+ (I-P-\hat{\textbf{1}}_S\hat{\pi}_S^T)^{-1}
%%        g^--\hat{\textbf{1}}_S\hat{\pi}_S^Tg^- \label{eq:evenx1}
%%        \\
%        &=\hat{\textbf{1}}_S\hat{\pi}_S^Tx_0-\frac{1}{2}\hat{\textbf{1}}_S\hat{\pi}_S^T\textbf{1}+\frac{1}{2}\textbf{1} \label{eq:evenx}
%        \\
%        \lim_{t\rightarrow \infty} x_{2t+1} %&=
% %       -\hat{\textbf{1}}_S\hat{\pi}_S^Tx_0+
%%        (I-P-\hat{\textbf{1}}_S\hat{\pi}_S^T)^{-1}g^-\label{eq:oddx1}\\
%        &=-\hat{\textbf{1}}_S\hat{\pi}_S^Tx_0+\frac{1}{2}\hat{\textbf{1}}_S\hat{\pi}_S^T\textbf{1}+\frac{1}{2}\textbf{1}.\label{eq:oddx}
%    \end{align}
%    % Hence, the overall sequence $x_t$ oscillates bounded by Eq.(\ref{eq:oddx}) and
%    % Eq.(\ref{eq:evenx}) at odd and even steps.
%  \item \textbf{\emph{$G$ is Strictly unbalanced}}
%    \begin{align}
%        x &=(I-{P})^{-1}g^- = \frac {1}{2} {\textbf 1}.
%        \label{eq:stubx}
%    \end{align}
%\end{itemize}
\end{theoremm}
\begin{proof}

We discuss the limit in Eq.~(\ref{eq:signedlim2}) for three possible
balance structures of $G$.

\noindent\textbf{Balanced digraphs.} %By Lemma~\ref{thm:BalDig2},
%$\lim_{t\rightarrow \infty}{P}^t=\hat{\textbf{1}}_S\hat{\pi}_S^T$
%holds. Using the facts $\hat{\pi}_S^TP=\hat{\pi}_S^T$,
%$P\hat{\textbf{1}}_S=\hat{\textbf{1}}_S$, and
%$(\hat{\textbf{1}}_S\hat{\pi}_S^T)^2=\hat{\textbf{1}}_S\hat{\pi}_S^T$,
{\color{black}{From Lemma~\ref{thm:BalDig2} and
Proposition~\ref{pro:exponentialofbarP} in
Appendix~\ref{sec:exponentialofbarP},}} it is easy to prove
$P^m-\hat{\textbf{1}}_S\hat{\pi}_S^T=(P-\hat{\textbf{1}}_S\hat{\pi}_S^T)^m$
for any integer $m>0$, which yields the following result on the
second part in Eq.~(\ref{eq:signedlim2}).
\begin{align}
\lim_{t\rightarrow \infty}&\sum_{i=0}^{t-1}{P}^ig^-=
    (I-{P}+\hat{\textbf{1}}_S\hat{\pi}_S^T)^{-1}g^-+
    \lim_{t\rightarrow \infty}\sum_{i=1}^{t-1}{\bf 1}_S\hat{\pi}_S^Tg^-\label{eq:balancedresults2}\\
&=(I-{P}+\hat{\textbf{1}}_S\hat{\pi}_S^T)^{-1}g^-=\frac{1}{2}\textbf{1}-\frac{1}{2}\hat{\textbf{1}}_S\hat{\pi}_S^T\textbf{1},\label{eq:balancedresults3}
\end{align}
where the last term of Eq.(\ref{eq:balancedresults2}) is canceled
out due to the digraph flow circulation
law~\cite{F-Green-cheeger,WAW2010Yanhuashort}, i.e.,
\begin{align}
\hat{\pi}_S^Tg^-& = \hat{\pi}_S^TD^{-1}A^-\textbf{1}=%0\nonumber
\sum_{i\in S}{\pi}(i)\sum_{j\in \bar{S}}{\bar{P}}_{ij}-\sum_{i\in
\bar{S}}{\pi}(i)\sum_{j\in S}{\bar{P}}_{ij}=0.\nonumber
\end{align}
The last equality in Eq.(\ref{eq:balancedresults3}) holds because%
%where ${\bar{P}}_{ij}$ is the transition probability matrix of
%the unsigned digraph ${\bar{G}}=(V,E)$. The last equality comes
%from the network flow circulation rule of digraphs. Now, consider
%the first term in Eq.~(\ref{eqn:positivebalancedsum}). Using the
%property of the fundamental matrix~\cite{WAW2010Yanhua}, we obtain
%\begin{align}
%\lim_{t\rightarrow \infty}(\sum_{i=0}^{t-1}{P}^i)g^- &=\hat{I}_S(I-{\bar
%P}+\textbf{1}\pi^T)^{-1}\hat{I}_Sg^-\nonumber
%\\
%&=(I-{P}+\hat{\textbf{1}}_S\hat{\pi}_S^T)^{-1}g^-=\frac{1}{2}\textbf{1}-\frac{1}{2}\hat{\textbf{1}}_S\hat{\pi}_S^T\textbf{1},\label{eq:balancedresults3}
%\end{align}
%The last equation is obtained by the following fact.
\begin{align}
&\frac{1}{2}(I-P+\hat{\textbf{1}}_S\hat{\pi}_S^T)(\textbf{1}-\hat{\textbf{1}}_S\hat{\pi}_S^T\textbf{1})-g^-=0.\nonumber%\\
%&=\frac{1}{2}(\textbf{1}-(\textbf{1}-2g^-)+\hat{\textbf{1}}_S\hat{\pi}_S^T\textbf{1}-(\hat{\textbf{1}}_S\hat{\pi}_S^T\textbf{1}-\hat{\textbf{1}}_S\hat{\pi}_S^T\textbf{1}+\hat{\textbf{1}}_S\hat{\pi}_S^T\textbf{1}))-g^-\nonumber
%\\
%&=\frac{1}{2}(2g^-)-g^-=0.\nonumber
\end{align}
%which infers
%\begin{align}
%&(I-P+\hat{\textbf{1}}_S\hat{\pi}_S^T)g^-=\frac{1}{2}\textbf{1}-\frac{1}{2}\hat{\textbf{1}}_S\hat{\pi}_S^T\textbf{1}\label{eq:no-fundamental}
%\end{align}
%Plugging Eq.(\ref{eq:no-fundamental}) in
%Eq.(\ref{eq:balancedresults3}) yields Eq.(\ref{eq:bx}).
Eq.(\ref{eq:bx}) is obtained by combining Eq.(\ref{eq:balancedresults3}) with
    Lemma~\ref{thm:BalDig2}.

\noindent\textbf{Anti-balanced Digraphs.}
%% by Theorem~\ref{thm:BalDig2}, we
%% have $\lim_{t\rightarrow \infty}{P}^t=\hat{\textbf{1}}_S\hat{\pi}_S$ when $t$
%% is even; and $\lim_{t\rightarrow \infty}{P}^t=-\hat{\textbf{1}}_S\hat{\pi}_S$
%% when $t$ is odd, where $\hat{\pi}_S=\hat{I}_S{\pi}$. \emph{To accomplish the
%% proof, we first introduce the following proposition},
Lemma~\ref{lem:GGp} directly yields Eq.(\ref{eq:evenx}). The odd
step influence distribution sequence is obtained by
\begin{align}
&x_o= Px_e+g^-=
-\textbf{1}_{S}\hat{\pi}_S^T(x_0-\frac{1}{2}\textbf{1})+\frac{1}{2}\textbf{1}.\nonumber
\end{align}
%\begin{align}
%&\lim_{t\rightarrow\infty}x_{2t+1}=
%P\lim_{t\rightarrow\infty}x_{2t}+g^-=%P(\textbf{1}_{S}\hat{\pi}_S^T(x_0-\frac{1}{2}\textbf{1})+\frac{1}{2}\textbf{1})+g^-\nonumber \\
%%&=-\textbf{1}_{S}\hat{\pi}_S^T(x_0-\frac{1}{2}\textbf{1})+\frac{\textbf{1}-2g^-}{2}+g^-=
%-\textbf{1}_{S}\hat{\pi}_S^T(x_0-\frac{1}{2}\textbf{1})+\frac{1}{2}\textbf{1}.\nonumber
%\end{align}

\noindent\textbf{Strictly unbalanced digraphs.} From
Theorem~\ref{thm:BalDig2}, $\lim_{t\rightarrow
\infty}{P}^t=\textbf{0}$ holds and thus we have
\begin{align}
\lim_{t\rightarrow \infty}\sum_{i=0}^{t-1}{P}^ig^-&=(I-{P})^{-1}g^-
= (D-A)^{-1}A^-{\textbf 1} = \frac 12 {\textbf 1}.%\nonumber
\label{eq:strictlyunb}
\end{align}
%%---Yanhua-Removed-due to limited space-below-
The last equality comes from the facts $(D-A)\textbf{1} = 2A^- \textbf{1}$.
%%---Yanhua-Removed-due to limited space-above-
%% Eq.(\ref{eq:signedlim2}) implies that on strictly
%% unbalanced digraphs, when $t$ is sufficiently large, the probability
%% of each user $i$ to be active converges a stationary value, which is
%% invariant and independent to the initial
%% state/distribution. Therefore, the influence of the active user is
%% deterministic no matter how the initial seeds are chosen.
\end{proof}

Theorem~\ref{thm:ltdyn} has several implications. First of all, for
strictly unbalanced digraphs, each node has equal steady state
probability of being black or white, and it
    is not determined by the initial distribution $x_0$.
Secondly, anti-balanced digraphs has the same steady state
distribution as
    the corresponding balanced graph for even steps, and for odd steps,
    the distribution oscillates to the opposite ($x_o = \textbf{1}-x_e$).
Moreover, Eq.(\ref{eq:bx}) can also be intuitively explained from the random
    walk interpretation of the voter model. %(see~\cite{VMreport}).
% Wei deleted due to space % Yanhua resumed the following for techreport.
{\color{black} In particular, starting from node $i$, if we perform
a random walk
 for an infinite number of steps, the probability that
    the random walk stops at $j$
 is given by the
     stationary distribution $\pi(j)$.
 For balanced graphs, if $i$ and $j$ are from the same component (either
     $S$ or $\bar{S}$), then the random walk must pass an even number
     of negative edges, so $i$ takes the same color as $j$; if
     $i$ and $j$ are from opposite components, then the walk passes an
     odd number of negative edges and $i$ takes the opposite of $j$'s color.
 Thus, the steady distribution of $i\in S$ being white is given by
     $\pi_S^T x_{0S} + \pi_{\bar{S}}^T (\textbf{1}_{\bar{S}}-x_{0\bar{S}})$, and
     the case of $i\in \bar{S}$ is symmetric. Some algebra manipulations can lead us to Eq.(\ref{eq:bx}).}

% Especially, when the graph is an unsigned undirected graph, which is
% a special case of positively balanced digraph, negative link
% receptive vector $g^-=0$, and $x$ converges to
% \begin{align}
%         x &=\textbf{1}{\pi}^Tx_0 =
%         ={\frac{1}{\sum_i({d}_i)}}\textbf{1}{d}^Tx_0.
%         \label{eq:unsignedx}
% \end{align}
% Eq.(\ref{eq:unsignedx}) is identical to the result in~\cite{} for
% non-negative undirected graphs.
% {\bf Yajun: we may place it into appendix.}

For a balanced ergodic digraph $G$ with partition $S,\bar{S}$, it is
    easy to check that it has the following two
    equilibrium states: in one state all nodes in $S$ are white while
    all nodes in $\bar{S}$ are black; and in the other state all nodes
    in $S$ are black while all nodes in $\bar{S}$ are white.
We call these two states the {\em polarized states}.
Using random walk interpretation, we show in the following theorem that
    with probability $1$, the voter model dynamic converges to
    one of the above two equilibrium states.%, in $O(n^3\log n)$ time.

\begin{theoremm}\label{thm:converge}
Given an ergodic signed digraph $G=(V,E,A)$, if $G$ is balanced with
partition $S,\bar{S}$,
    the voter model dynamic converges
    to one of the polarized states with probability $1$, and
    the probability of nodes in $S$ being white is
    $\hat{\pi}_S^T(x_0-\frac{1}{2}\textbf{1})+\frac{1}{2}$.
Similarly, if $G$ is anti-balanced, with probability $1$
    the voter model dynamic oscillates between
    the two polarized states eventually, and
    the probability of nodes in $S$ being white at even steps is
    $\hat{\pi}_S^T(x_0-\frac{1}{2}\textbf{1})+\frac{1}{2}$.
\end{theoremm}
\begin{proof}
%%% We follow the result of Theorem 4 of~\cite{even2007note} and its proof.
%(Sketch)
%This can be proven by knowing that random walks starting from any two nodes in
%    an ergodic graph eventually meet, and by checking  the parity of the number
%    of negative edges on the two walks before they meet.

% Wei deleted due to space % Yanhua resumed for techreport
{\color{black} Consider a balanced ergodic digraph $G$ with
partition $S,\bar{S}$.
 By ergodicity, given any two nodes $i$ and $j$,
  with probability $1$ the random walks starting from $i$ and $j$
  will meet eventually.
 If $i$ and $j$ are both in $S$, when the two walks meet at some
     node $u$, they both pass either an even number of
     negative edges (if $u\in S$) or an odd number of negative edges
     (if $u\in \bar{S}$).
 Therefore, $i$ and $j$ must be in the same color with probability
     $1$.
 If $i$ and $j$ are from different components $S$ and $\bar{S}$, a similar
     argument shows that they will have the opposite color with
     probability $1$.
 Therefore the final state
     is one of the two polarized states.
 The probability of nodes in $S$ being white is simply given by
     Theorem~\ref{thm:ltdyn}, Eq.(\ref{eq:bx}).
 The case of anti-balanced ergodic digraphs can be argued in
     a similar way.}
\end{proof}

Theorem~\ref{thm:ltWeakly} below introduces the long-term dynamics of the
weakly connected signed digraphs. We consider weakly connected $G$
with a single sink ergodic component $G_Z$, and use the same
notations as in Section~\ref{sec:balance}.

\begin{theoremm}[]\label{thm:ltWeakly}
Let $G=(V,E,A)$ be a weakly connected signed digraph with a single
sink component
    $G_Z$ and a non-sink component $G_X$.
The long-term white color distribution vector $x$ is expressed in
two parts:
\begin{align}
x^T = \lim_{t \rightarrow \infty} x_t^T = [ x_{XY}^T,
x_Z^T].\nonumber
\end{align}
where $x_Z$ is the limit of $x_{tZ}$ on $G_Z$
  with initial distribution $x_{0Z}$ and is given as in
  Theorem~\ref{thm:ltdyn}, and
    vector $x_{XY}$ is given below
    with respect to the balance structure of $G_Z$:
\begin{eqnarray}
\mbox{Balanced $G_Z$:} &  \mbox{
$x_{XY}=\frac{1}{2}\textbf{1}_X+u_b\hat{\pi}_{Z,S_Z}^T(x_{0Z}-\frac{1}{2}\textbf{1}_Z)$}
\nonumber \\
\mbox{Strictly unbalanced $G_Z$:} & \mbox{
$x_{XY}=\frac{1}{2}\textbf{1}_X$}
\nonumber\\
\mbox{Anti-balanced $G_Z$, even $t$:} &  \mbox{ $x_{XY,e}=\frac{1}{2}\textbf{1}_X-u_u\hat{\pi}_{Z,S_Z}^T(x_{0Z}-\frac{1}{2}\textbf{1}_Z)$}\nonumber\\
\mbox{Anti-balanced $G_Z$, odd $t$:} &  \mbox{
$x_{XY,o}=\frac{1}{2}\textbf{1}_X+u_u\hat{\pi}_{Z,S_Z}^T(x_{0Z}-\frac{1}{2}\textbf{1}_Z)$
},\nonumber
\end{eqnarray}
where $u_b$ and $u_u$ are defined in Eq.(\ref{eq:u_b}) and
Eq.(\ref{eq:u_u}).
%\noindent(1)when $G_Z$ is unbalanced,
%$x_{XY}=\frac{1}{2}\textbf{1}_X$;
%
%\noindent(2)when $G_Z$ is balanced
%\begin{align}
%x_{XY}&=\frac{1}{2}\textbf{1}_X+(I_X-P_X)^{-1}P_Y\hat{\textbf{1}}_{Z,S}\hat{\pi}_{Z,S_Z}^T(x_{0Z}-\frac{1}{2}\textbf{1}_Z)\nonumber
%\end{align}
%\noindent(3)when $G_Z$ is anti-balanced and $t$ is odd
%\begin{align}
%x_{XY,o}&=\frac{1}{2}\textbf{1}_X-(I_X+P_X)^{-1}P_Y\hat{\textbf{1}}_{Z,S}\hat{\pi}_{Z,S_Z}^T(x_{0Z}-\frac{1}{2}\textbf{1}_Z)\nonumber
%\end{align}
%\noindent when $G_Z$ is anti-balanced and $t$ is even
%\begin{align}
%x_{XY,e}&=\frac{1}{2}\textbf{1}_X+(I_X+P_X)^{-1}P_Y\hat{\textbf{1}}_{Z,S}\hat{\pi}_{Z,S_Z}^T(x_{0Z}-\frac{1}{2}\textbf{1}_Z)\nonumber
%\end{align}
\end{theoremm}
% Wei deleted due to space %% Yanhua resumed for techreport.
{\color{black} \begin{proof}
 Let initial distribution $x_0^T=[x_{0X}^T,x_{0Z}^T]$ and
 ${g^-}^T=[{g^-_X}^T,{g^-_Z}^T]$. When $t\rightarrow \infty$,
 Eq.~(\ref{eq:signed2}) can be written as
 $$x^T=\lim_{t\rightarrow
 \infty}(P^tx_0)^T=[x_{XY}^T, x_Z^T]=[x_X^T+x_Y^T, x_Z^T],$$
 where
 $x_X=\lim_{t\rightarrow\infty}(P_X^tx_{0X}+\sum_{i=0}^{t-1}P_X^ig_X^-)$,
 $x_Y=\lim_{t\rightarrow\infty}(P_Y^{(t)}x_{0Z}+\sum_{i=0}^{t-1}P_Y^{(i)}g_Z^-)$,
 and
 $x_Z=\lim_{t\rightarrow\infty}(P_Z^tx_{0Z}+\sum_{i=0}^{t-1}P_Z^ig_Z^-)$.

 From Lemma~\ref{thm:convergWeakly2}, $\lim_{t\rightarrow
 \infty}P_X^t=\textbf{0}$, thus $x_X=(I_X-P_X)^{-1}g_X^-$ holds for
 any ergodic $G_Z$. Since $G_Z$ is ergodic, $x_Z$ follows
 Theorem~\ref{thm:ltdyn}. Below we will focus on deriving $x_Y$,
 where the first part of $x_Y$ satisfies
 Lemma~\ref{thm:convergWeakly2}, i.e.,
 \begin{align}
 &\lim_{t\rightarrow\infty}P_Y^{(t)}x_{0Z}=\left\{
 \begin{array}{c c}
             \textbf{0} & G_Z \text{\ is strictly unbalanced} \\
             u_b\hat{\pi}_{Z,S_Z}^Tx_{0Z} & G_Z \text{\ is balanced}\\
             -u_u\hat{\pi}_{Z,S_Z}^Tx_{0Z} & G_Z \text{\ is anti-balanced, even $t$}\\
             u_u\hat{\pi}_{Z,S_Z}^Tx_{0Z} & G_Z \text{\ is anti-balanced, odd $t$.}
             \end{array}\right.\nonumber
 \end{align}
 The second part of $x_Y$ can be further written down as
 \begin{align}
 &\lim_{\substack{m\rightarrow \infty}}\sum_{t=1}^{m}P_Y^{(t)}g^-_Z=
 \lim_{\substack{m\rightarrow
 \infty}}\sum_{t=0}^{m-1}\sum_{i=0}^{t}(P_X^{t-i}P_Y
 P_Z^{i})g^-_Z\nonumber \\
 &=\lim_{\substack{m\rightarrow\infty}}\sum_{t=0}^{m-1}\sum_{i=0}^{m-t}(P_X^{t}P_Y
 P_Z^{i})g^-_Z=\sum_{t=0}^{\infty}(P_X^{t}P_Y
 \sum_{i=0}^{\infty}P_Z^{i})g^-_Z\label{eq:indexchange}
 \end{align}
 Now we discuss Eq.(\ref{eq:indexchange}) under different balance
 structures of $G_Z$.

 \noindent(1) \textbf{$G_Z$ is strictly unbalanced.}
 From Lemma~\ref{thm:convergWeakly2},
     $\lim_{\substack{t\rightarrow \infty}}P^t=\textbf{0}$.
 Then by Eq.(\ref{eq:strictlyunb}) we directly
     obtain that $x_{XY}=\frac{1}{2}\textbf{1}_X$.
 %For the purpose of derivation for the other cases, we derive a useful
%     equality below using the above result.
 Applying
 Eq.(\ref{eq:strictlyunb}) to $\sum_{i=0}^{\infty}P_Z^{i}g^-_Z$ in
     Eq.(\ref{eq:indexchange}), we have
 \begin{align}
 &\lim_{\substack{m\rightarrow
 \infty}}\sum_{t=1}^{m}P_Y^{(t)}g^-_Z=\frac{1}{2}(I_X-P_X)^{-1}P_Y\textbf{1}_Z.
     \nonumber
 \end{align}
 Thus, we obtain the following equation:
 \begin{align}
     x_{XY}&=x_{X}+x_{Y}=(I_X-P_X)^{-1}(g_X^-+\frac{1}{2}P_Y\textbf{1}_Z)=\frac{1}{2}\textbf{1}_X.\nonumber
 \end{align}
 %% where the last equality holds since
 %% \begin{align}
 %% &g_X^-+P_Y\frac{1}{2}\textbf{1}_Z-(I_X-P_X)\frac{1}{2}\textbf{1}_X \nonumber\\
 %% &=\frac{1}{2}(2g_X^-+P_Y\textbf{1}_Z-\textbf{1}_X+P_X\textbf{1}_X)=0.\nonumber
 %% \end{align}

 \noindent(2) \textbf{$G_Z$ is balanced.} Using
 Eq.(\ref{eq:balancedresults3}), we have
 \begin{align}
 &\lim_{\substack{m\rightarrow \infty}}\sum_{t=1}^{m}P_Y^{(t)}g^-_Z%=
 %\sum_{t=0}^{\infty}(P_X^{t}P_Y\lim_{\substack{k\rightarrow
 %\infty}}\sum_{i=0}^{\infty}
 %P_Z^{i})g^-_Z\nonumber \\
 =\frac{1}{2}(I_X-P_X)^{-1}P_Y(\textbf{1}_Z-\hat{\textbf{1}}_{Z,S_Z}\hat{\pi}_{Z,S_Z}^T\textbf{1}_Z)\nonumber
 \end{align}
 Hence, we have
 \begin{align}
 x_{XY}&=(I_X-P_X)^{-1}(g_X^-+\frac{1}{2}P_Y\textbf{1}_Z) +
     u_b\hat{\pi}_{Z,S_Z}^T(x_{0Z}-\frac{1}{2}\textbf{1}_Z)
=\frac{1}{2}\textbf{1}_X+u_b\hat{\pi}_{Z,S_Z}^T(x_{0Z}-\frac{1}{2}\textbf{1}_Z)\label{eq:posbal}
 \end{align}

 \noindent(3) \textbf{$G_Z$ is anti-balanced.}
 Using Lemma~\ref{lem:GGp}, we can negate the signs of
     all edges in $G$ so that the sink becomes balanced.
 Hence, we know that at even
 steps in long term,
 \begin{align}
 x_{XY,e}&=\frac{1}{2}\textbf{1}_X-u_u\hat{\pi}_{Z,S_Z}^T(x_{0Z}-\frac{1}{2}\textbf{1}_Z),\label{eq:negbaleven}
 \end{align}
 where Eq.(\ref{eq:negbaleven}) and Eq.(\ref{eq:posbal}) are
 identical in the sense that $P_X$'s and $P_Y$'s in Eq.(\ref{eq:negbaleven}) and
 Eq.(\ref{eq:posbal}) have opposite signs. Moreover, the odd step
 influence distribution sequence is obtained
 \begin{align}
 x_{XY,o}&=P_X x_{XY,e}+ P_Y x_{Z,e} + g_X^-%\nonumber \\
 %&-(I_X+P_X)^{-1}P_Y\hat{\textbf{1}}_{Z,S}\hat{\pi}_{Z,S_Z}^T(x_{0Z}-\frac{1}{2}\textbf{1}_Z)
=\frac{1}{2}\textbf{1}_X+u_{u}\hat{\pi}_{Z,S_Z}^T(x_{0Z}-\frac{1}{2}\textbf{1}_Z).
 \end{align}
 \end{proof}
%% deletion end
}

Theorem~\ref{thm:ltWeakly} characterizes the long-term dynamics when
the underlying graph is a weakly connected signed digraph with one
ergodic sink component. We can see that the results for balanced and
anti-balanced sink components
    are more complicated than the ergodic digraph
    case, since how non-sink components are connected to the sink
    subtly affects the final outcome of the steady state behavior.
In steady state,
while the sink component is still in one of the two polarized states as
    stated in Theorem~\ref{thm:converge}, the non-sink components exhibit
    more complicated color distribution, for which we provide
    probability characterizations in Theorem~\ref{thm:ltWeakly}.
{\color{black} Using Eq.(\ref{eq:Pmatmultisink}),
Theorem~\ref{thm:ltdyn} and Theorem~\ref{thm:ltWeakly} can be
readily extended to the
    case with more than one ergodic sink components
    and disconnected digraphs.} When the network only contains
    positive directed edges, the voter model dynamics can be
    interpreted using digraph random walk
    theory~\cite{li2012ToN,li2010random,li2010randomWAW,li12IM}.

%% file: ArXiv-sec-4Maximization.tex
\section{Influence maximization}\label{sec:max}

With the detailed analysis on voter model dynamics for signed digraphs,
    we are ready now to solve the influence maximization
    problem.
Intuitively, we want to address the following question:
\noindent{\emph{If only at most $k$ nodes could be selected initially
            and be turned white while all other nodes are black,
how should we choose seed nodes so as to maximize the expected
number of white nodes in short term and in long term,
respectively?}}
%\noindent{\emph{Given a fixed budget $B$, and the cost $c$ for
%choosing each user $u$ to be active initially, how to select the
%initial active users in voter model, so as to maximize the expected
%number of active users in short term and in long term,
%respectively?}}

\subsection{Influence maximization problem}
%\noindent{\textbf{Influence contributions of nodes.}}

\noindent\textbf{Influence maximization objectives.} We consider two
types of short-term influence objectives, one is the \emph{{instant
influence}}, which counts the total number of influenced nodes at a
step $t>0$; the other is the \emph{{average influence}}, which takes
the average number of influenced nodes within the first $t$ steps.
These two objectives have different implications and applications.
For example, political campaigns try to convince voters who may
change their minds back and forth, but only the voters' opinions on
the voting day are counted, which matches the \emph{instant
influence}. On the other hand, a credit card company would like to
have customers keep using its credit card service as much as
possible, which is better interpreted by the \emph{average
influence}. When $t$ is sufficiently large, it becomes the long-term
objective,
    and long-term average influence
    coincides with long-term instant influence when the dynamic converges.

Formally, we define the {\emph{short-term instant influence}
$f_t(x_0)$ and the \emph{short-term average influence}
$\bar{f}_t(x_0)$ as follows:
\begin{align}
f_t(x_0)&:=\textbf{1}^T x_t(x_0) \mbox{ and }
%\label{eq:objsto}\\
\bar{f}_t(x_0):=\frac{\sum_{i=0}^t{f_i(x_0)}}{t+1}.\label{eq:objsta}
\end{align}
Moreover, we define {\emph long term influence} as
\begin{align}
f(x_0)&:=\lim_{t\rightarrow
\infty}\frac{\sum_{i=0}^t{f_i(x_0)}}{t+1}.\label{eq:objlto}
\end{align}
%\begin{align}
%f(x_0)&:=\lim_{t\rightarow
%\infty}\frac{\sum_{i=0}^t{f_i(x_0)}}{t}\nonumber\\
%&= \left\{
%            \begin{array}{l|l}
%            f(x_0)=(f_{o}(x_0)+f_{e}(x_0))/2 & \text{Negatively balanced}\\ \hline
%            \textbf{1}^T x_t(x_0) & \text{otherwise}
%            \end{array}\right.
%            \label{eq:objlto}
%\end{align}
Note that when the dynamic converges (e.g. ergodic
    balanced or ergodic strictly unbalanced
    graphs), $f(x_0)=\lim_{t\rightarrow\infty} f_t(x_0)$.
For ergodic anti-balanced graphs (or sink components), it is
essentially
    the average of even- and odd-step limit influence.

Given a set $W\subseteq V$,
Let $e_W$ be the vector in which $e_{W}(j)=1$ if $j\in W$
and $e_{W}(j)=0$ if $j\not\in W$, which represents the initial seed
distribution with only nodes in $W$ as white seeds.
Let $e_i$ be the shorthand of $e_{\{i\}}$.
%% Let $e_{i}$ be the vector in which $e_{i}(j)=1$ if $j=i$
%% and $e_{i}(j)=0$ if $j\ne i$, which represents the initial seed
%% distribution with only node $i$ selected.
%The instant influence
%$f_t(x_{0i})$ at step $t$ is referred to as the node.
Unlike unsigned graphs, if initially no white seeds are selected on
a signed digraph $G$, i.e., $x_0=\textbf{0}$, the instant influence
$f_t(\textbf{0})$ at step $t$ is in general non-zero, which is
referred to as the \emph{ground influence} of the graph $G$ at $t$.
The influence contribution of a seed set $W$ does not count such
ground influence, as shown in definition~\ref{dfr:con}.
\begin{dfr}[Influence contribution]\label{dfr:con}
The \emph{instant influence contribution} of a
seed set $W$ to the $t$-th step instant influence objective,
    denoted by $c_t(W)$, is the
difference between the instant influence at step $t$ with only nodes
in $W$ selected as seeds and the ground influence at step $t$:
$c_t(W)= f_t(e_W)-f_t(\textbf{0})$. The {\em average influence
contribution} $\bar{c}_t(W)$ and
    {\em long-term influence contribution} $c(W)$ are defined in the same way:
$\bar{c}_t(W)= \bar{f}_t(e_W)-\bar{f}_t(\textbf{0})$ and
$c(W)= f(e_W)-f(\textbf{0})$.
\end{dfr}

We are now ready to formally define the influence maximization problem.

\begin{dfr}[Influence maximization]
The {\em influence maximization} problem for short-term instant influence
    is finding a seed set $W$ of at most $k$ seeds that maximizes
    $W$'s instance influence contribution at step $t$, i.e.,
    finding $W^*_t=\arg\max_{|W|\le k} c_t(W)$.
Similarly, the problem for average influence and long-term influence is
    finding $\bar{W}^*_{t}=\arg\max_{|W|\le k} \bar{c}_t(W)$ and
    $W^*=\arg\max_{|W|\le k} c(W)$, respectively.
\end{dfr}

We now provide some properties of influence contribution, which
    lead to the optimal seed selection rule.
By Eq.(\ref{eq:signed2}), we have
\begin{align}
c_t(W)&=f_t(e_W)-f_t(\textbf{0})= \textbf{1}^T x_t(e_W)-
\textbf{1}^T x_t(\textbf{0})= \textbf{1}^T P^t e_W. \label{eq:stimc}
\end{align}%
Let $c_t(i)$ be the shorthand of $c_t(\{i\})$, and let $c_t=[c_t(i)]$
    denote the vector of influence contribution of individual nodes.
    Then
$c_t^T=[c_t(i)]^T=\textbf{1}^T P^t$.
% moved to the text line by Yanhua due to the space limitation-below
%\begin{align}
%c_t^T&=[c_t(i)]^T=\textbf{1}^T P^t. \label{eq:stimcM}
%\end{align}
% moved to the text line by Yanhua due to the space limitation-above
When $t\rightarrow \infty$, the long term influence contributions of
    individual nodes are obtained as a vector $c$:
\begin{align}%
c^T&=\lim_{t\rightarrow\infty}\frac{\sum_{i=0}^t c_i^T}{t+1}
    =\lim_{t\rightarrow\infty} \frac{\textbf{1}^T\sum_{i=0}^t P^i}{t+1}.
    \label{eq:stimcMc}
\end{align}
%-Yanhua-moved the following sentence to Formula line-below
When $P^t$ converges, we simply have
%-Yanhua-moved the following sentence to Formula line-above
\begin{align}%
%\mbox{\noindent When $P^t$ converges, we simply have}
c^T&=\textbf{1}^T \lim_{t\rightarrow\infty}  P^t.
\label{eq:stimcMc2}
\end{align}%
%and the influence contribution vector $c_t$ is
%\begin{align}
%\end{align}

Lemma~\ref{lem:linear} below discloses the important property that
    the influence contribution is a linear set function.
\begin{lemmaa}\label{lem:linear}
Given a white seed set $W$, $c_t(W)=\sum_{i\in W}c_t(i)$,
    $\bar{c}_t(W)=\sum_{i\in W}\bar{c}_t(i)$, and $c(W)=\sum_{i\in W}c(i)$.%the influence contribution
%made by the seed set $S_e$ to the $t$ step instant influence, is the
%sum of the influence contributions of these individual seeds.
\end{lemmaa}
\begin{proof}
From Eq.(\ref{eq:stimc}), we have
\begin{align}
c_t(W)&=\textbf{1}^T P^t e_W = \textbf{1}^T P^t\sum_{i\in W}e_i
    = \sum_{i\in W}\textbf{1}^T P^te_i = \sum_{i\in W}c_t(i).
\nonumber
\end{align}%
The linearity of $\bar{c}_t$ and $c$ can be derived from that of $c_t$.
\end{proof}

%any two node $j$ and $l$, the influence contribution of the two
%nodes is equal to the sum of the influence contribution of each of
%them, namely,
%\begin{align}
%&c_t(j)+c_t(l)=f_t(x_{0j})-f_t(\textbf{0})+f_t(x_{0l})-f_t(\textbf{0})\nonumber \\
%&=\textbf{1}^T P^t (x_{0j}+x_{0l})=
%f_t(x_{0j}+x_{0l})-f_t(\textbf{0})= c_t(\{j,l\}).\nonumber
%\end{align}
%% Hence, the instant influence contributions of nodes captures the
%% capabilities of nodes to influence others, which could be either
%% positive or negative. If $c_t(i)>0$, node $i$ will boost up the
%% instant influence $f_t$, when selected. On the other hand, nodes
%% with non-positive instant influence contributions, i.e., $c_t(i)\leq
%% 0$, will decrease (or keep invariant) the instant influence, thus
%% are not worth of selecting.

Given a vector $v$, let $n^+(v)$ denote the number of
    positive entries in $v$.
%% Let $n^+(c_t)$ (resp. $n^+(\bar{c}_t)$ or $n^+(c)$) denote the total number
%% of nodes with positive instant (resp. average or long-term)
%%  influence contributions.
    By applying
Lemma~\ref{lem:linear}, we have the optimal seed selection rule for
instant influence maximization as follows.

\noindent{\textbf{Optimal seed selection rule for instant influence
maximization.}} \emph{Given a signed digraph and a limited budget
$k$, selecting top $\min\{k,n^+(c_t)\}$ seeds with the highest
$c_t(i)$'s, $i \in V$, leads to the maximized instant influence at
step $t>0$.}

Note that the influence contributions of some nodes may be negative and
    these nodes should not be selected as white seeds, and thus
    the optimal solution may have less than $k$ seeds.
The rules for average influence maximization and long-term influence
    maximization are patterned in the same way.
Therefore, the central task now becomes the computation of
    the influence contributions of individual nodes.
Below, we will introduce our SVIM algorithm, for Signed Voter model
    Influence Maximization.
%% Below, we will introduce our algorithm SiVoIM, for {\it i}fluence
%%  {\it m}aximization for {\it vo}ter model on {\it s}igned networks.

%% Below, we will introduce the signed voter model influence
%% maximization (SVMIM) algorithm which computes the exact influence
%% contribution vector for various signed digraphs, and selects the
%% optimal initial seeds with maximized influence.
%% %solves the influence maximization problems by selecting initial
%% %seeds with highest influence contributions and we prove the
%% %optimality.
%% At the end of this section, we will also briefly discuss some
%% variations and generalizations of the influence maximization
%% problems, such as heterogeneous user selection costs and
%% probabilistic user selection scheme.

\subsection{Short-term influence maximization}

%In this section, we will solve this short term influence
%maximization problem and introduce an efficient algorithm to select
%the initial seeds.

By applying Definition~\ref{dfr:con} and Lemma~\ref{lem:linear}, we
develop SVIM-S algorithm to solve the short-term instant and
average influence maximization problem, as shown in Algorithm~\ref{alg:stmax}.
\begin{algorithm}[!htb]
\caption {Short-term influence maximization SVIM-S} \label{alg:short}
\begin{algorithmic}[1]
%\footnotesize{
%\scriptsize{
    \STATE {\textbf{INPUT:} Signed transition matrix $P$, short-term period $t$, budget $k$;}
    \STATE {\textbf{OUTPUT:} White seed set $W$.}
    %% \STATE {\textbf{OUTPUT:} Initial seed indicator $x_0$ where $x_0(i)=1$, if $i$ is a seed; $x_0(i)=0$, otherwise.}
    \STATE {$c_t=\textbf{1}$; $\bar{c}_t=\textbf{1}$;}
    \FOR {$i=1:t$}
        \STATE {$c_t^T=c_t^T P$;(for instant influence maximization.)}
        \STATE {$\bar{c}_t=\bar{c}_t+c_t$; (for average influence maximization.)}
    \ENDFOR
    \STATE {$W=$ top $\min\{k,n^+(c_t)\}$ (resp. $\min\{k,n^+(\bar{c}_t)\}$)
    nodes with the highest $c_t(i)$ (resp. $\bar{c}_t(i)$) values,
    for instant (resp. average) influence maximization.}
%    \STATE {Select $\min\{k,n^+(c_t)\}$ nodes, with the highest $c'_t(i)$ values as initial seeds.(for average influence maximization.)}
%   \STATE {Select $k=\lfloor\frac{B}{c}\rfloor$ nodes, with the highest $\pi'_t(i)$ values as initial seeds.(for average influence maximization.)}
    %% \STATE {Set $x_0(i)=1$, if $i$ is a seed, and set $x_0(i)=0$, otherwise.}
\end{algorithmic}
\label{alg:stmax}%\vspace*{-0.1cm}
\end{algorithm}\vspace*{-0.1cm}

%\begin{theoremm}[]\label{thm:stmax}
%The seeds selected by Algorithm~\ref{alg:stmax}, denoted by $x_0$ is
%the optimal solution to the short-term instant (resp. average)
%influence maximization problems, i.e.,
%\begin{align}
%\max_{x_0} &f_t(x_0) \textrm{  for short term},
%\label{eq:objst} \\
%\max_{x_0} &f(x_0) \textrm{  for long term}, \label{eq:objlt}
%\end{align}
%\end{theoremm}
%\begin{proof}
%First, we consider the instant influence maximization. From
%Theorem~\ref{thm:stdyn}, we have
%\begin{align}
%\max_{x_0} f_t(x_0)&=\max_{x_0}\textbf{1}^T x_t(x_0) =
%\max_{x_0}\textbf{1}^T
%(P^tx_0+(\sum_{i=0}^{t-1}P^i)g^-)\nonumber \\
%&=\max_{x_0}(c_t^Tx_0)+f_t(\textbf{0})\label{eq:stmaxxx}
%\end{align}
%where the second part is instant influence at step $t$ obtained with
%no initial seed selected. Hence, we only need to maximize the first
%part $\max_{x_0}c_t^Tx_0$, Thus, to maximize the influence at step
%$t$, at most $k$ positive influence contribution nodes should be
%selected,
%%$k$ seeds need to be selected with highest instant influence
%%contribution $c_t(i)$'s,
%which yields Alg~\ref{alg:stmax}. Similar
%proof applies to the average influence maximization.
%\end{proof}
%\noindent\textbf{Complexity analysis.}
SVIM-S algorithm requires $t$ vector-matrix multiplications, each of
which takes $|E|$ times entry-wise multiplication operations. Hence
the total time complexity of SVIM-S is $O(t\cdot|E|)$.

%In general, graph $G$ is very sparse, thus $m\sim O(n)$.
%Moreover, given a short term factor $t$ which is in general smaller
%than the mixing time of the Markov chain on unsigned the digraph
%${\cal G}$ of $G$, since once $t$ is larger than the mixing time,
%the problem transfers to a long term influence maximization problem.
%Many empirical
%studies~\cite{SandPHaifengMixing,lesniewski2010whanau} show that
%real-world social networks are indeed fast mixing, where the mixing
%time is $O(log n)$ for $\epsilon = O(1/n)$, where $n=|V|$ is size of
%the graph. Hence, $t\leq O(\log n)$. All in all, the total time
%complexity of the algorithm is $O(m\cdot \log n)$ which is near
%linear.

\subsection{Long-term influence maximization}

%When considering the long term influence maximization problem, the
%node long term influence contribution vector $c$ is in form of
%\begin{align}
%c^T&:=\lim_{t\rightarrow\infty}c_t^T=\lim_{t\rightarrow\infty}
%\textbf{1}^TP^t. \label{eq:stimcMc}
%\end{align}
We now study the long-term influence contribution $c$ and introduce
the corresponding
    influence maximization algorithm SVIM-L.
We will see that the computation of influence contribution $c$ and seed
    selection schemes depends on the structural balance and connectedness
    of the graph.
While seed selection for balanced ergodic digraphs still has intuitive
    explanations, the computation for weakly connected and disconnected digraphs
    is more involved and less intuitive.

%% In this part, we study the long term influence contribution $c$,
%% where it will be clear shortly that $c$ turns to be with various
%% formats for different underlying signed digraph structures, such as
%% ergodic signed digraphs, weakly connected signed digraphs. Below, we
%% will introduce efficient algorithms to compute the long term
%% influence contribution vectors for various signed graph structures,
%% which in turn determine the initial seed selection that maximizes
%% the long term influence of the voter model.

\subsubsection{Case of ergodic signed digraphs}

When the signed digraph $G=(V,E,A)$ is ergodic,
Lemma~\ref{thm:ltinfmax} below characterizes the long-term influence
contributions of nodes, with respect to various balance structures.
\begin{lemmaa}[]\label{thm:ltinfmax}
Consider an ergodic signed digraph $G=(V,E,A)$. If $G$ is
\emph{balanced}, with bipartition $S$ and $\bar{S}$, the influence
contribution vector $c=(|S|-|\bar{S}|)\hat{\pi}_S$. If $G$ is
\emph{anti-balanced} or \emph{strictly unbalanced}, $c=\textbf{0}$.
\end{lemmaa}

\begin{proof}
\noindent(1) When $G$ is balanced, by Lemma~\ref{thm:BalDig2} and
    Eq.(\ref{eq:stimcMc2}),
\begin{align}
    c^T&=\textbf{1}^T \lim_{t\rightarrow \infty}P^t=\textbf{1}^T \hat{\textbf{1}}_S\hat{\pi}_S^T=(|S|-|\bar{S}|)\hat{\pi}_S.\nonumber
\end{align}
%Hence, when $|S|=|\bar{S}|$, namely, two components have the same
%size, the objective always equals to 0, the influence is invariant
%to the initial seed selection. When $|S|\neq |\bar{S}|$, without
%loss of generality, we assume $|S|>|\bar{S}|$. Then, we must choose
%seed node from set $S$ (the larger component), since choosing any
%node as seed from the $\bar{S}$ will decrease the objective. Hence
%the maximization is achieved, when choosing  $k=min\{|S|,
%\lfloor\frac{B}{c}\rfloor\}$ nodes with the highest stationary
%distribution ${\pi}_i$'s in the larger component $S$.
%
%In particular, when $G$ is unsigned, it is equivalent to $S=V$ and
%$\bar{S}=\phi$, thus, the optimal solution to the maximization
%problem Eq.(\ref{eq:bx}) is to choose $k=\lfloor\frac{B}{c}\rfloor$
%nodes with the highest stationary distribution ${\pi}_i$'s in the
%graph.

\noindent(2) When $G$ is strictly unbalanced, again by Lemma~\ref{thm:BalDig2}
    and Eq.(\ref{eq:stimcMc2}), we have
     $c^T=\textbf{1}^T \lim_{t\rightarrow \infty}P^t=\textbf{0}$.

%$f(x_0):=\textbf{1}^T x = \textbf{1}^T(I-P)^{-1}g^-$, which
%infers that when $t$ is sufficiently large, the probability of each
%node $i$ to be active converges to a stationary value, which is
%invariant and independent to the initial state, and only hinges on
%the underlying strictly
%unbalanced graph structure.

\noindent(3) When $G$ is
anti-balanced, by Lemma~\ref{thm:BalDig2}
    and Eq.(\ref{eq:stimcMc}), we have
\begin{align}
    c^T&=\textbf{1}^T \frac{\lim_{t\rightarrow \infty}P^{2t}+
    \lim_{t\rightarrow \infty}P^{2t+1}}{2}
    =\textbf{0}.\nonumber
\end{align}
%% the long term influence of white nodes (as the
%% objective function Eq.(\ref{eq:objlto})) is
%% $f(x_0)=\textbf{1}^T({x_{o}(x_0)+x_{e}(x_0)})/{2}=|V|/2$
%% %\begin{align}
%% %f(x_0)&=\textbf{1}^T\frac{x_{o}(x_0)+x_{e}(x_0)}{2}=\textbf{1}^T\textbf{1}/2=|V|/2\nonumber
%% %\end{align}
%% %\begin{align}
%% %f(x_0)&:=\frac{f_{o}(x_0)+f_{e}(x_0)}{2}\nonumber \\
%% %&=\textbf{1}^T(I-P-\textbf{1}_s{\pi}_s^T)^{-1}g^-
%% % -\frac{1}{2}\textbf{1}^T\frac{1}{2}\textbf{1}_s{\pi}_s^Tg^-\nonumber \\
%% %&=\textbf{1}^T(I-P)^{-1}g^-=\frac{1}{2}\textbf{1},\nonumber
%% %\end{align}
%% Thus, $f(e_i)=f(\textbf{0})=|V|/2$, and %the long term influence
%% %contribution $c(i)$ of each node $i$ is
%% $c(i)=f(e_i)-f(\textbf{0})=0$.%, i.e., the long term influence of
%% %voter model on anti-balanced digraphs is independent to the initial
%% %seed selection.
\end{proof}

Based on Lemma~\ref{thm:ltinfmax}, Algorithm~\ref{alg:strongly}
summarizes how to compute the long-term influence contribution $c$
on ergodic signed digraphs.

%Based on Theorem~\ref{thm:ltinfmax}, the selection rule for ergodic
%signed digraphs is summarized in Alg~\ref{alg:strongly}.

\begin{algorithm}[!htb]
\caption{{%(Long term) SVMIM algorithm for strongly connected signed
%digraphs
$c=ergodic(G)$}}
\begin{algorithmic}[1]
%\footnotesize{
%\scriptsize{
    \STATE {\textbf{INPUT:} Signed transition matrix $P$.}
    \STATE {\textbf{OUTPUT:}  Long term influence contribution vector $c$}
%    \STATE {\textbf{OUTPUT:} Initial seed indicator $x_0$ where $x_0(i)=1$, if $i$ is a seed; $x_0(i)=0$, otherwise.}
    \STATE {Detect the structure of ergodic signed digraph $G$;}
    \IF {$G$ is balanced, with bipartition $S$ and $\bar{S}$}
%        \STATE {$Se$ include $\min\{|S|,n\}$ seeds with highest $\pi_i$ in $S$;}
        \STATE {Compute stationary distribution $\pi$ of $\bar{P}$;}
        \STATE {$c=(|S|-|\bar{S}|)\hat{\pi}_S$;}
    \ELSE
        \STATE {$c=\textbf{0}$;}
    \ENDIF
%    \STATE {Set $x_0(i)=1$, if $i\in Se$, and set $x_0(i)=0$, otherwise.}
%    \STATE {\textbf{RETURN}}
\end{algorithmic}
\label{alg:strongly}
\end{algorithm}\vspace*{-0.0cm}

Lemma~\ref{thm:ltinfmax} suggests that for ergodic balanced
digraphs, we should pick the larger component, e.g., $S$, if
$|S|>|\bar{S}|$, and select the top $\min\{k, |S|\}$ nodes from $S$
with the largest stationary distributions as white seeds. Selecting
these nodes will make the probability of the larger component
    being white the largest.

% Wei deleted due to space % Yanhua resumed for techreport

{\color{black} Theorem~\ref{thm:ltdyn} indicates that given an
anti-balanced
 digraph $G$, with bipartition $S$ and $\bar{S}$, the long-term
 dynamic $x_t$ oscillates on odd and even steps, and their long-term
  influence contribution is $0$.
 However, we can still maximize the strength of the
 oscillation of the voter model on an anti-balanced ergodic digraph
 by properly choosing the initial white seeds (See
 Remark~\ref{rmk:Rmax}.)

 \begin{remark}[]\label{rmk:Rmax}
 In an anti-balanced ergodic digraph $G=(V,E,A)$ with the bipartition $S$ and
 $\bar{S}$ and a budget $k$. Let $W'$ (resp. $W''$) denote two
 initial seed sets, where $\min\{k,|S|\}$ (resp.
 $\min\{k,|\bar{S}|\}$) nodes, with highest stationary distribution
 $\pi(i)$'s in $S$ (resp. $\bar{S}$), are selected. Then, the optimal
 $W^*$ that maximizes the strength of oscillation is
 \begin{align}
 W^*&:=\argmax_{W\in \{W',
 {W''}\}}|\hat{\pi}_S^T(e_W-\frac{1}{2}\textbf{1})|.\label{eq:Remark1}
 \end{align}
 \end{remark}}

{\color{black}
\begin{proof}
From Theorem~\ref{thm:ltdyn}, when $t$ becomes sufficiently large,
the vector $x$ oscillates at two vectors on odd and even steps,
respectively. The strength of the oscillation is
\begin{align}
\frac{|f_{o}(x_0)-f_e(x_0)|}{2}&=|\textbf{1}^T\frac{x_{o}(x_0)-x_e(x_0)}{2}|
=|\textbf{1}^T\hat{\textbf{1}}_S\hat{\pi}_S^T
(x_0-\frac{1}{2}\textbf{1})|=||S|-|\bar{S}||\cdot|\hat{\pi}_S^T
(x_0-\frac{1}{2}\textbf{1})|.\nonumber
\end{align}
Let $W$ be the initial seed set, then the oscillation strength
maximization is formulated as

\begin{align}
&\max_{|W|\leq k}||S|-|\bar{S}||\cdot|\hat{\pi}_S^T
(e_W-\frac{1}{2}\textbf{1})|
%&=\argmax_{x_0}\{\max_{x_0}\{\pi_S^T
%(x_0-\frac{1}{2}\textbf{1})\},\max_{x_0}\{-\pi_S^T
%(x_0-\frac{1}{2}\textbf{1})\}\}\nonumber\\
=||S|-|\bar{S}||\cdot\max\{\max_{|W|\leq k}\{\hat{\pi}_S^T
e_W\}-\frac{1}{2}\hat{\pi}_S^T\textbf{1},\max_{|W|\leq
k}\{-\hat{\pi}_S^T
e_W\}+\frac{1}{2}\hat{\pi}_S^T\textbf{1}\},\label{eq:maxosc}
%=\max_{x_0}\{\max_{x_0}\{\pi_S^T x_0\},\max_{x_0}\{-\pi_S^T
%x_0\}+\textbf{1}\}}\nonumber
\end{align}
%\begin{align}
%&W^*=\argmax_{|W|\leq k}|\hat{\pi}_S^T (e_W-\frac{1}{2}\textbf{1})|
%%&=\argmax_{x_0}\{\max_{x_0}\{\pi_S^T
%%(x_0-\frac{1}{2}\textbf{1})\},\max_{x_0}\{-\pi_S^T
%%(x_0-\frac{1}{2}\textbf{1})\}\}\nonumber\\
%=\argmax\{\max_{|W|\leq k}\{\hat{\pi}_S^T
%e_W\}-\frac{1}{2}\hat{\pi}_S^T\textbf{1},\max_{|W|\leq
%k}\{-\hat{\pi}_S^T
%e_W\}+\frac{1}{2}\hat{\pi}_S^T\textbf{1}\}.\nonumber
%%=\max_{x_0}\{\max_{x_0}\{\pi_S^T x_0\},\max_{x_0}\{-\pi_S^T
%%x_0\}+\textbf{1}\}}\nonumber
%\end{align}
which contains two sub-problems, i.e., $\max_{|W|\leq
k}\{\hat{\pi}_S^T e_W\}$ and $\max_{|W|\leq k}\{-\hat{\pi}_S^T
e_W\}$. The first maximization problem can be rewritten as
\begin{align}
\max_{|W|\leq k}\{\hat{\pi}_S^T e_W\}&=\max_{|W|\leq k} \big(
\sum_{i\in S}{\pi}(i)e_W(i)-\sum_{j\in
\bar{S}}{\pi}(i)e_W(j)\big).\label{eq:max1}
\end{align}
Thus, let $W'$ denote the optimal solution to the problem in
Eq.(\ref{eq:max1}), which is obtained by choosing $\min\{k,|S|\}$
seeds with highest $\pi(i)$'s from $S$. Similarly, choosing
$\min\{k,|\bar{S}|\}$ nodes with the highest $\pi(i)$'s from
$\bar{S}$ yields the optimal solution, denoted by $W''$, to the
second maximization problem $\max_{|W|\leq k}\{-\hat{\pi}_S^T
e_W\}$. The optimal $W$ to the problem in eq.(\ref{eq:maxosc}) that
maximizes the oscillation strength is the one in $\{W',W''\}$, with
higher $|\hat{\pi}_S^T (e_W-\frac{1}{2}\textbf{1})|$, which
completes the proof of eq.(\ref{eq:Remark1}).

\end{proof}}

\subsubsection{Case of weakly connected signed digraphs}

%% In general, a connected signed digraph $G$ could be weakly
%% connected, with several sink ergodic components $G_1,\cdots,G_k$,
%% which only have incoming edges from the rest of the graph. Below, we
%% show how to compute the node influence contributions on weakly
%% connected signed digraphs.

We first consider a weakly connected signed $G$ which has a single
ergodic sink component $G_Z$ with only incoming edges from the
remaining nodes $X=V\setminus Z$.
%Define the \emph{weak connection vector},
%\begin{align}
%u_b=(I_X-P_X)^{-1}P_Y\textbf{1}_{S_Z},\nonumber
%\end{align}
%% Lemma~\ref{thm:weaksel} discloses that the node influence
%% contributions on $G$ depends on $u_b$ and the balance structures of
%% $G_Z$.
\begin{lemmaa}[]\label{thm:weaksel}
Consider a weakly connected digraph $G=(V,E,A)$ with a single
ergodic sink
    component $G_Z$.
If $G_Z$ is \emph{balanced}, with partition $S_Z$ and $\bar{S_Z}$,
the long term influence contribution vector $c^T=[c_X^T, c_Z^T]$,
where
    $c_X=\textbf{0}_X$
    and $c_Z=(\textbf{1}_X^Tu_b+|S_Z|-|\bar{S}_Z|)\hat{\pi}_{Z,S_Z}$. If
$G$ is \emph{anti-balanced} or \emph{strictly unbalanced},
$c=\textbf{0}$.
\end{lemmaa}
\begin{proof}
(1)When $G_Z$ is balanced, by Lemma~\ref{thm:convergWeakly2},
    $c_X=\textbf{0}_X$, and
\begin{align}
&c_Z^T=%\textbf{1}^T \lim_{t\rightarrow \infty}P^t=%\textbf{1}^T \left[
%            \begin{array}{c|c}
%            \textbf{0} & {P}_Y^{(\infty)} \\ \hline
%            \textbf{0} & {P}_Z^{(\infty)}
%            \end{array}\right]\nonumber\\
%[\textbf{1}_X^T|\textbf{1}_Z^T]\left[
%            \begin{array}{c|c}
%            \textbf{0} & u_b{\pi}_{S_Z}^T \\ \hline
%            \textbf{0} & \textbf{1}_{S_Z}{\pi}_{S_Z}^T
%            \end{array}\right]\nonumber \\
%&
(\textbf{1}_X^Tu_b+\textbf{1}_Z^T\hat{\textbf{1}}_{Z,S_Z})\hat{\pi}_{Z,S_Z}^T
%&=(\textbf{1}_X^Tu_{b}{\pi}_{S_Z}^T+\textbf{1}_Z^T\textbf{1}_{S_Z}{\pi}_{S_Z}^T)
%\nonumber \\
%&
=(\textbf{1}_X^Tu_b + |S_Z| -
|\bar{S}_Z|)\hat{\pi}_{Z,S_Z}^T.\nonumber
\end{align}

(2) When $G_Z$ is strictly unbalanced, $c^T=\textbf{1}^T
\lim_{t\rightarrow \infty}P^t=\textbf{0}$

(3) When $G_Z$ is anti-balanced, by Lemma~\ref{thm:convergWeakly2} the
    limits of odd and even subsequences of $P^t$ cancel out, thus
    $c=\textbf{0}$.
%% From Theorem~\ref{thm:ltWeakly}, when $G$ is anti-balanced, the
%% long term influence of white nodes (as the objective function
%% Eq.(\ref{eq:objlto})) is
%% $f(x_0)=\textbf{1}^T({x_{o}(x_0)+x_{e}(x_0)})/{2}=|V|/2$. Thus, the
%% long term influence contribution $c(i)$ of each node $i$ is
%% $c(i)=f(e_i)-f(\textbf{0})=0$.
\end{proof}

Lemma~\ref{thm:weaksel} indicates that influence contribution of
    the balanced ergodic sink component is more complicated than that of
    the balanced ergodic digraph.
This is because the sink component affects the colors of the
non-sink
    component in a complicated way depending on
    how non-sink and sink components are connected.
Therefore, the optimal seed selection depends on the calculation of
the influence contributions of each sink node, and is not as
intuitive as that for the
    ergodic digraph case.

% Wei deleted due to space % Yanhua resumed for techreport
{\color{black} Theorem~\ref{thm:ltWeakly} shows that in a weakly
connected signed
 digraph $G$, with single anti-balanced sink component $G_Z$, the
 long term influence $f(x_0)$ oscillates on odd and even steps, and
 the average is $|V|/2$, which is invariant to the initial seed
 selection. Similar to Remark~\ref{rmk:Rmax}, we can maximize the
 oscillation strength by properly selecting initial seeds, i.e.,
 \begin{align}
 &W^*=\argmax_{|W|\leq k}|f_{e}(e_W)-f_o(e_W)|/2\nonumber\\
 &=\argmax_{|W|\leq k}|(\textbf{1}_X^Tu_u\hat{\pi}_{Z,S_Z}^T+\textbf{1}_Z^T\hat{\textbf{1}}_{Z,S_Z}\hat{\pi}_{Z,S_Z}^T)(e_{WZ}-\frac{1}{2}\textbf{1}_Z)|\nonumber
 \\
 &=|\textbf{1}_X^Tu_u + |S_Z| -
 |\bar{S}_Z||\cdot\argmax_{|W|\leq k}|\hat{\pi}_{Z,S_Z}^T(e_{WZ}-\frac{1}{2}\textbf{1}_Z)|\label{eq:weaklyseloscillation}
 \end{align}
 where the maximization objective is independent from $x_{0X}$, thus
 oscillation strength maximization problem objective in
 Eq.(\ref{eq:weaklyseloscillation}) for $G$ is identical to that in
 Remark~\ref{rmk:Rmax}. Hence, Remark~\ref{rmk:Rmax} also applies
 here.}
%\textcolor{black}{[Wei: In the above formulas, should it be $|W| \le
%k$? Answer: Yes.]}

{\color{black} Using Eq.(\ref{eq:Pmatmultisink}),
Lemma~\ref{thm:ltinfmax} and Lemma~\ref{thm:weaksel} can be readily
extended to the
    case with more than one ergodic sink components
    and disconnected digraphs.}
{\color{black} Algorithm~\ref{alg:weaklyconnected} below summarizes
how to compute the node influence contributions of weakly connected
signed digraphs. Note that by our assumption, we consider all sink
components to be ergodic.}

\begin{algorithm}[!htb]
\caption {{%(Long term) SVIM algorithm for weakly connected signed
%digraphs
$c=weakly(G)$}} \label{alg:weaklyconnected}
\begin{algorithmic}[1]
%\footnotesize{
%\scriptsize{
    \STATE {\textbf{INPUT:} Signed transition matrix $P$.}
    \STATE {\textbf{OUTPUT:} Influence contribution vector $c$.}
    \STATE {Detect the structure of the weakly connected signed digraph $G$, and
    find its $m\geq 1$ signed ergodic sink components $G_{Z1},\cdots,G_{Zm}$;}
        \FOR {$i=1:m$}
                \IF{$G_{Zi}$ is balanced with partition $S_{Zi},\bar{S}_{Zi}$}
                    \STATE{Compute stationary distribution
                        $\pi_{Zi}$ of $\bar{P}_{Zi}$};
                    \STATE{$u_{bi}=(I_X-P_X)^{-1}P_{Yi}\hat{\textbf{1}}_{Zi,S_{Zi}}$;}
                    \STATE{$c_{Zi}=(\textbf{1}^T_Xu_{bi} + |S_{Zi}| -
                            |\bar{S}_{Zi}|)\hat{\pi}_{Zi,S_{Zi}}^T$;}
                \ENDIF
        \ENDFOR
%        \STATE{$Se$ include $\min\{\sum_{i}|Se_{si}|,n\}$ seeds that have largest values in $\bigcup_{i\leq m} (\sum{u_{bi}}+|S_i|-|\bar{S}_i|) Val_{si}$;}
%    \STATE {Set $x_0(i)=1$, if $i\in Se$, and set $x_0(i)=0$, otherwise.}
    \STATE {$c=[\textbf{0}_X; c_{Z1}; \cdots; c_{Zm}]$}
\end{algorithmic}
\end{algorithm}\vspace*{-0.3cm}

\subsubsection{General case and SVIM-L algorithm}

Given the above systematic analysis, we are now in a position to
summarize and introduce our SVIM-L algorithm which solves the long-term
voter model influence maximization problem for general aperiodic
signed digraphs.

In general, a signed digraph consists $m\geq1$ disconnected
components, within each of which the node influence contribution
follows %Lemma~\ref{thm:ltinfmax} and
Lemma~\ref{thm:weaksel}. The long-term signed voter model influence
maximization (SVIM-L) algorithm is constructed in
Algorithm~\ref{alg:gen}.

%To summarize, the selection rule for weakly connected signed
%digraphs is presented in pseudo-code in Alg~\ref{alg:gen}.
\begin{algorithm}[!htb]
\caption {Long-term influence maximization SVIM-L}
\begin{algorithmic}[1]
%\footnotesize{
%\scriptsize{
    \STATE {\textbf{INPUT:} Signed transition matrix $P$, budget $k$.}
    \STATE {\textbf{OUTPUT:} White seed set $W$.}
    \STATE {Detect the structure of a general aperiodic signed digraph $G$, and
    find the $m\geq 1$ disconnected components $G_1,\cdots, G_m$;}
        \FOR {$i=1:m$}
            \STATE {$c_{G_i}=weakly(G_i)$;}
        \ENDFOR
        \STATE {$c=[c_{G_1}; \cdots; c_{G_m}]$;}
%%    \ENDIF
    \STATE {$W =$ top $\min\{k, n^+(c)\}$ nodes with
the highest $c(i)$ values.}
%    \STATE {\textbf{RETURN}}
\end{algorithmic}
\label{alg:gen}
\end{algorithm}\vspace*{-0.0cm}

\noindent\textbf{Complexity analysis.}
We consider $G=(V,E,A)$ to be weakly connected, since disconnected graph case
    can be treated independently for each connected component for the
    time complexity.
SVIM-L algorithm consists of two parts.
The first part extracts the connectivity and balance structure of the graph,
    which can be done using depth-first search with complexity $O(|E|)$.
The second part uses Algorithm~\ref{alg:weaklyconnected}
    to compute influence contributions of balanced ergodic sink components.
The dominant computations are on the stationary distribution $\pi_{Zi}$'s and
    $(I_X-P_X)^{-1}$, which can be done by %singular value decomposition
    solving a linear equation system~\cite{PiInverse}
    and matrix inverse in $O(|Z_i|^3)$ and $O(n_X^3)$, respectively, where
    $n_X= |X|$. Let $b$ be the number of balanced sink components in $G$,
    $n_Z$ be the number of nodes in the largest balanced sink
    component.
Thus SVIM-L can be done in $O(b n_Z^3+n_X^3)$ time.
%%     can be done by matrix inverse in $O(|V|^3)$ time (or faster algorithm
%%     in $O(|V|^{2.38})$ time) and dominates rest computations.
%% Thus SVIM-L can be done in $O(t_M)$ time, where $t_M$ is the time of
%%     matrix inversion algorithm.
Alternatively, we can use iterative method for computing both
    $\pi_{Zi}$'s and $\textbf{1}_X^T(I_X-P_X)^{-1}$, if
    the largest convergence time $t_C$ of $P_{Zi}^t$'s and $P_X^t$ is
    small.
    {\color{black}(Note that the convergence time
    of ergodic digraphs could be exponentially large in general, as
    illustrated by an example in Appendix~\ref{sec:illustration})}.
In this case, each iteration step involves
    vector-matrix multiplication and can be done in $O(m_B)$ time,
    where $m_B$ is the number of edges of the induced
    subgraph $G_B$ consisting of all nodes in the balanced sink
    components and $X$.
Note that $m_B$ and $t_C$ are only related to subgraph $G_B$, which could be
    significantly smaller than $G$, and thus $O(t_Cm_B)$ could be much
    smaller than the time of naive iterations on the entire graph.
%% $m_B$ could be significantly smaller than $|E|$ and $t_C$
%%  could be significantly smaller than the convergence time of $P^t$ for
%%  the entire graph, and thus the iterative computation could be much faster
%%  than naive iterations on the entire graph.
Overall SVIM-L can be done in $O(|E|+\min(b n_Z^3+n_X^3, t_Cm_B))$
time.

%% file: ArXiv-sec-5EvaluationNew.tex
\section{Evaluation}\label{sec:evaluation}

{\color{black}In this section, we first use both synthetic datasets
and real social network datasets to demonstrate the efficacy of our
short-term and long-term seed selection schemes by comparing the
performances with four baseline heuristics. Then, we evaluate how
much the short-term and long-term influence can be improved by
taking the edge signs %, %i.e., both the friend and foe relations,
into consideration.}

\subsection{Performance comparison with baseline heuristics}\label{sec:baseline}
For different scenarios, we compare our SVIM-L and SVIM-S algorithms
with {\color{black}four} heuristics, i.e., (1) selecting seed nodes
with the highest weighted outgoing degrees (denoted by $d^{+}+d^{-}$
in the figures), (2) highest weighted outgoing positive degrees
(denoted by $d^{+}$), (3) highest differences between weighted
outgoing positive and negative degrees (denoted by
$d^{+}-d^{-}$){\color{black}, and (4) randomly selecting seed nodes
(denoted by ``Rand''), where in our evaluations, we run random seed
selection $1000$ times, and compare the average number of white
nodes between our algorithm and other heuristics}. Our evaluation
results demonstrate that our seed selection scheme can increase up
to $72\%$ long-term influence, and {\color{black}$145\%$} short-term
influence over other heuristics.

\subsubsection{Synthetic datasets}% 3000 6500

In this part, we generate synthetic datasets with different
structures to validate our theoretical results. %% We first introduce the models we

{\color{black}\noindent\textbf{Dataset generation model.} We
generate six types of signed digraphs, including balanced ergodic
digraphs, anti-balanced ergodic digraphs, strictly unbalanced
ergodic digraphs, weakly connected signed digraphs, disconnected
signed digraphs with ergodic components, and disconnected signed
digraph with weakly connected components (WCCs). All edges have unit
weights. %The graph configuration details are delegated
%to~\cite{VMreport} for brevity.
} The following are graph configuration details.

%\noindent{\emph{ergodic signed digraphs.}}
We first create an unsigned ergodic digraph $\bar{G}$ with $9500$
nodes, which has two ergodic components $\bar{G}_A$ and $\bar{G}_B$,
with $[3000,6500]$ nodes and $[3000,6500]\times 8$ random directed
edges, respectively. Moreover, there are $3000\times 8$ random
directed edges across $\bar{G}_{A}$ and $\bar{G}_{B}$. Ergodicity is
checked through a simple connectivity and
    aperiodicity check.
Given $\bar{G}$, a {\em balanced digraph} is obtained by assigning
all edges within $\bar{G}_{A}$ and $\bar{G}_{B}$ with positive
signs, and those across them with negative signs. Then, an
\emph{anti-balanced digraph} is generated by negating all edge signs
of the balanced ergodic digraph. To generate a \emph{strictly
unbalanced digraph}, we randomly assign edge signs to all edges in
$\bar{G}$ and make sure that there does not exist a balanced or
anti-balanced bipartition.

%\noindent{\emph{ergodic signed digraphs.}}
Moreover, we generated a \emph{disconnected signed digraph} and a
weakly connected signed digraph for our study. We first generate $5$
ergodic unsigned digraphs, $\bar{G}_1,\cdots,\bar{G}_5$ with $[500,
200, 800, 300, 2700]$ nodes and $[500, 200, 800, 300, 2700]\times 8$
edges, respectively. Then, we group $G_{23}=(G_2,G_3)$ and
$G_{45}=(G_4,G_5)$ to form two ergodic balanced digraphs, and
generate a strictly unbalanced ergodic digraph $G_1$ by randomly
assigning signs to edges in ${\bar{G}_1}$. Three disconnected
components $G_1, G_{23}, G_{45}$ together form a disconnected signed
digraph. To form a \emph{weakly connected signed digraph}, we place
in total $3000$ random direct edges from $G_1$ to the balanced
ergodic components $G_{23}$ and $G_{45}$, where the nodes in
subgraph $G_1$ only have outgoing edges to $G_{23}$ and $G_{45}$.
Moreover, we combine the above generated balanced ergodic digraph
and the weakly connected signed digraph together forming a larger
\emph{disconnected signed digraph, with the weakly connected signed
digraph as a component}.

Fig.~\ref{fig:eps1}-Fig.~\ref{fig:eps6} present the evaluation
results for one set of digraphs, where we observe that all digraphs
we randomly generated exhibit consistent results.
%For each balance structure setting, we generate $1000$ signed
%digraphs, where below we only present the results for one set of
%digraphs, and all other graphs have similar results.
%The average
%results obtained on these $1000$ simulations are used for
%comparison.
Our tests are conducted using Matlab on a standard PC server.
%\vspace{-2mm}

\begin{figure}[!htb]
    \centering
    \begin{minipage}[t]{5.35cm}
    \begin{center}
    \includegraphics[width=2.0in]{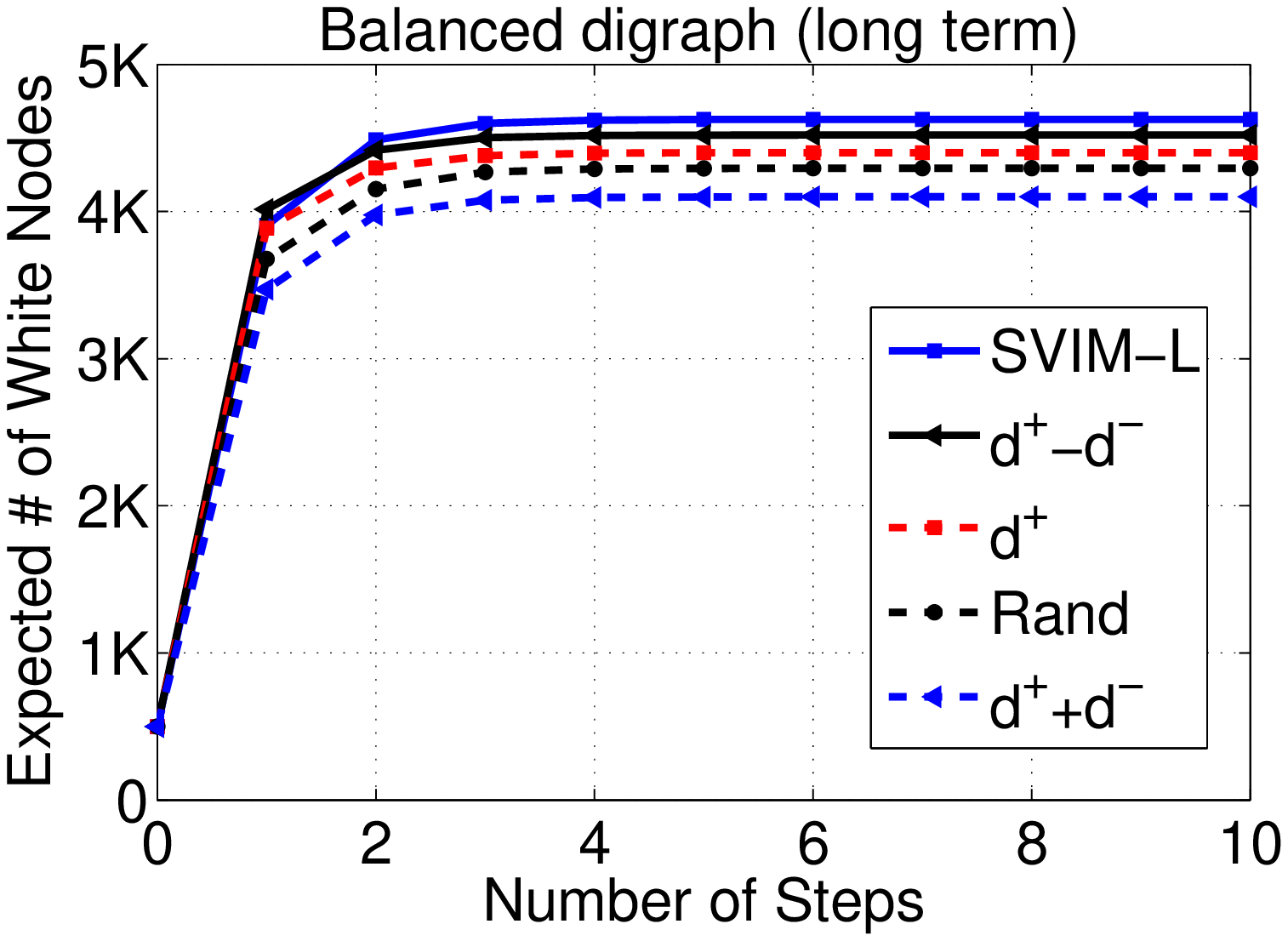}
    \caption{$G$ is balanced}\label{fig:eps1}
    \end{center}
    \end{minipage}
   \centering
    \begin{minipage}[t]{5.35cm}
    \begin{center}
    \includegraphics[width=2.0in]{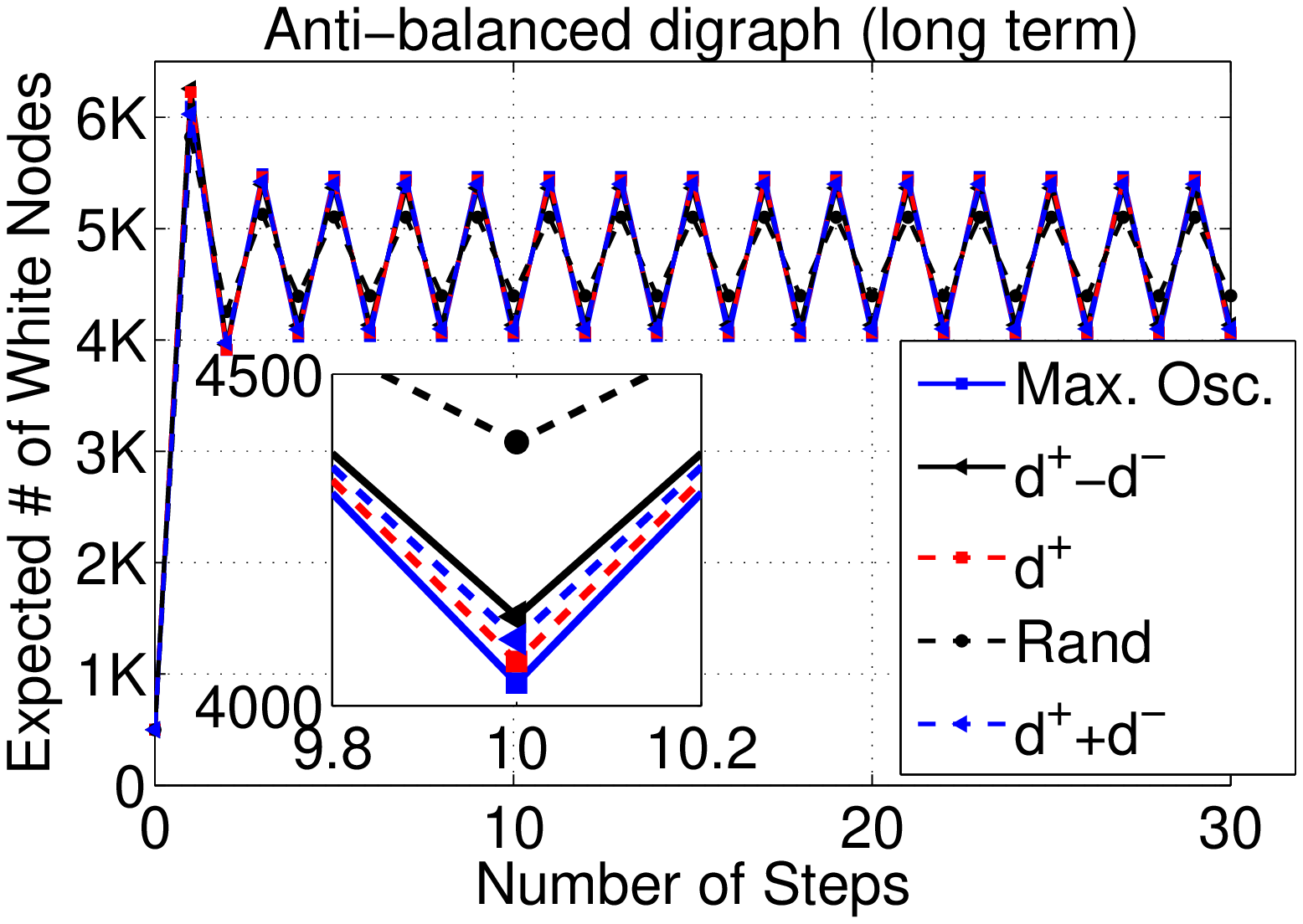}
    \caption{$G$ is anti-balanced}\label{fig:eps2}
    \end{center}
    \end{minipage}
    \begin{minipage}[t]{5.35cm}
    \begin{center}
    \includegraphics[width=2.0in]{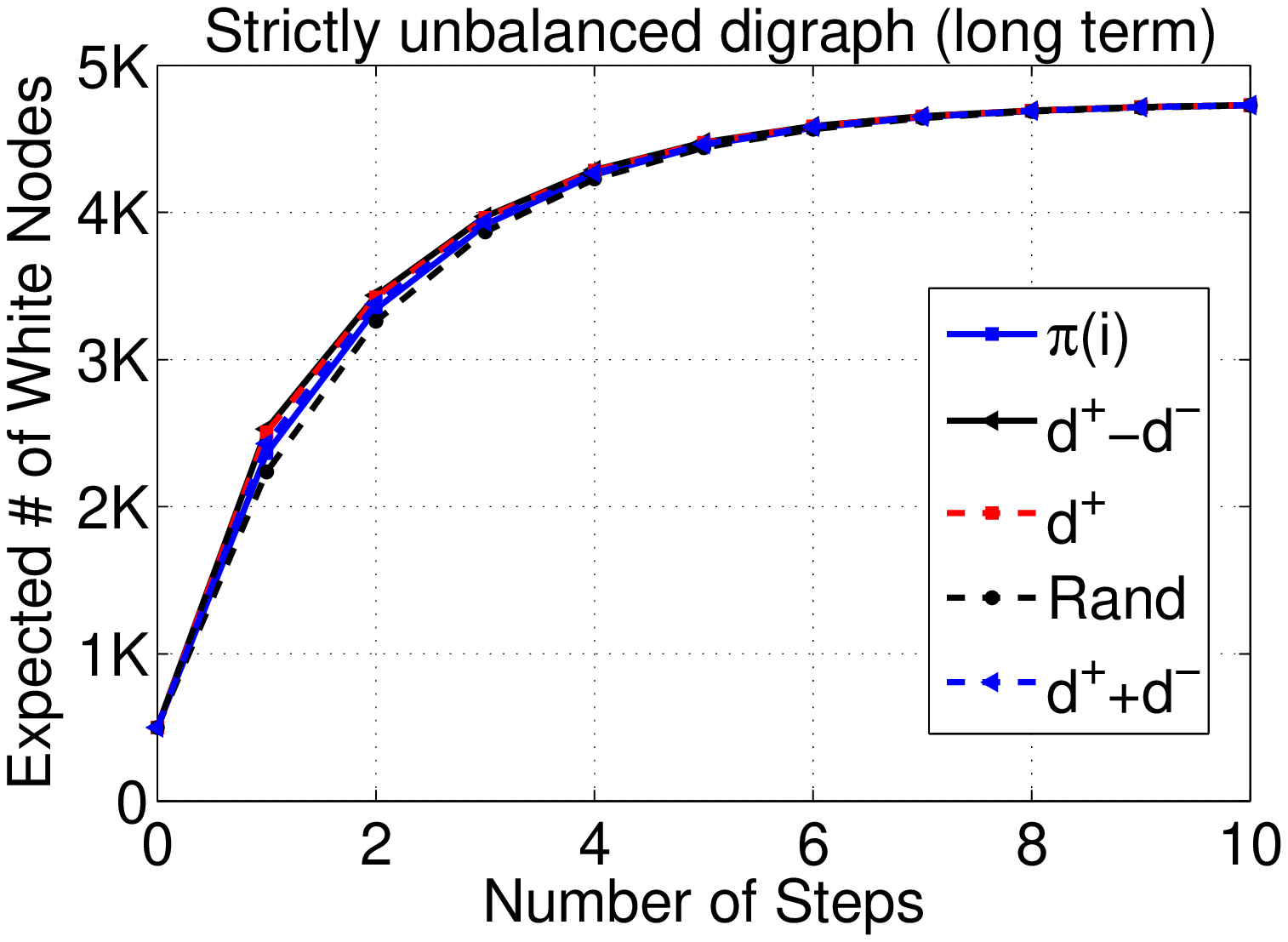}
    \caption{$G$ is strictly unbalanced}\label{fig:eps3}
    \end{center}
    \end{minipage}
   \centering
    \begin{minipage}[t]{5.35cm}
    \begin{center}
    \includegraphics[width=2.0in]{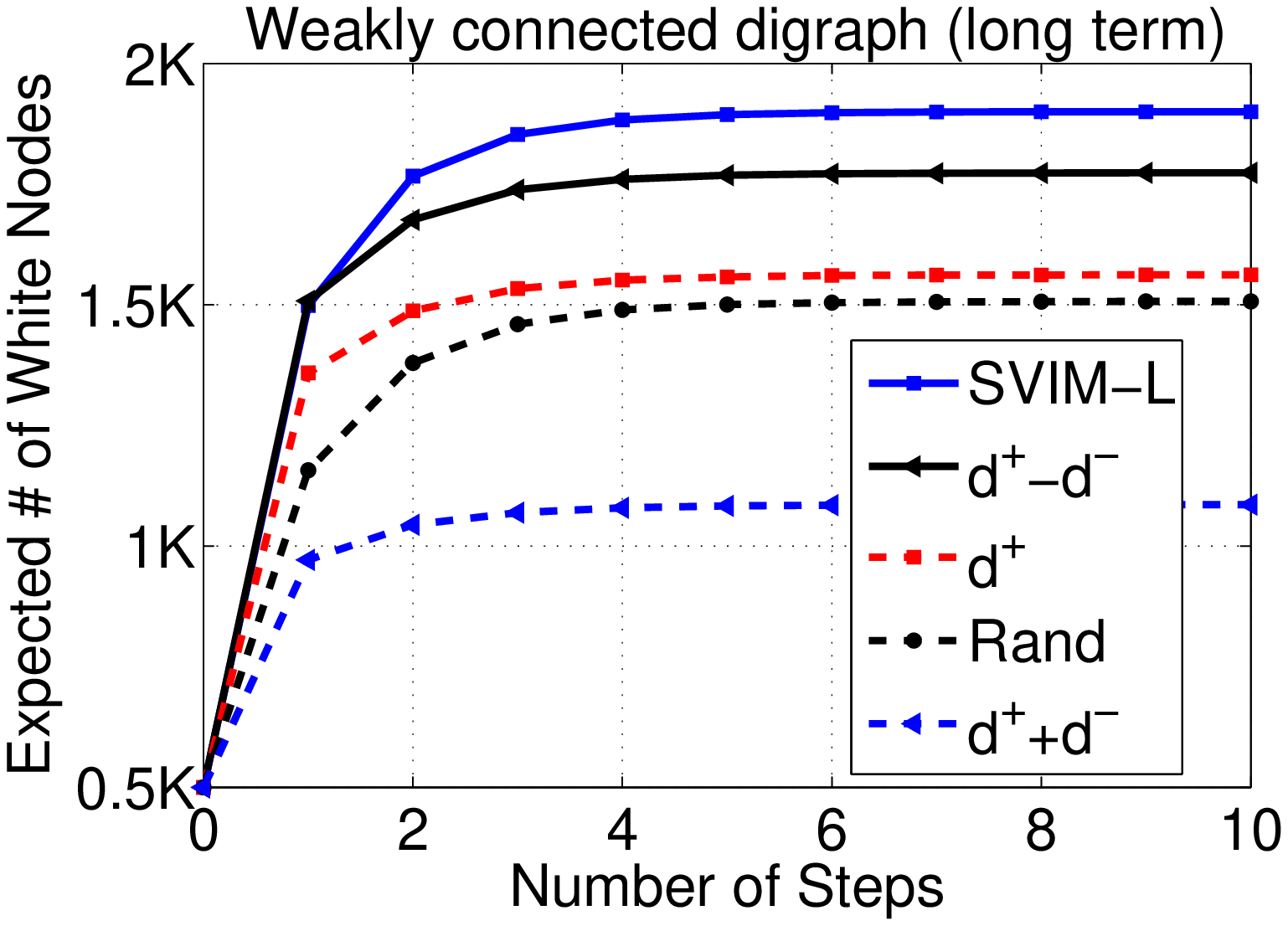}
    \caption{$G$ is weakly connected}\label{fig:eps4}
    \end{center}
    \end{minipage}
    \begin{minipage}[t]{5.35cm}
    \begin{center}
    \includegraphics[width=2.0in]{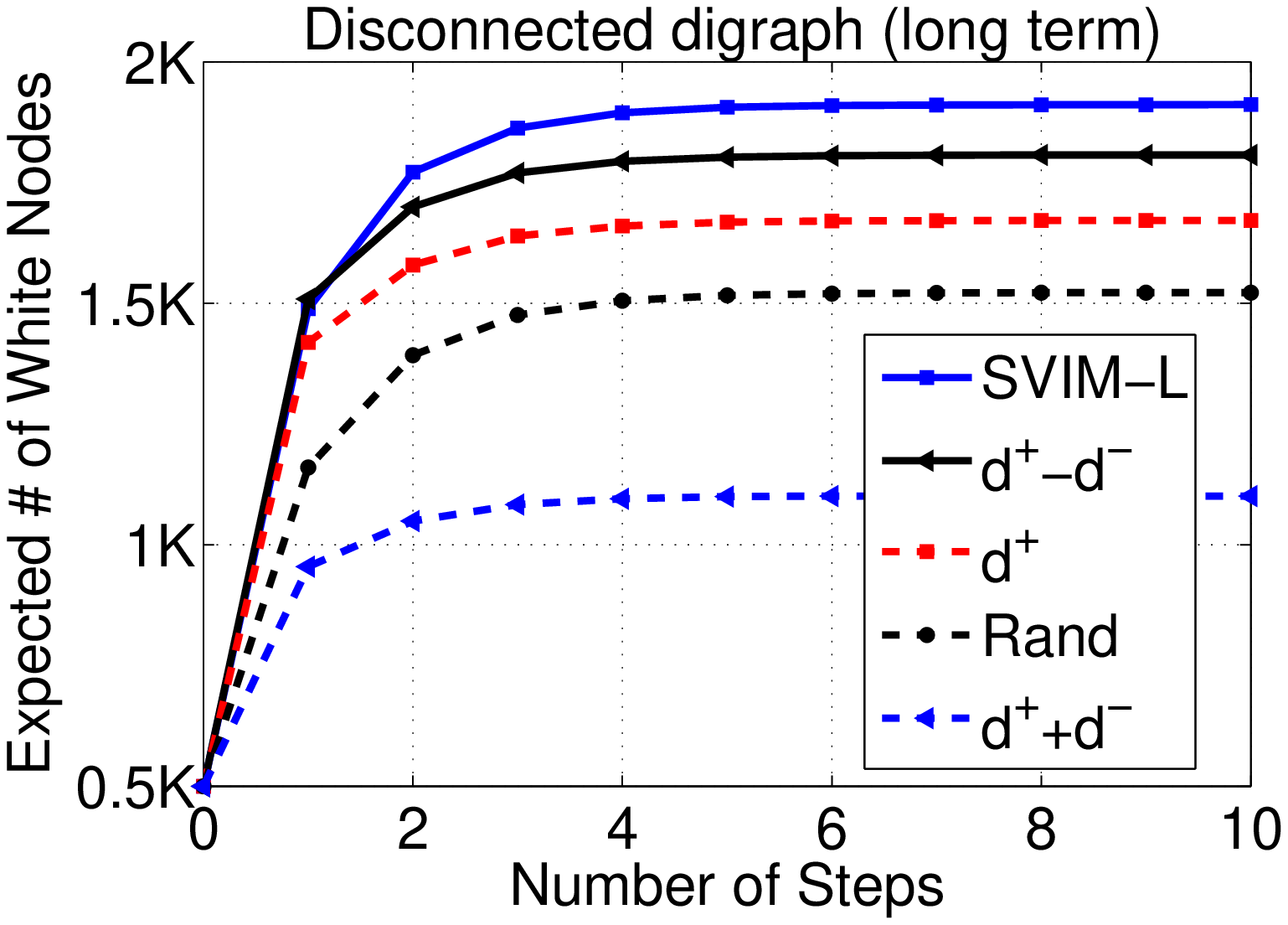}
    \caption{$G$ is disconnected}\label{fig:eps5}
    \end{center}
    \end{minipage}
   \centering
    \begin{minipage}[t]{5.35cm}
    \begin{center}
    \includegraphics[width=2.0in]{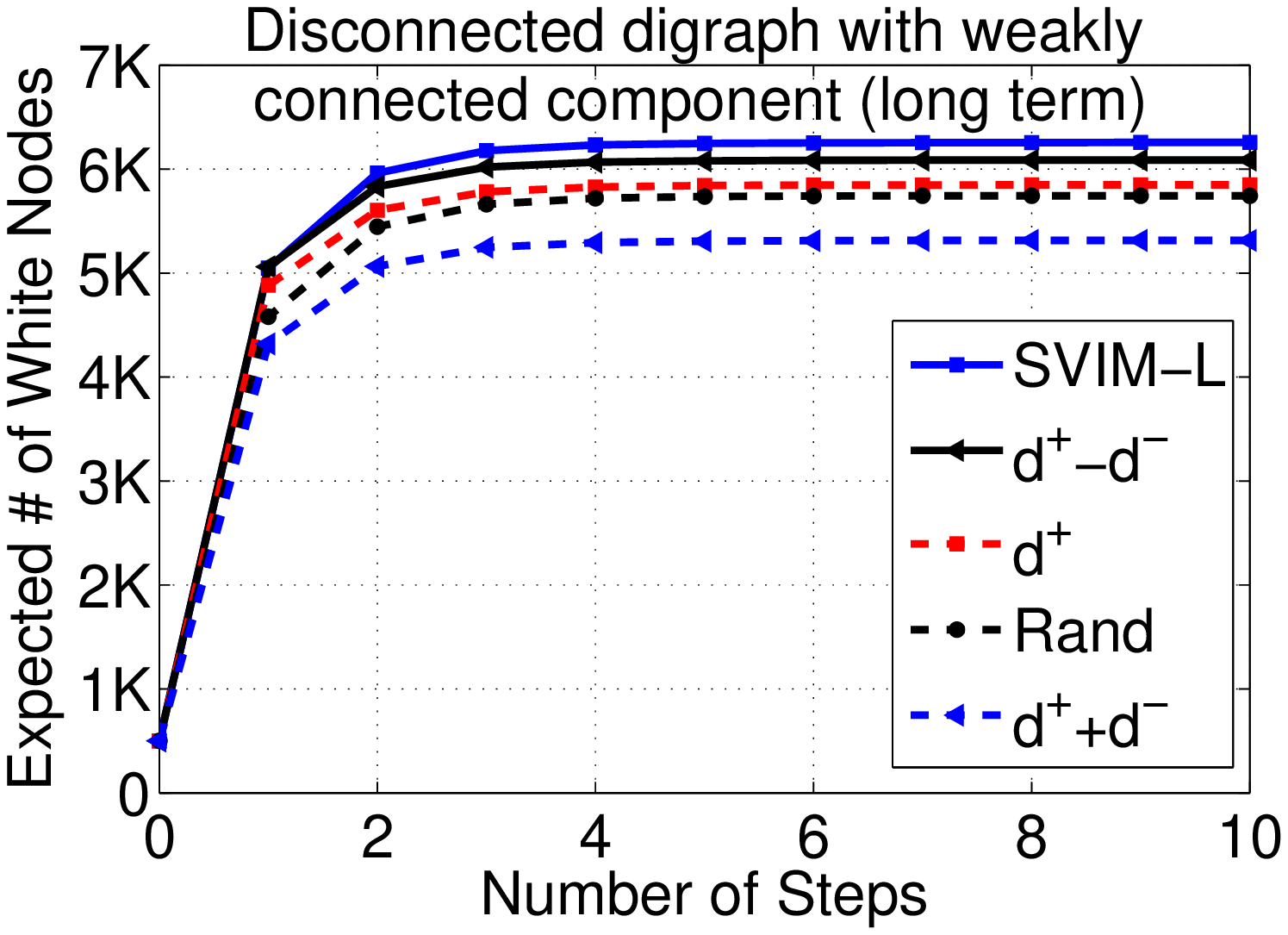}
    \caption{$G$ is disconnected with WCC}\label{fig:eps6}
    \end{center}
    \end{minipage}
  \vspace*{-0.5cm}
\end{figure}

\begin{figure*}[!htb]
    \centering
    \begin{minipage}[t]{6.7cm}
    \begin{center}
    \includegraphics[width=2.25in]{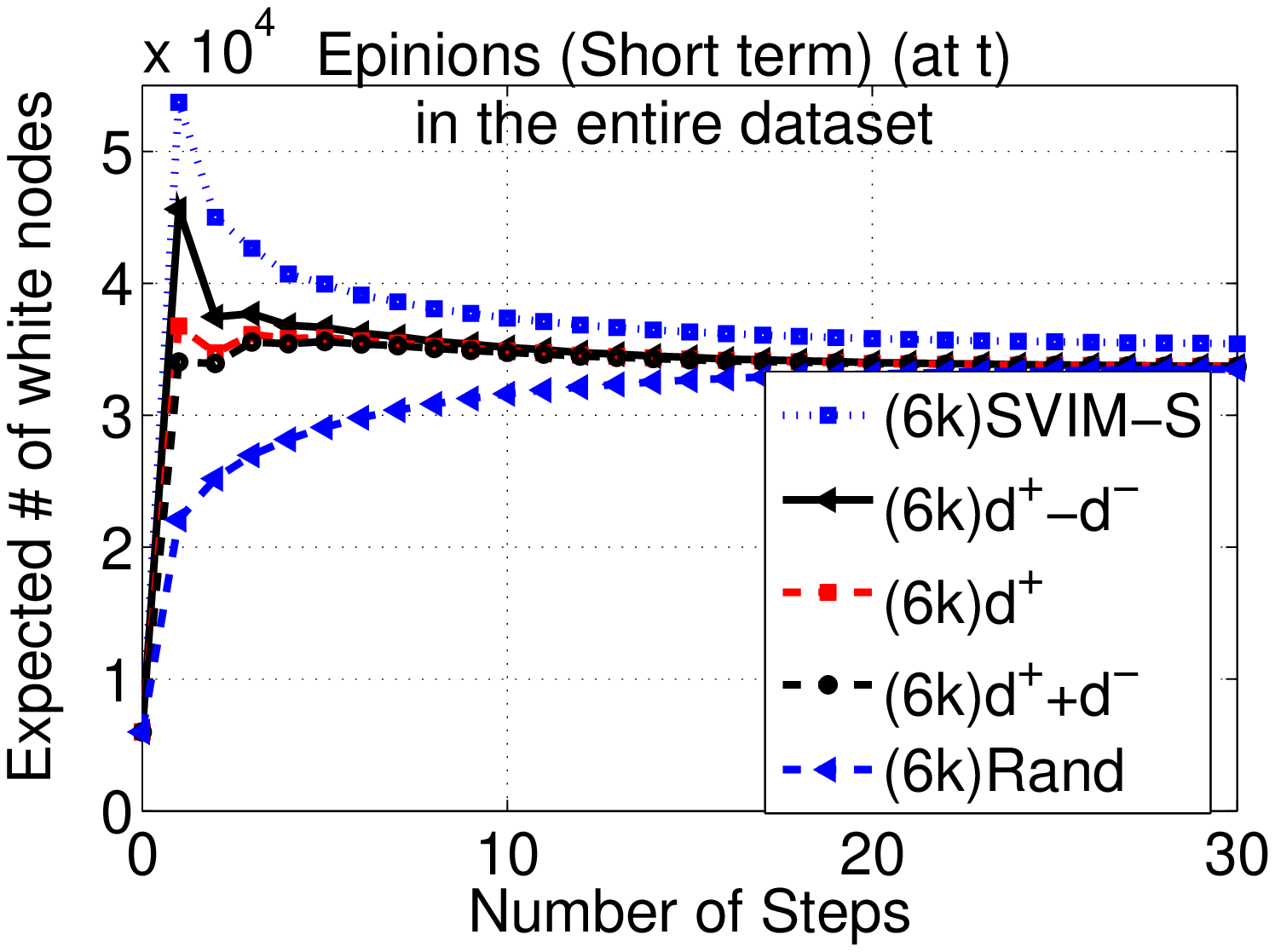}
             \vspace*{-0.3cm}
    \caption{\color{black}Instant influence in Epinions data with $k=6$k}\label{fig:epsEPSin1entire}
%    \caption{Maximize instant influence for each $t$}\label{fig:epsEPSin1entire}
    \end{center}
    \end{minipage}
    \centering
    \begin{minipage}[t]{6.7cm}
    \begin{center}
    \includegraphics[width=2.25in]{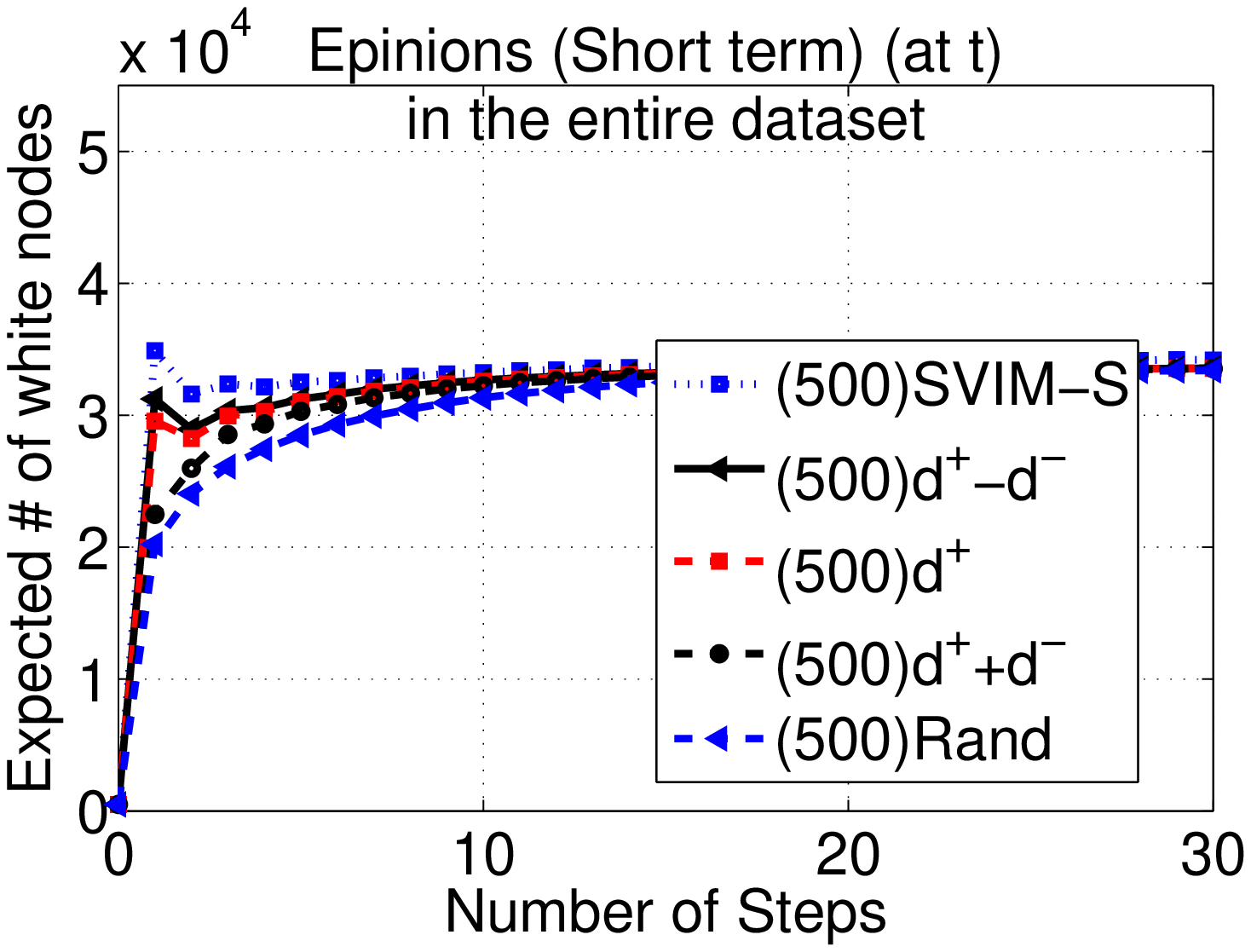}
             \vspace*{-0.3cm}
    \caption{\color{black}Instant influence in Epinions data with $k=500$}\label{fig:epsEPSin1entire500}
    \end{center}
    \end{minipage}\\
    \centering
    \begin{minipage}[t]{6.7cm}
    \begin{center}
    \includegraphics[width=2.25in]{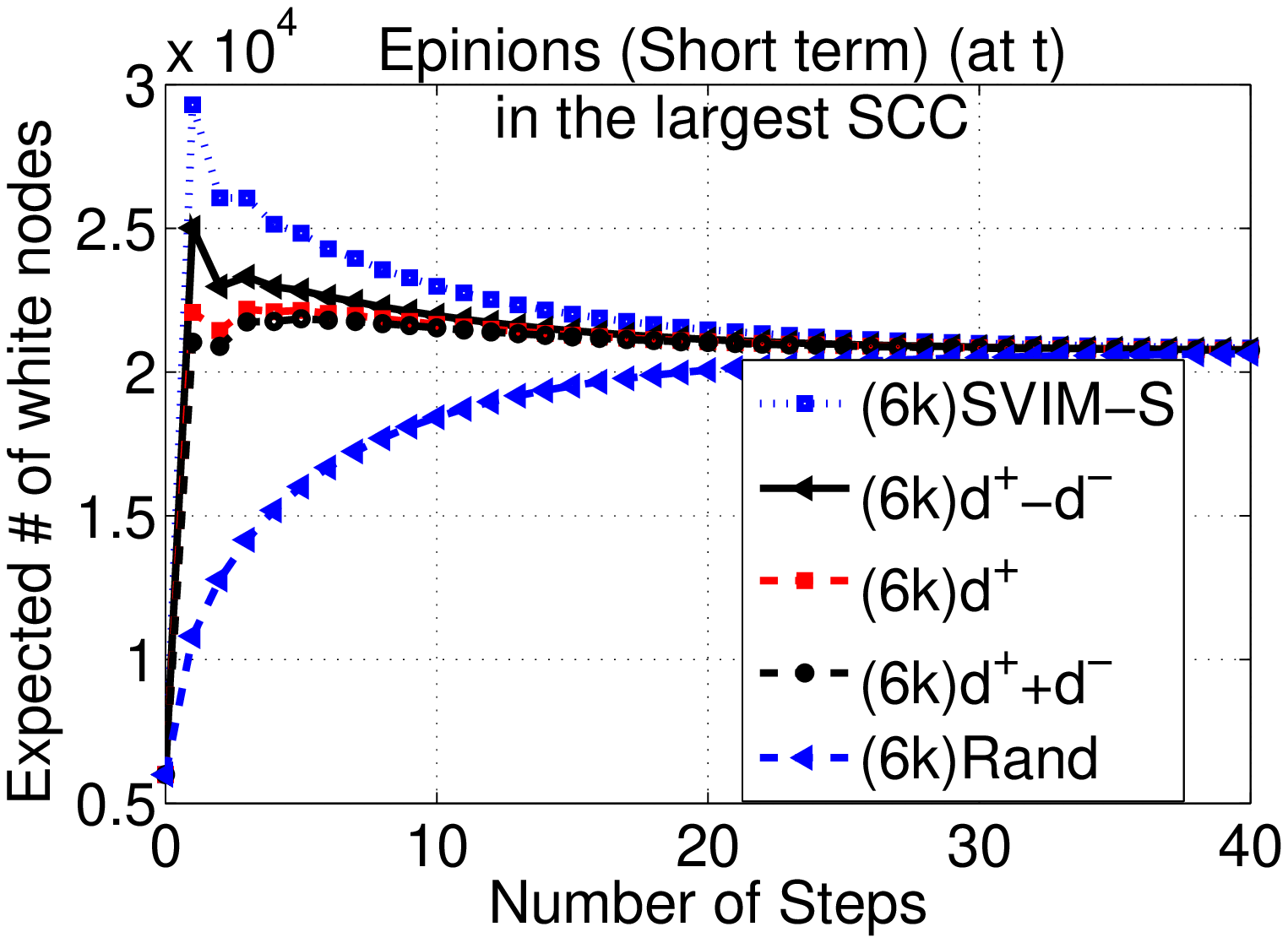}
             \vspace*{-0.3cm}
    \caption{\color{black}Instant influence in SCC with $k=6$k}\label{fig:epsEPSin1}
    \end{center}
    \end{minipage}
    \centering
    \begin{minipage}[t]{6.7cm}
    \begin{center}
    \includegraphics[width=2.25in]{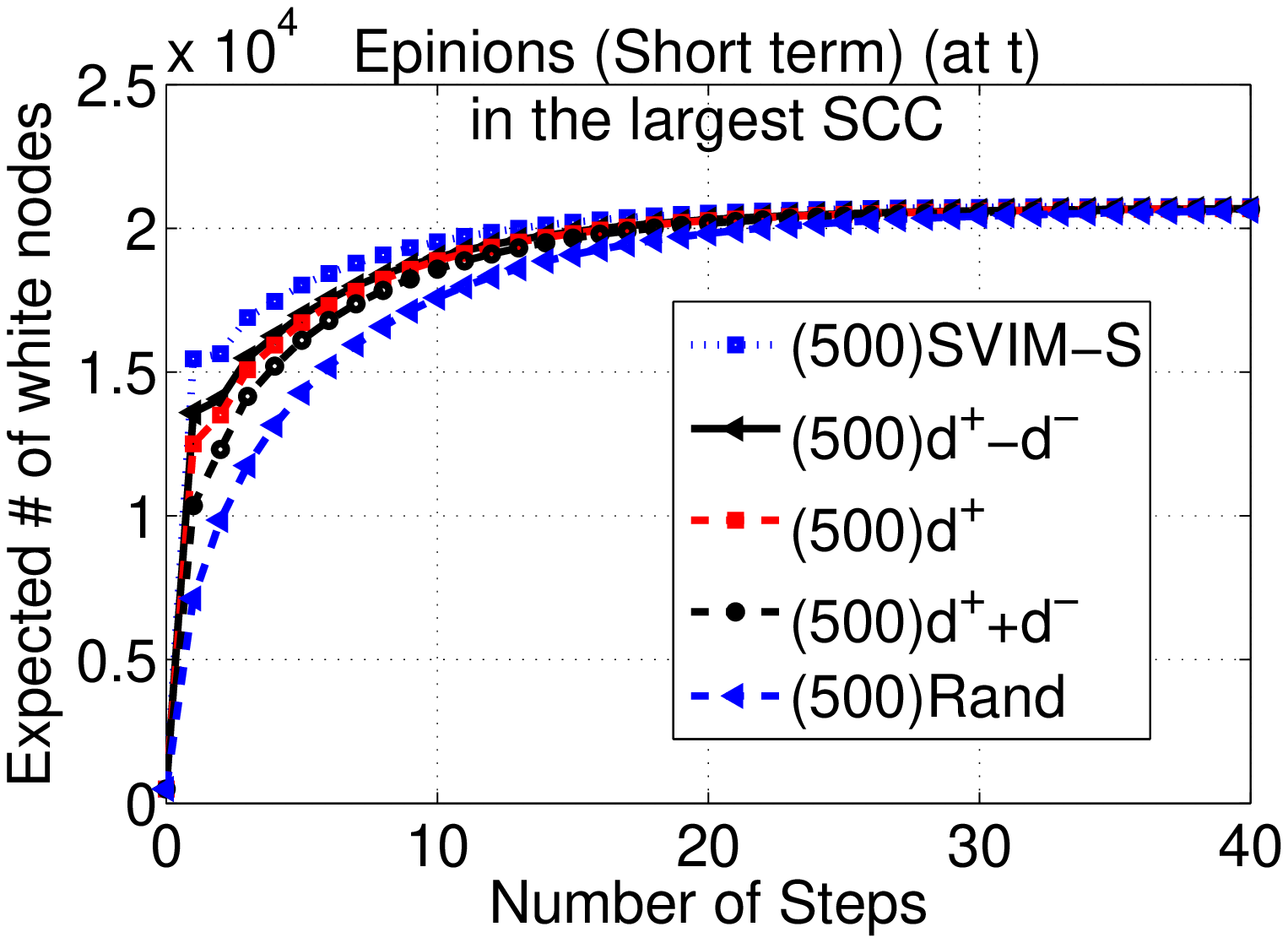}
             \vspace*{-0.3cm}
    \caption{\color{black}Instant influence in SCC with $k=500$}\label{fig:epsEPSin1500}
    \end{center}
    \end{minipage}
    \centering
  \vspace*{-0.3cm}
\end{figure*}
\begin{figure*}[!htb]
    \centering
    \begin{minipage}[t]{6.7cm}
    \begin{center}
    \includegraphics[width=2.25in]{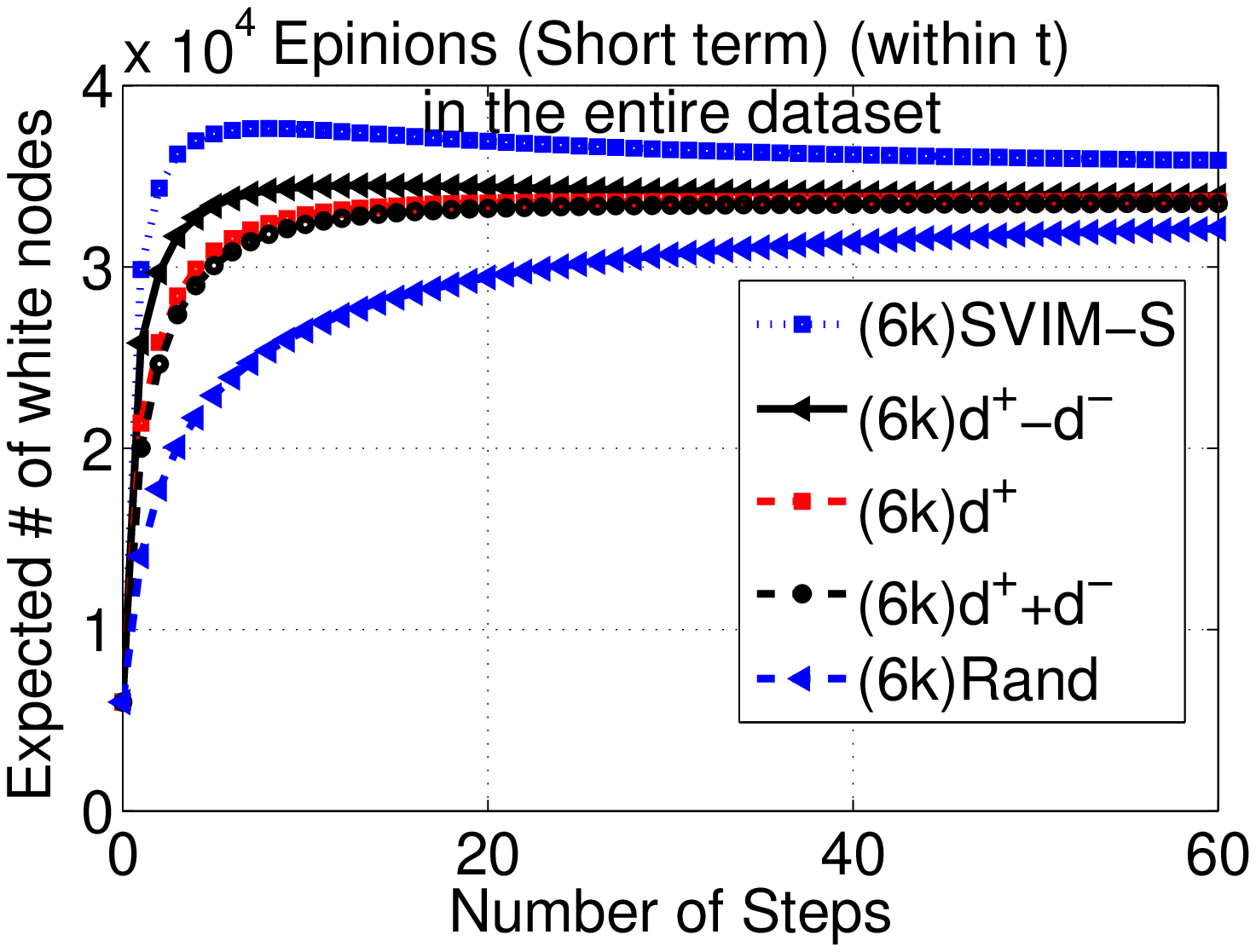}
    \vspace*{-0.3cm}
    \caption{\color{black}Average influence in Epinions data with $k=6$k}\label{fig:epsEPSin3entire}
    \end{center}
    \end{minipage}
    \centering
    \begin{minipage}[t]{6.7cm}
    \begin{center}
    \includegraphics[width=2.25in]{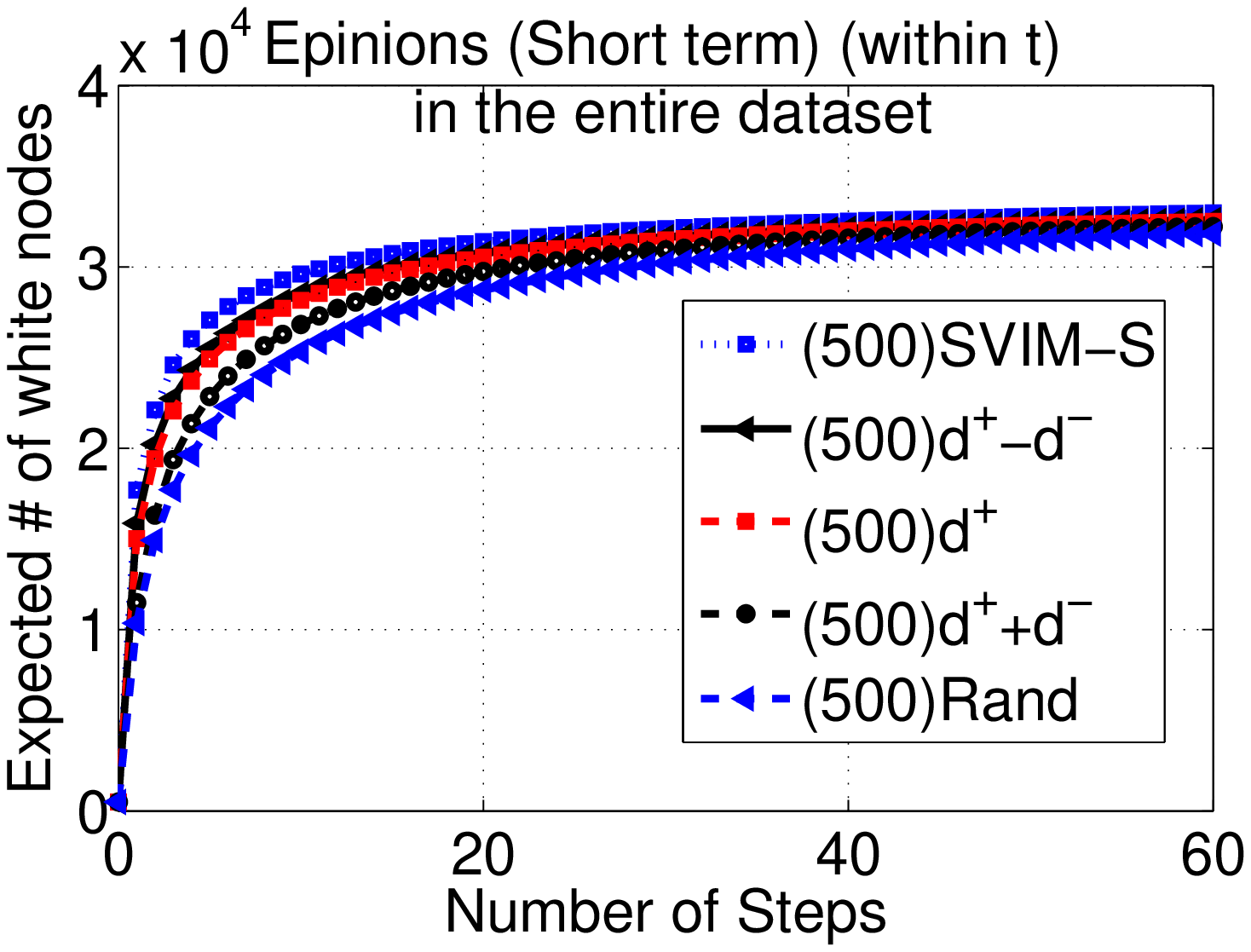}
    \vspace*{-0.3cm}
    \caption{\color{black}Average influence in Epinions data with $k=500$}\label{fig:epsEPSin3entire500}
    \end{center}
    \end{minipage}\\
    \centering
    \begin{minipage}[t]{6.7cm}
    \begin{center}
    \includegraphics[width=2.25in]{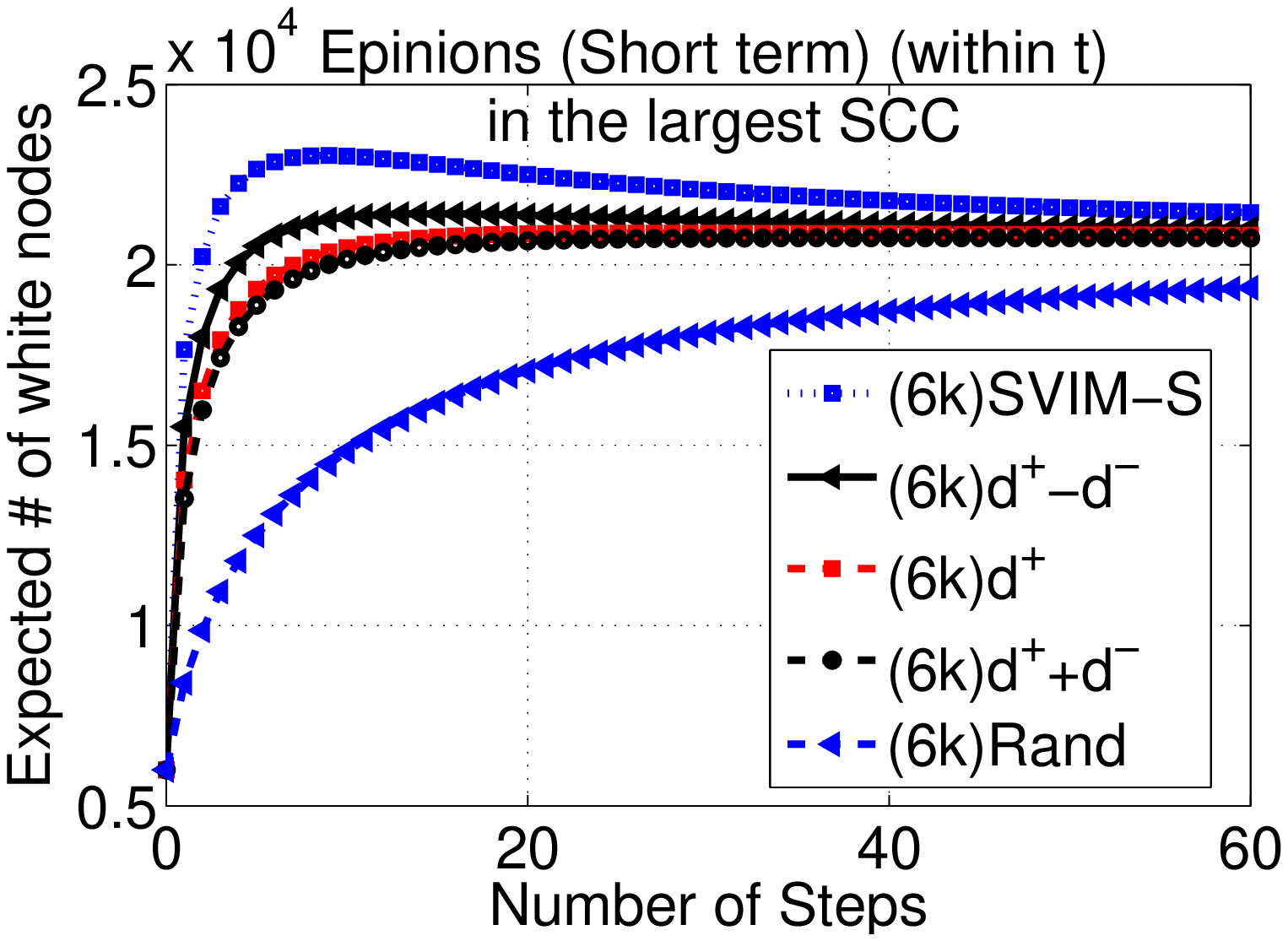}
    \vspace*{-0.3cm}
    \caption{\color{black}Average influence in SCC with $k=6$k}\label{fig:epsEPSin3}
    \end{center}
    \end{minipage}
    \centering
    \begin{minipage}[t]{6.7cm}
    \begin{center}
    \includegraphics[width=2.25in]{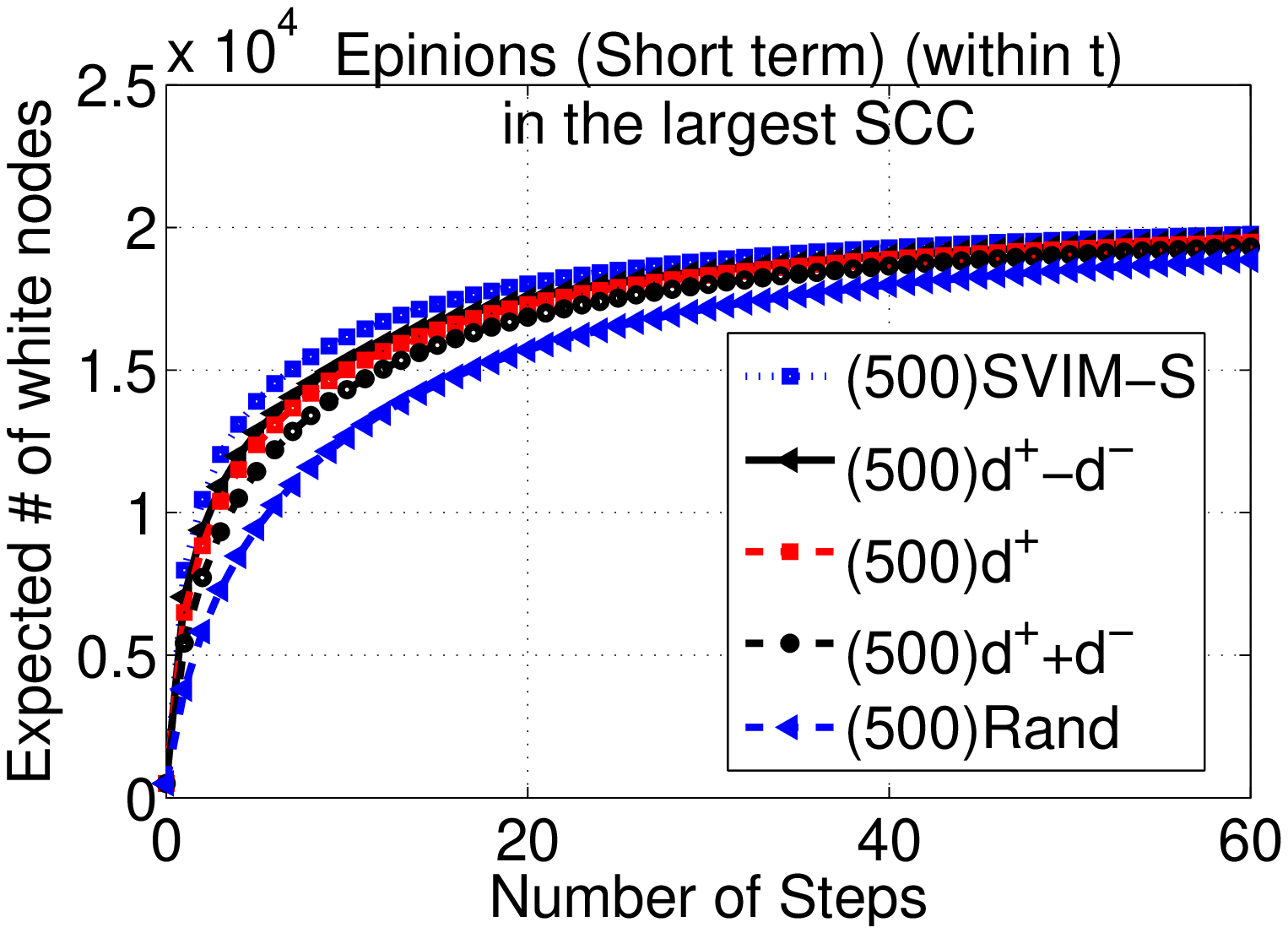}
    \vspace*{-0.3cm}
    \caption{\color{black}Average influence in SCC with $k=500$}\label{fig:epsEPSin3500}
    \end{center}
    \end{minipage}
    \centering
  \vspace*{-0.5cm}
\end{figure*}

\noindent\textbf{{Long-term influence maximization.}} In the
evaluations, we set the influence budget as $k=500$,
{\color{black}and compare the average numbers of white nodes over
steps between our algorithm and other heuristics.}
%which
%means initially at most $500$ white seeds could be selected.
Fig.~\ref{fig:eps1} shows that in the balanced ergodic digraph,
SVIM-L algorithm achieves the highest long-term influence over other
heuristics. {\color{black} When applying a heuristic seed selection
scheme, denoted by ${\tt H}$, $f_t^{\tt
 H}$ represents the number of white nodes at step $t(\geq 1)$. Similarly, denote
$f_t^{{\tt SVIM}}$ as the number of white nodes at step $t(\geq 1)$
for SVIM algorithm. We consider $\Delta f_t({\tt SVIM, H})=(f_t^{\tt
SVIM}-f_t^{\tt H})/f_t^{\tt H}$ as the influence increase of SVIM
over the heuristic algorithm ${\tt H}$ at step $t$. The maximum
influence increase is the maximum $\Delta f_t({\tt SVIM}, \cdot)$
among all steps ($t\geq 1$) and all heuristics. Hence, in
Fig.~\ref{fig:eps1}, we see that our SVIM-L algorithm outperforms
all other heuristics. Especially, a maximum of $14\%$ influence
increase is observed for $t\geq 4$ with $4.68k$ and $4.1k$ white
nodes for SVIM-L and random selection scheme, respectively. In the
rest of this section, we will use the maximum influence increase as
a metric to illustrate the efficacy of our SVIM algorithm.}
% with up to {\color{black}$14\%$ influence increase, which
%occurs after the fourth step}.
Fig.~\ref{fig:eps2} shows the clear oscillating behavior on the
    anti-balanced ergodic digraph, and the average influence is the
    same for all algorithms.
%In fact, we also designed an algorithm to maximize the oscillation
%in this
%    case, %{\color{black}(Report removed.)}%(See details in~\cite{VMreport}),
%    but due to space constraint we omit it in this paper.
The inset shows that our algorithm (denoted as ``Max. Osc.'') indeed
provides the largest
    oscillation.
%% Fig.~\ref{fig:eps2}
%% shows that in the anti-balanced digraph, our scheme and all
%% heuristics result in the same long term average influence, around
%% $4750$, which confirms our theoretical results in
%% Lemma~\ref{thm:ltdyn}, i.e.,
%% $f(x_0)=\textbf{1}^T\textbf{1}/2=|V|/2=4750$. Moreover,
%% by choosing the initial seed
%% nodes based on Remark~\ref{rmk:Rmax}, the influence oscillation
%% strength obtained is the highest over other heuristics.
Fig.~\ref{fig:eps3} shows the results in strictly unbalanced graph
case, where the long-term
    influences of all algorithms converge to $4750=|V|/2$,
    which matches Theorem~\ref{thm:ltdyn}.
%%  indicates that the long term influence of
%% voter model on strictly unbalanced digraph is always $|V|/2$, which
%% is independent from the initial seed selection. We generate initial
%% seed sets using three heuristics, as well as stationary distribution
%% $\pi(i)$'s, and the results of all heuristics, shown in
%% Fig.~\ref{fig:eps3}, converge to $4750$ white nodes, which confirms
%% Theorem~\ref{thm:ltdyn}.
{\color{black} Fig.~\ref{fig:eps4} and Fig.~\ref{fig:eps5} show that
SVIM-L algorithm performs the best, and it generates $5.6\%-72\%$
long-term influence increases after the sixth step over other
heuristics in the weakly connected signed digraph and the
disconnected signed digraph.} {\color{black}Fig.~\ref{fig:eps6}
shows that in a more general signed digraph, which consists of a
weakly connected signed component and a balanced ergodic component,
SVIM-L algorithm outperforms all other heuristics with  up to $17\%$
more long term influence, which occurs for $t\geq 4$}. In general,
we see that for weakly connected and disconnected digraphs,
    SVIM-L has larger winning margins over all other heuristics than the case
    of balanced ergodic digraphs (Fig.~\ref{fig:eps4}--\ref{fig:eps6}
    vs. Fig.\ref{fig:eps1}).
We attribute this to our accurate computation of influence
contribution
    in the more involved weakly connected and disconnected digraph
    cases. Moreover, in all cases, the dynamics converge very fast, i.e., in only a few steps,
    which indicates that the convergence time of voter model on these random graphs are very small.
\vspace*{-0.2cm}
\begin{table*}[!htb]
\centering \caption{Statistics of Epinions and Slashdot
datasets}\label{tab::t1} {\small
\begin{tabular}{|c c c||c c c|}
\hline Statistics & Epinions & Slashdot & Statistics & Epinions & Slashdot\\
\hline \hline \# of nodes & $131580$ & $77350$ & \# of nodes in largest SCC & $41441$ & $26996$ \\
 \# of edges & $840799$ & $516575$ & \# of edges in largest SCC & $693507$ & $337351$ \\
 \# of positive edges & $717129$ & $396378$ & \# of positive edges in largest SCC & $614314$ & $259891$\\
 \# of negative edges & $123670$ & $120197$ & \# of negative edges in largest SCC & $79193$ & $77460$ \\
  & & & \# of strongly connected components & $88361$ & $49209$\\
\hline
\end{tabular}}\vspace*{-0.4cm}
\end{table*}

\subsubsection{Real datasets}

%In this section, we use real datasets, namely, Epinions and Slashdot
%datasets, to validate our theoretical results and our SVIM
%algorithm.

{\color{black} We conduct extensive simulations using real datasets,
such as Epinions and Slashdot datasets, to validate our theoretical
results and evaluate the performance of our SVIM algorithm.}

\noindent\textbf{Epinions Dataset.} Epinions.com~\cite{EP} is a
consumer review online social site, where users can write reviews to
various items and vote for or against other
    users.
The signed digraph is formed with positive or negative directed edge
$(u,v)$
    meaning that $u$ trusts or distrusts $v$.
%% The reviews will help the users decide whether or not purchase the
%% items. Members of the site can decide whether to ``trust'' each
%% other. All the trust/distrust relationships interact and form the
%% signed the online social network, which is then combined with review
%% ratings to determine which reviews are shown to the user.
The statistics are shown in Table~\ref{tab::t1}.
%% Epinions dataset has $131580$ nodes and $840799$ directed edges.
%% There are in total $88361$ disconnected signed components. There is
%% a giant strongly connected components with $41441$ nodes, and all
%% other components are with at most size $15$, where most of there are
%% isolated nodes.
We compare our short-term SVIM-S algorithm with {\color{black}four}
heuristics, i.e., $d^{+}+d^{-}$, $d^{+}$, $d^{+}-d^{-}$
{\color{black}and random seed selection}, on the entire Epinions
digraph as well as the largest strongly connected component (SCC).

%{fig:epsEPSin1entire}{fig:epsEPSin1500}{fig:epsEPSin3entire}{fig:epsEPSin3500}

{\color{black}Our tests are conducted on both Epinions dataset and
its largest strongly connected component (SCC), where the largest
SCC is ergodic and strictly unbalanced. We first look at the
comparison of instant influence maximization (\emph{at} step $t$)
among various seed selection schemes.}
%We first look at the seed selection schemes for maximizing the
%instant influence \emph{at} step $t$.
Fig.~\ref{fig:epsEPSin1entire}-\ref{fig:epsEPSin1500} shows the
expected maximum instant influence at each step by different
methods. Note that since the initial seeds selected by SVIM-S
algorithm hinge on $t$, the values on the curve of our selection
scheme are associated with different optimal initial seed sets. On
the other hand, the seed selections of other heuristics are
independent to $t$, thus the corresponding curves represent the same
initial seed sets. {\color{black}We choose the budget as $500$ and
$6000$ in our evaluations, i.e., selecting at maximum $500$ or
$6000$ initial white seeds.} {\color{black}From
Fig.~\ref{fig:epsEPSin1entire}-\ref{fig:epsEPSin1500},
%that %and $6k$ seeds, the
SVIM-S algorithm consistently performs better, and in some cases,
e.g., Fig.~\ref{fig:epsEPSin1}, it generates $16\%-145\%$ more
influence than other heuristics at step $1$.}
% {\color{black}where SVIM-S algorithm has up
%to $145\%$ more
%influence over other heuristics. %for $t=2$
%(Fig.~\ref{fig:epsEPSin1entire})}.
%Fig.~\ref{fig:epsEPSin2entire} shows how the influence evolves over
%steps, given that the goal is to maximize the influence at step $3$,
%where SVIM-S algorithm achieves up to $27\%$ more influence at step
%$3$.

Next we compare the seed selection schemes for maximizing the
average influence \emph{within} the first $t$ steps.
Fig.~\ref{fig:epsEPSin3entire}-\ref{fig:epsEPSin3500} show the
expected maximum average influence within the first $t$ steps by
different methods. Again, the values on the curve of SVIM-S
algorithm are
associated with different initial seed sets. %On the other hand, the
%seed selection of other heuristics are independent to $t$, thus the
%corresponding curves represent the same initial seed sets.
{\color{black}Fig.~\ref{fig:epsEPSin3entire}-\ref{fig:epsEPSin3500}
show that with different budgets, i.e., $500$ and $6000$ seeds,
SVIM-S algorithm performs better than all other heuristics, where in
Fig.~\ref{fig:epsEPSin3} a maximum of $64\%$ more influence is
achieved at $t=8$.}
%Fig.~\ref{fig:epsEPSin4entire} shows how the expected average
%influence evolves over steps for $t=3$, where SVIM-S algorithm
%achieves up to $36\%$ more average influence over the first $3$
%steps.
Moreover, in all these figures, we observe that our seed selection
scheme results in the highest long-term influence over other
heuristics.
%{\color{black}we show similar results as that in the
%entire Epinions dataset, where SVIM-S algorithm always outperforms
%other heuristics. The slight difference is that the performance of
%algorithms converges
%    faster than on the entire graph, because SCC has better connectivity than
%    the entire graph.}

{\color{black}Moreover, from
Fig.~\ref{fig:epsEPSin1entire}-\ref{fig:epsEPSin3500}, we observe
that as $t$ increases, the influences (i.e., the expected number of
white nodes), for SVIM-S and all heuristics except for random seed
selection schedule, increase for small $t$'s, and then decrease and
converge to the stationary state. In contrast, from
Fig.~\ref{fig:eps1}-\ref{fig:eps6}, the influence increases
monotonically with $t$. This happens because Epinions dataset (as
well as many real network datasets) has large portion (around
$80\%$) of nodes in the non-sink components, where to maximize the
long-term influence, only nodes in sink components should be
selected, which governs the long-term influence dynamics of the
whole graph,
%since after a long run, all nodes
%from non-sink components would be affected by those sink nodes,
namely, sink nodes have higher long-term influence contributions.
However, for short-term influence maximization, nodes with higher
chances to influence more nodes in a few steps generally have large
number of incoming links, which %in either sink or non-sink
%components
are able to influence a large number of nodes in either sink or
non-sink components in a short period of time. Hence, in signed
digraphs with large non-sink component, given a sufficiently large
budget, the short-term influence can definitely outnumber the
long-term influence. Our evaluations confirm this explanation.} %, i.e., a
%short term influence boost-up happens when the non-sink component is
%huge and the seed budget is relatively large, e.g., 6K in Epinions
%dataset.
{\color{black}This interesting observation also leads to a problem
that given a budget $k$, how to find the optimal time step $t$ that
generates the largest influence among all possible $t$'s. We leaves
this problem as our future work.}

\noindent\textbf{Slashdot Dataset.} Slashdot.org~\cite{Sd} provides
a discussion forum on various technology-related topics, where
members can submit their stories, and comment on other members'
stories. Its Slashdot Zoo feature allows members to tag each other
as friends or foes, which in turn forms a signed online social
network. The network was collected on $6$-th November
2008~\cite{leskovec2010signed} and the statistics are shown in
Table~\ref{tab::t1}.

%We evaluate our SVIM-S algorithm on the entire slashdot dataset
%    (Fig.~\ref{fig:epsSdSin1entire}-\ref{fig:epsSdSin4entire}) and
%    its largest strongly connected component ({\color{black}Fig.~\ref{fig:epsSdSin1}-Fig.~\ref{fig:epsSdSin4}}), %which is delegated to our technical report~\cite{VMreport}, due to the limited space),
%    respectively, and our results show that our SVIM-S algorithm
%    performs the
%    best among all methods tested, especially in the early steps.

%(Fig.~\ref{fig:epsSdSin1entire}-\ref{fig:epsSdSin4entire}) ({\color{black}Fig.~\ref{fig:epsSdSin1}-Fig.~\ref{fig:epsSdSin4}}), %which is delegated to our technical report~\cite{VMreport}, due to the limited space),

We evaluate instant influence and average influence of our SVIM-S
algorithm on the entire slashdot dataset and its largest strongly
connected component, respectively. Our results for $k=6000$ are
presented in
Fig.~\ref{fig:epsSdSin1entire}-Fig.~\ref{fig:epsSdSin3},
%Fig.~\ref{fig:epsSdSin1entire}-\ref{fig:epsSdSin4entire}-Fig.~\ref{fig:epsSdSin1}-Fig.~\ref{fig:epsSdSin4}},
which show that our SVIM-S algorithm performs the best among all
methods tested, especially in the early steps. When changing the
budget $k$, similar results were obtained, where we omitted them
here for brevity.

\begin{figure*}[!htb]
    \centering
    \begin{minipage}[t]{6.7cm}
    \begin{center}
    \includegraphics[width=2.25in]{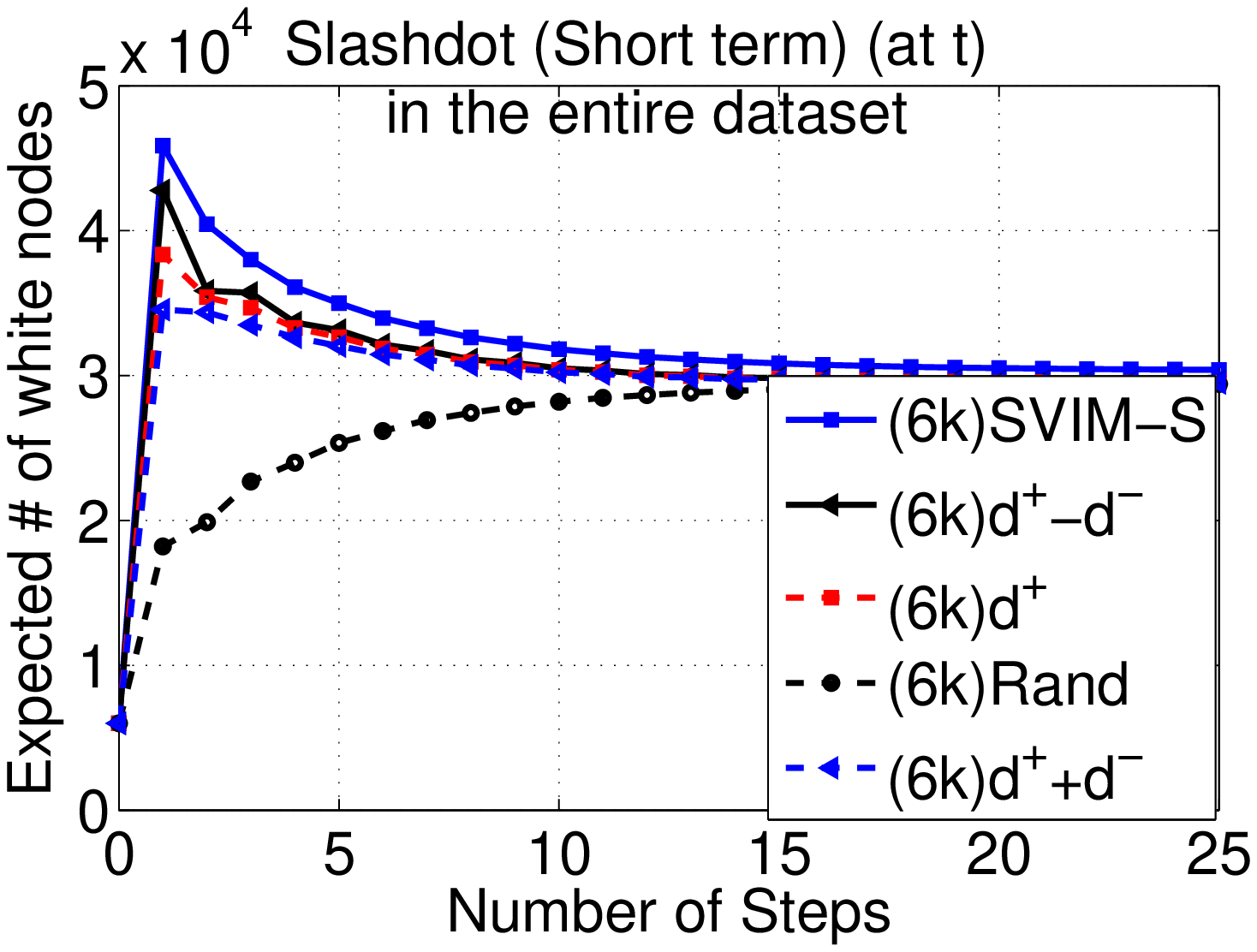}
    \vspace*{-0.3cm}
    \caption{\color{black}Instant influence in Slashdot data with $k=6$k}\label{fig:epsSdSin1entire}
    \end{center}
    \end{minipage}
    \centering
    \begin{minipage}[t]{6.7cm}
    \begin{center}
    \includegraphics[width=2.25in]{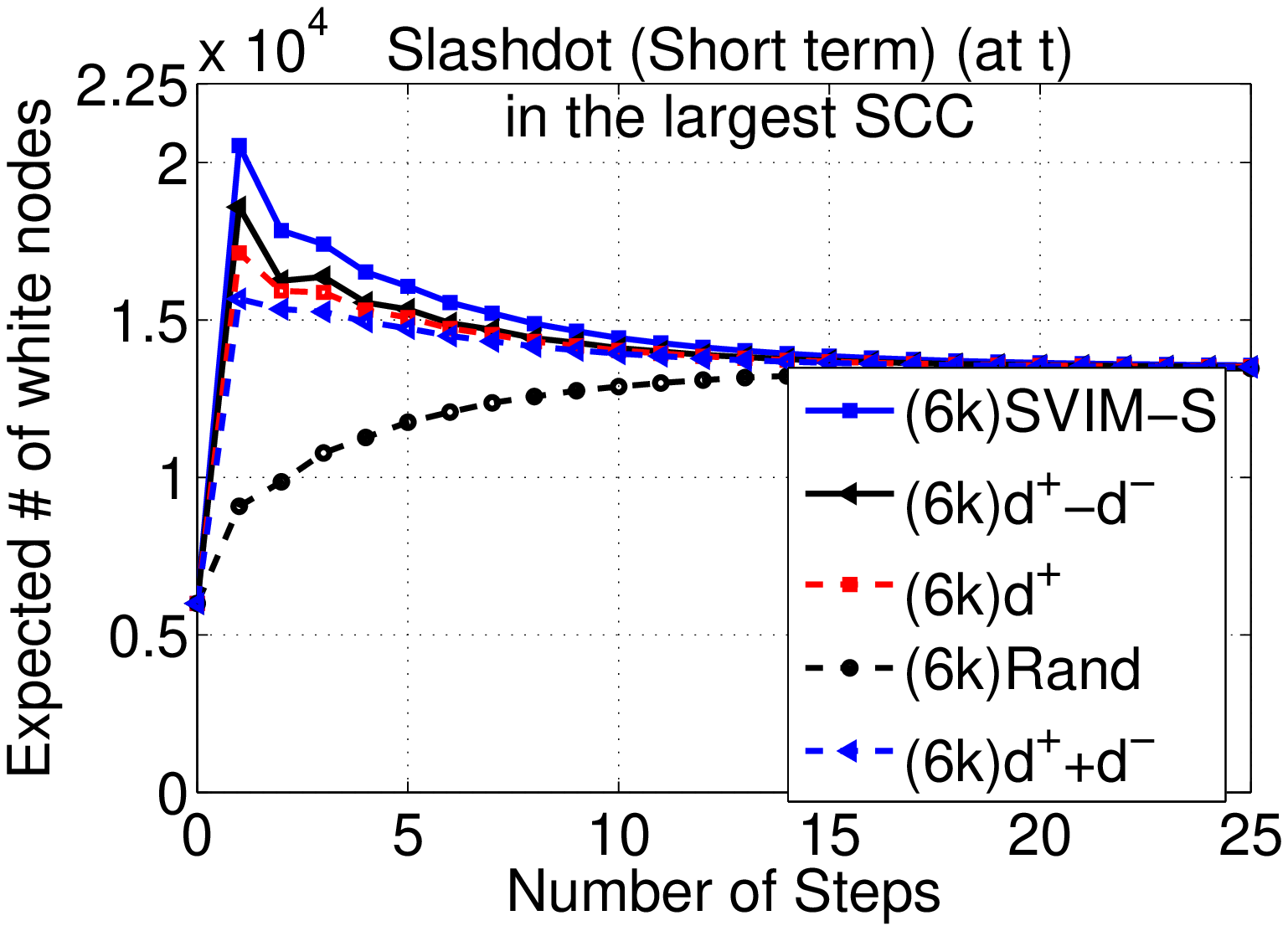}
    \vspace*{-0.3cm}
    \caption{\color{black}Instant influence in Slashdot SCC with $k=6$k}\label{fig:epsSdSin1}
    \end{center}
    \end{minipage}\\
    \centering
    \begin{minipage}[t]{6.7cm}
    \begin{center}
    \includegraphics[width=2.25in]{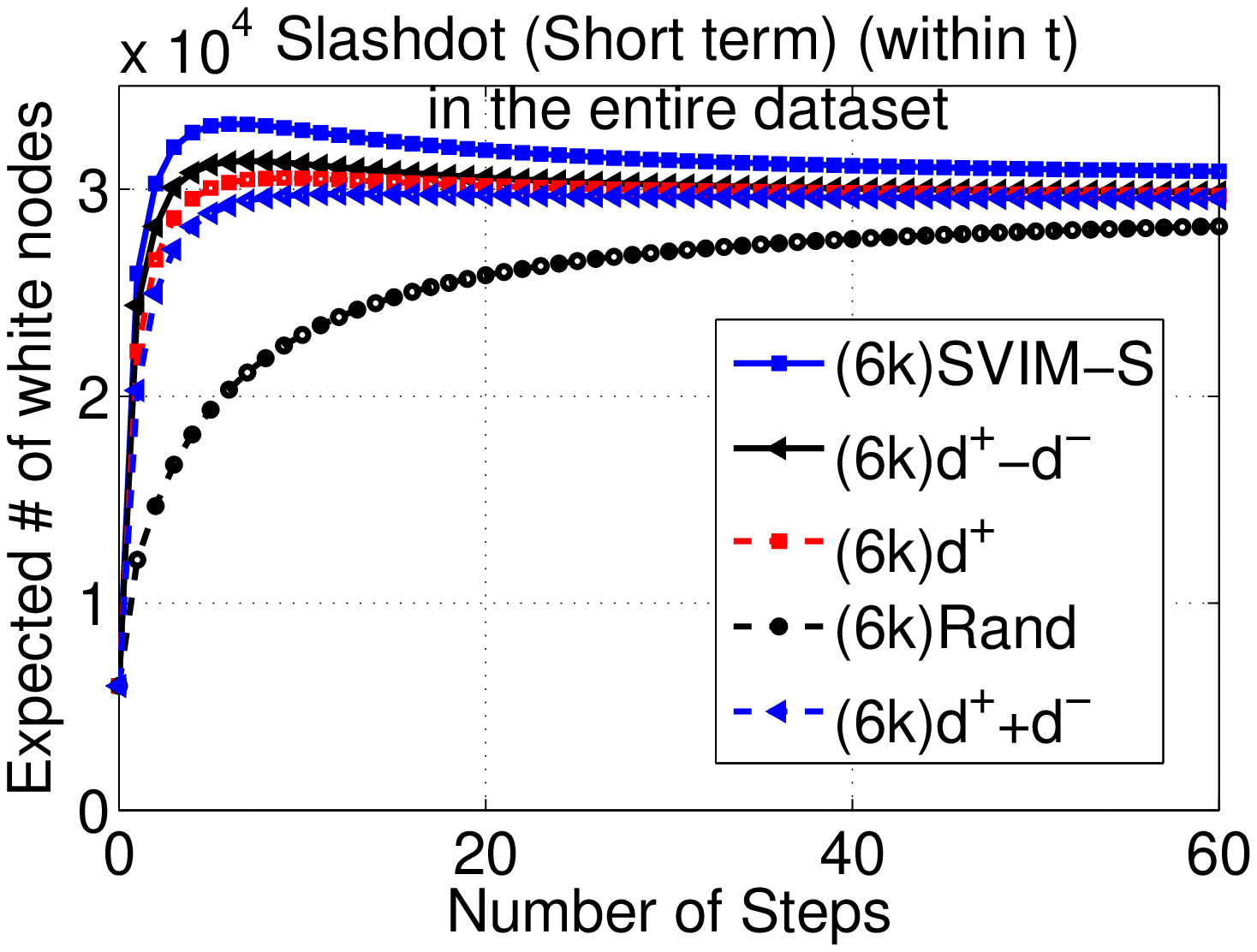}
    \vspace*{-0.3cm}
    \caption{\color{black}Average influence in Slashdot data with $k=6$k}\label{fig:epsSdSin3entire}
    \end{center}
    \end{minipage}
    \centering
    \begin{minipage}[t]{6.7cm}
    \begin{center}
    \includegraphics[width=2.25in]{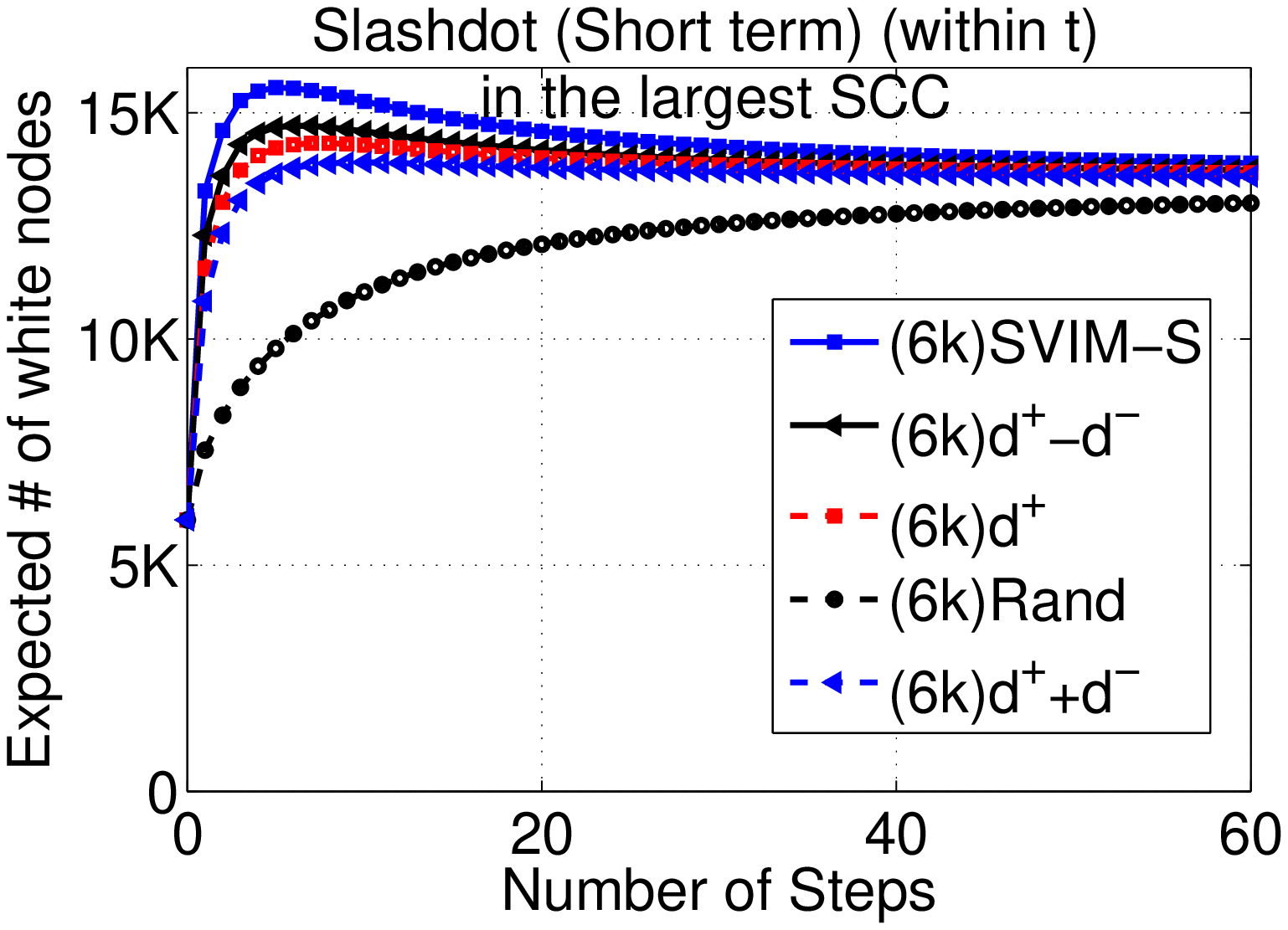}
    \vspace*{-0.3cm}
    \caption{\color{black}Average influence in Slashdot SCC with $k=6$k}\label{fig:epsSdSin3}
    \end{center}
    \end{minipage}
    \centering
  \vspace*{-0.4cm}
\end{figure*}

Moreover, the convergence times for both real-world datasets are
    fast, in a few tens of steps, indicating good connectivity and fast
    mixing property of real-world networks. In summary, our evaluation results on both synthetic and
    real-world networks validate our theoretical results and demonstrate
    that our SVIM algorithms for both short term and long term are
    indeed the best, and often have significant winning margins.

\vspace{-0.1cm}
\subsection{The impacts of signed information}
\vspace{-0.1cm}
\begin{figure*}[!htb]
    \centering
    \begin{minipage}[t]{6.7cm}
    \begin{center}
    \includegraphics[width=2.25in]{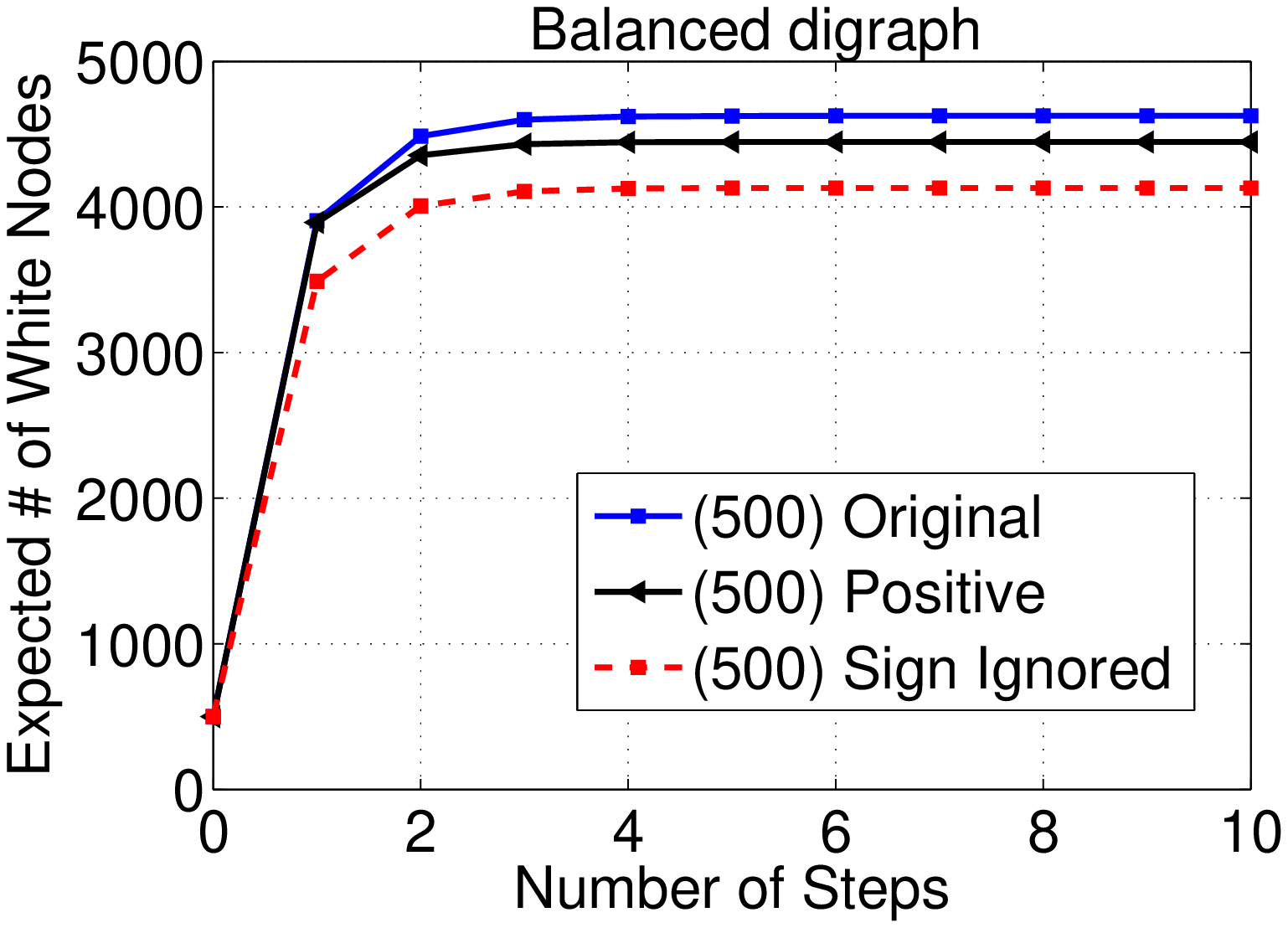}
             \vspace*{-0.4cm}
    \caption{Synthetic balanced digraph}\label{fig:eps11}
    \end{center}
    \end{minipage}
   \centering
    \begin{minipage}[t]{6.7cm}
    \begin{center}
    \includegraphics[width=2.25in]{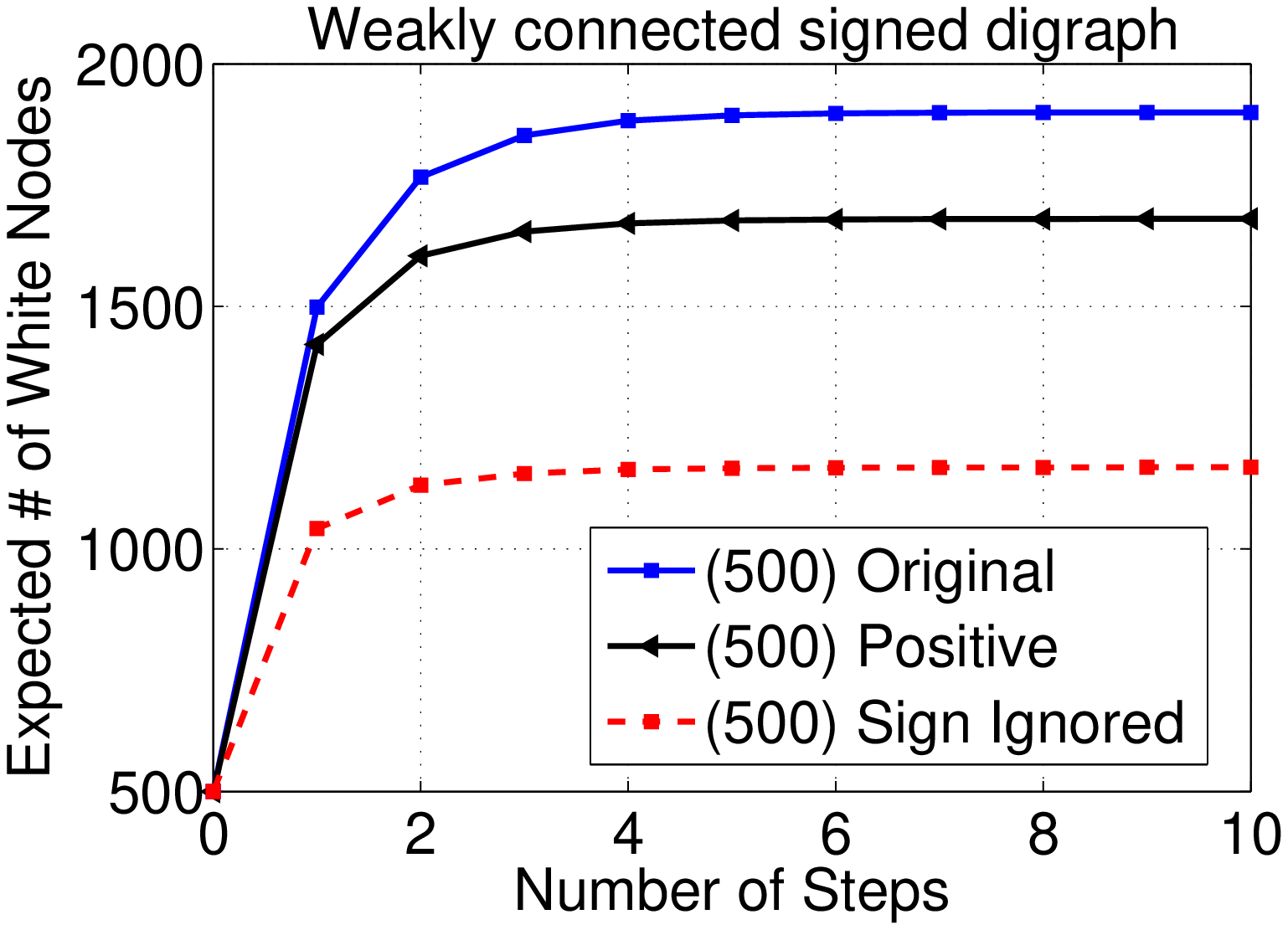}
             \vspace*{-0.4cm}
    \caption{Synthetic weakly connected digraph}\label{fig:eps12}
    \end{center}
    \end{minipage}\\
    \centering
    \begin{minipage}[t]{6.7cm}
    \begin{center}
    \includegraphics[width=2.25in]{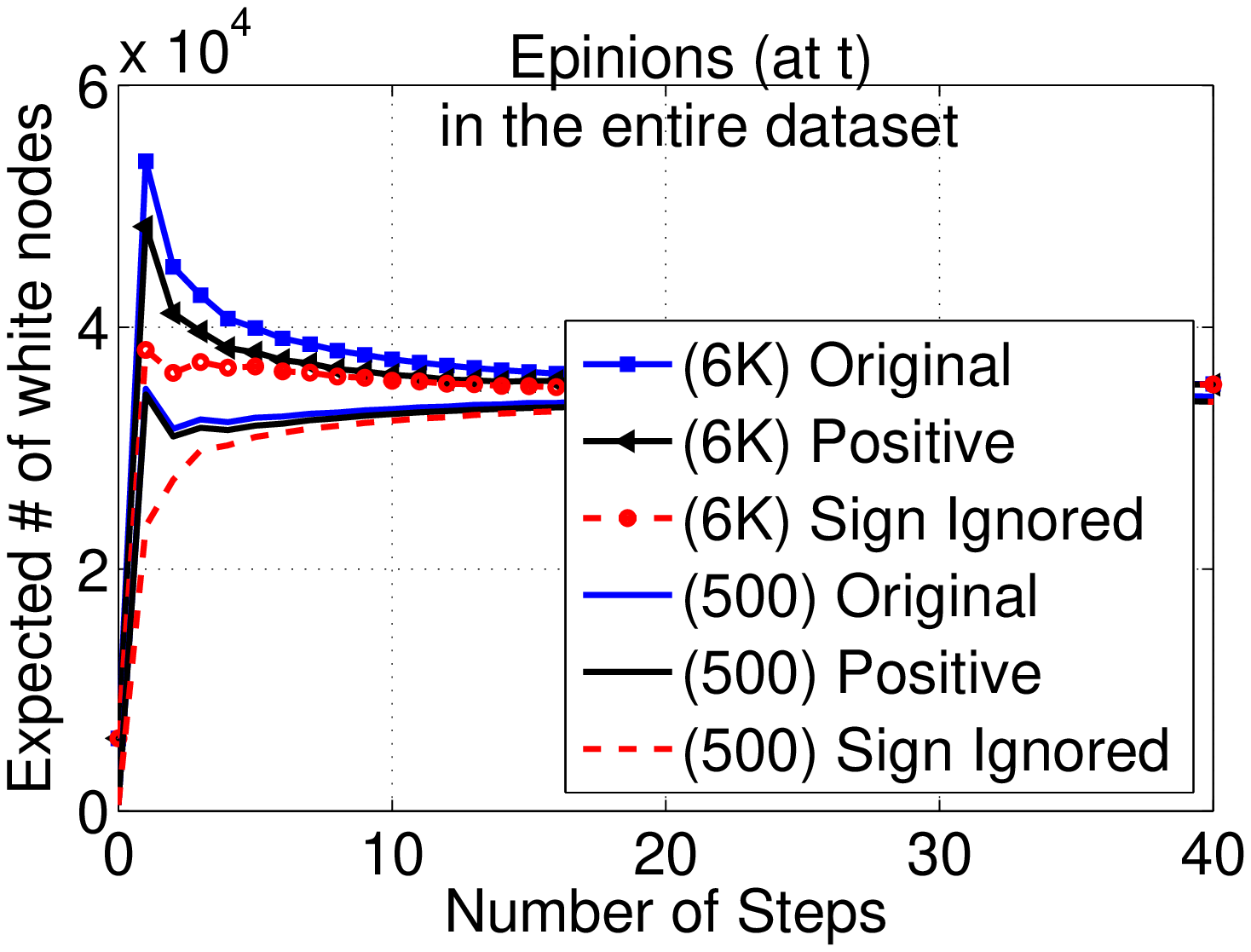}
             \vspace*{-0.4cm}
    \caption{Epinions (the entire dataset)}\label{fig:eps13}
    \end{center}
    \end{minipage}
    \centering
    \begin{minipage}[t]{6.7cm}
    \begin{center}
    \includegraphics[width=2.25in]{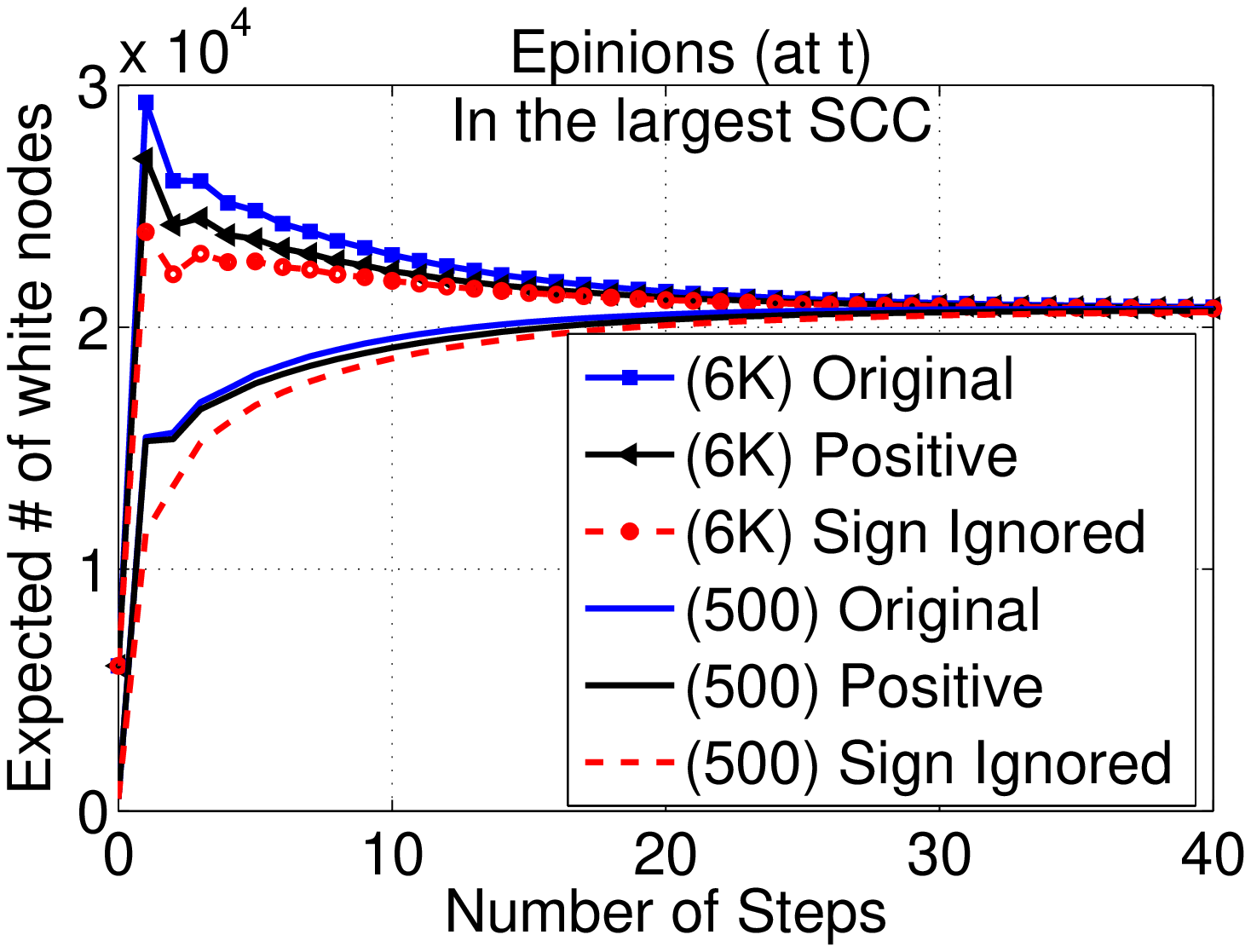}
             \vspace*{-0.4cm}
    \caption{Epinions (the largest SCC)}\label{fig:eps14}
    \end{center}
    \end{minipage}
%      \vspace*{-0.55cm}
  \vspace*{-0.55cm}
\end{figure*}

{\color{black}Unlike Epinions and Slashdot, many online social
networks such as Twitter are simply represented by unsigned directed
graphs, where friends and foe relationships are not explicitly
    represented on edges.
Without edge signs, two types of information may be
    mis-represented or under-represented:
    (1) one may follow his foes for tracking purpose, but this link may be
    mis-interpreted as friend or trust relationship; and
    (2) one may not follow his foes publicly to avoid  being noticed, but
    his foes may still generate negative influence to him.}
%% {\color{black}missing information} may emerge: (1) for
%% example, one may follow his foes for simply tracking and getting
%% aware of any updates from them, where without an edge sign, there is
%% no difference whether following a friend or a foe; (2) On the other
%% hand, one may ignore his foe, where these negative relations are not
%% indicated and hard to detect.}
{\color{black}In this section, we investigate how much influence
gain
can be obtained by taking the edge signs into consideration, %over
%those with sign information missing,
thus illustrate the significance of utilizing both friend and foe
relationships in influence maximization.}

{\color{black}Taking the synthetic networks and Epinions dataset
(used in Sec~\ref{sec:baseline}) as examples, we apply our SVIM
algorithm to compute the optimal initial seed sets in the original
signed digraphs, and two types of ``sign-missing'' scenarios,
{\color{black}i.e.,} the unsigned digraphs with only original
positive edges (denoted by ``Positive'' graphs) and with all edges
labeled by the same signs (denoted by ``Sign ignored'' graphs).
Then, we examine the performances of those three initial seed sets
in original signed digraphs.}
%such as the unsigned digraphs with only original positive edges
%(denoted by ``Positive'' graphs) and with all edges labeled by the
%same signs (denoted by ``Sign ignored'' graphs).
%Taking those synthetic networks, and Epinions dataset (used in
%Sec~\ref{sec:baseline}) as examples, we apply our SVIM algorithm to
%the original signed digraphs, and the corresponding ``Positive''
%graphs and the ``Sign Ignored'' graphs.

{\color{black}Fig.~\ref{fig:eps11}-\ref{fig:eps14} show the
evaluation results, where the seed sets obtained by considering edge
signs perform consistently better than those using unsigned graphs.
In synthetic networks,} {\color{black}we observed $5\%-16\%$ more
influence in balanced digraph for $t\geq6$ (See
Fig.~\ref{fig:eps11}), and $11.7\%-58\%$ more influence in weakly
connected digraph for $t\geq6$ (See Fig.~\ref{fig:eps12})}.
Moreover, in Epinions dataset from
Fig.~\ref{fig:eps13}-\ref{fig:eps14}, there is no impact on the
long-term influence, since the underlying graphs are strictly
unbalanced. However, {\color{black}in short term, the results
demonstrate that taking edge signs into consideration always
performs better, which generates at maximum of $38\%$ and $21\%$
more influence for the entire dataset (See Fig.~\ref{fig:eps13}) and
the largest SCC (See Fig.~\ref{fig:eps14}), respectively. Both
maximums occur at step $1$.{}
%Moreover, we observed similar results when using other real dataset,
%e.g., Slashdot.
These results clearly demonstrate the necessity of utilizing sign
information in influence maximization.}

%\begin{figure*}[!htb]
%\centering
%    \subfigure[Synthetic balanced digraph] {
%    \includegraphics[width=3in]{pb3}
%    \label{fig:eps1}
%    }
%    \subfigure[Synthetic weakly connected digraph] {
%    \includegraphics[width=3in]{wc3}
%    \label{fig:eps2}
%    }
%    \subfigure[Epinions (the largest strongly connected component)] {
%    \includegraphics[width=3in]{EpSCC3}
%    \label{fig:eps3}
%    }
%    \subfigure[Epinions (the entire dataset)] {
%    \includegraphics[width=3in]{EpEnt3}
%    \label{fig:eps4}
%    }
%\caption{Comparison among three graphs, i.e., original signed
%digraph (denoted by ``Original''), unsigned digraph with only
%positive edges (denoted by ``Positive''), unsigned digraph with all
%edges (using positive signs) (denoted by ``IgnoreSign''). The seed
%sets are the same computed by our SVIM algorithm for original signed
%digraph.}\label{fig:syn} \centering \vspace*{-0.2cm}
%\end{figure*}

%% file: ArXiv-sec-6con.tex
\section{Conclusion}\label{sec:con}

In this paper, we propose and study voter model dynamics on signed
digraphs, and apply it to solve the influence maximization problem.
We provide rigorous mathematical analysis to completely characterize
    the short-term and long-term dynamics, and provide
    efficient algorithms to solve both short-term and long-term
    influence maximization problems. Extensive simulation results on both synthetic and real-world graphs demonstrate
    the efficacy of our signed voter model influence maximization (SVIM) algorithms.
    We also identify a class of anti-balanced digraphs, which is not
covered in the social balance theory before, and exhibits
oscillating steady state behavior.
% Removed by Yanhua for technical report preparation.
%{\color{black} [removed a sentence.]}
%In addition, our model can be
%easily extended to cover heterogeneous seed costs and fractional
%seed selections, but due to space constraint we delegate these
%discussions to our technical report~\cite{VMreport}.

There exist several open problems and future directions.
One open problem is the convergence time of voter model dynamics on signed
    digraphs.
For balanced and anti-balanced ergodic digraphs, our results show that
    their convergence times are the same as the corresponding unsigned
    digraphs. %, which is an independent research topic.
For strictly unbalanced
    ergodic digraphs and more general weakly connected signed digraphs,
    the problem is quite open.
A future direction is to study influence diffusion in signed networks
    under other models, such as the voter model with a background color,
    the independent cascade model, and the linear threshold model.

%% file: ArXiv-sec-7app.tex
\color{black}
%\section{Omitted proof}
{\color{black}
%\section{Convergence of $\bar{P}^t$ of ergodic digraphs}\label{sec:exponentialofbarP}
\section{Properties of ergodic digraphs}\label{sec:exponentialofbarP}

\begin{proposition}[]\label{pro:evenoddpath}
Let $G=(V,E,A)$ be an ergodic digraph. For any nodes $i,j\in V$,
there exist two paths from $i$ to $j$ with even and odd length,
respectively.
\end{proposition}
\begin{proof}
Suppose, for a contradiction, that all paths from $i$ to $j$ have even
    lengths.
This implies that all cycles passing through $i$ must be even length, since
    otherwise we could follow node $i$'s odd-length cycle followed
    by the even length path from $i$ to $j$, making the entire path
    from $i$ to $j$ odd.
Now we can consider any cycle $C_r$ in $G$, not necessarily passing
$i$. We claim that $C_r$ must have even length. In fact, we can pick
any node $u$ on $C_r$, and construct a path from $i$ to $j$ with the
following segments: $R_1$ from $i$ to $u$,
    $C_r$, $R_2$ from $u$ back to $i$, and
    $R_3$ from $i$ to $j$.
Since we know that $R_1+R_2$ has even length and $R_3$ has even length,
    it must be the case that $C_r$ has even length by our assumption.
However, this means that all cycles in $C$ has even lengths, contradicting
    to the aperiodicity of $G$.

The case of odd length paths can be proved in the same way.
\end{proof}

\begin{proposition}[]\label{pro:exponentialofbarP}
Let $\bar{G}=(V,E,\bar{A})$ be an ergodic unsigned digraph, with
transition probability matrix $\bar{P}$ and stationary distribution
vector $\pi$.
%\noindent(1)The matrix sequence $\bar{P}^t$, for $t>0$,
%exponentially converges to $\textbf{1}\pi^T$, i.e.,
%$\lim_{t\rightarrow \infty}\bar{P}^t=\textbf{1}\pi^T$.
%\noindent(2)
$\bar{P}^t-\textbf{1}\pi^T=(\bar{P}-\textbf{1}\pi^T)^t$
holds for any integer $t>0$.
\end{proposition}
\begin{proof}

%\noindent(1) Let $\bar{u}_0$ be a non-negative distribution vector,
%with $\|\bar{u}_0\|_1=1$ and $\bar{u}_0(k)\geq 0$ for $k\in V$, and
%$U$ be the set of all such vectors. Since $\bar{G}$ is ergodic,
%$\exists m>0$, such that for any node pair $i$ and $j$ in $V$, there
%exists a path from $i$ to $j$, with length equal to $m$, which
%implies that $u_1^T=\bar{v}_0^T\bar{P}^m$, such that $u_1(k)>0$ for
%any $k\in V$. Let $\alpha:=\min_{\bar{u}_0\in U}\{\min_{k\in
%V}u_1(k)\}$, which is thus an constant for the given graph
%$\bar{G}$. We can rewrite $u_1$ as $u_1=\alpha \pi +
%(1-\alpha)\bar{u}_1$, where $\bar{u}_1=\frac{{u}_1}{\|{u}_1\|_1}$
%and $\bar{u}_1\in U$.
%
%Hence, by induction, we have for any $t=Km+t_0$, with integer
%$K=\lfloor\frac{t}{m}\rfloor\geq 0$ and $0\leq t_0< m$,
%\begin{align}
%\bar{u_0}^T\bar{P}^{t} &= \bar{u_0}^T\bar{P}^{Km+t_0}
%=(1-(1-\alpha)^K)\pi^T+(1-\alpha)^K \bar{u}_{K}^T\bar{P}^{t_0},
%\nonumber
%\end{align}
%\begin{align}
%\mbox{and } &
%\|\bar{u}_0^T\bar{P}^{t}-\pi^T\|_1=(1-\alpha)^K\|\pi^T-\bar{u}_{K}^T\bar{P}^{t_0}\|_1
%\leq 2(1-\alpha)^K\leq
%2(1-\alpha)^{\frac{t}{2m}}.\label{eq:barPconv}
%\end{align}
%Eq.(\ref{eq:barPconv}) implies
%$\|\bar{P}^{t}-\textbf{1}\pi^T\|_\infty\leq
%(1-\alpha)^{\frac{t}{2m}}$, for $t>2m$, namely, $\bar{P}^t$
%converges exponentially to $\textbf{1}\pi^T$.

%\noindent(2)
Using the facts that $\bar{P}\textbf{1}=\textbf{1}$ and
$\pi^T\bar{P}=\pi^T$, it is easy to prove by induction that for any
integer $t>0$
$\bar{P}^t-\textbf{1}\pi^T=(\bar{P}-\textbf{1}\pi^T)^t$ holds.
\end{proof}
}

{\color{black}
\section{Special matrix power series}\label{sec:Matrixform} %{\color{black}(We are currently not using this subsection)}}
%
%{\color{black}[Wei: Do we really need (iii) below?
%Eq.(24) seem to be the one close to (iii), but the proof there does not
%    use this result. We can remove (iii) if it is not used.]}

\begin{proposition}
\label{thm:matrixlimit} Let $X \in \mathbb{R}^{m\times m}$,
$Y\in\mathbb{R}^{m\times n}$ and $Z \in \mathbb{R}^{n\times n}$. If
$\limt X^t = \limt Z^t = \textbf{0}$, the following equalities hold:
\begin{align}
\mbox{(i) }&\limt \sum_{i=0}^{t-1} X^i = (I-X)^{-1},\label{eq:MF1}\\
\mbox{(ii) }&\limt \sum_{i=0}^{t-1}X^iYZ^{t-1-i} = 0,\label{eq:MF2} %\\
%%%Reomoved by Yanhua, due to the redudency. 11-18-2011
%%\mbox{(iii) }&\limt \sum_{i=0}^{t-1}\sum_{j=0}^iX^jYZ^{i-j} =
%%(I-X)^{-1}Y(I-Z)^{-1}.\label{eq:MF3}
\end{align}
\end{proposition}
\begin{proof}
\noindent(i) Let $\rho(X)$ be the spectral radius of matrix $X$,
i.e., the largest absolute value of the eigenvalues of $X$. Notice
that $\limt X^t = \textbf{0}$ if and only if $\rho(X) <1$.

We first claim that, $I-X$ and $I-Z$ are invertible. Suppose $I-X$
is not invertible, there is a non-zero vector $p$ such that $(I-X)p
= \textbf{0}$. Therefore, $p$ is the eigenvector of $X$ with
eigenvalue $1$, which contradicts $\limt X^t =\textbf{0}$. Same
argument can be applied to $I-Z$. Hence, the left hand side of
Eq.(\ref{eq:MF1}) equals to
\begin{align}
&\limt \sum_{i=0}^tX^i = \limt (I-X)^{-1}(I-X^{t+1}) =
(I-X)^{-1}.\nonumber
\end{align}

\noindent(ii) The max-norm of $X$ is given by
$\|X\|_{max}=\max_{i,j\leq m}\{X_{ij}\}$. Let $X=Q_X{\cal
J}Q_X^{-1}$ be the standard Jordan form of $X$, where $Q_X$ is an
invertible matrix. Denote $J=\textbf{1}\textbf{1}^T$ as the all-one
matrix. Hence, we have
\begin{align}
\|X^i\|_{max} &= \|Q_X{\cal J}^iQ_X^{-1}\|_{max} \leq
\|Q_X\|_{max}\|Q_X^{-1}\|_{max}\|J{\cal J}^iJ\|_{max} \nonumber\\
&\leq \|Q_X\|_{max}\|Q_X^{-1}\|_{max} m^2\|{\cal
J}^i\|_{max}\nonumber% \\
%\leq& \|Q_X\|_{max}\|Q_X^{-1}\|_{max} m^2 i^m \rho^{i-m}(X)
%\nonumber
\end{align}

${\cal J}^i$ is in form as
%$\left(\begin{array}{c}
%            i  \\
%            l \\
%            \end{array}\right)\lambda_i$
\begin{align}
\small {\cal J}^i&=\left[
            \begin{array}{c c c c c}
            \lambda_1^i & C_i^1\lambda_1^{i-1} & C_i^2\lambda_1^{i-2} & {0} & {0} \\ %\hdashline
            {0} & \lambda_1^i & C_i^1\lambda_1^{i-1} & {0} & {0} \\ %\hdashline
            {0} & {0} & \lambda_1^i & {0} & {0} \\ %\hdashline
            {0} & {0} & {0} & \lambda_{m_0}^i & C_i^1\lambda_{m_0}^{i-1}  \\ %\hdashline
            {0} & {0} & {0} & {0}  & \lambda_{m_0}^i \\ %\hdashline
            \end{array}\right]\label{eq:jordan},
\end{align}
where $C_i^\ell=\frac{i!\ell!}{(i-\ell)!}\leq i^m$ and each non-zero
entry in ${\cal J}^i$ can be expressed as
$C^i_\ell\lambda_k^{i-\ell}$, $1\leq k\leq m_0$, $1\leq \ell\leq
\ell_0(k)$, with $m_0$ as the number of different eigenvalues of $X$
and $\ell_0(k)$ as the multiplicity of the $k$-th eigenvalue of $X$.
Hence, the absolute value of each non-zero entry in ${\cal J}^i$ is
upper bounded as $|C_i^\ell\lambda_k^{i-\ell}|\leq i^m
\rho(X)^{i-m}$, which implies that
\begin{align}
&\|X^i\|_{max} \leq \|Q_X\|_{max}\|Q_X^{-1}\|_{max} m^2 i^m
\rho(X)^{i-m} \nonumber
\end{align}
%\begin{align}
%\small {\cal J}^i&=\left[
%            \begin{array}{c:c:c:c:c:c}
%            \lambda_1^i & \left(\begin{array}{c}
%            i  \\
%            l \\
%            \end{array}\right)\lambda_1^{i-1} & \left(\begin{array}{c}
%            i  \\
%            2 \\
%            \end{array}\right)\lambda_1^{i-2} & \cdots & \textbf{0} & \textbf{0} \\ \hdashline
%            \textbf{0} & \lambda_1^i & \left(\begin{array}{c}
%            i  \\
%            l \\
%            \end{array}\right)\lambda_1^{i-1} & \cdots & \textbf{0} & \textbf{0} \\ \hdashline
%            \textbf{0} & \textbf{0} & \lambda_1^i & \cdots & \textbf{0} & \textbf{0} \\ \hdashline
%            \cdots & \cdots & \cdots & \ddots & \cdots & \cdots \\ \hdashline
%            \textbf{0} & \textbf{0} & \textbf{0} & \cdots & \lambda_m^i & \left(\begin{array}{c}
%            i  \\
%            l \\
%            \end{array}\right)\lambda_m^{i-1}  \\ \hdashline
%            \textbf{0} & \textbf{0} & \textbf{0} & \cdots & \textbf{0}  & \lambda_m^i \\ \hdashline
%            \end{array}\right]\label{eq:Pmatmultisink}
%\end{align}

Let $\rho = \max(\rho(X), \rho(Z))$, we have
\begin{align}
&\limt \|\sum_{i=0}^{t-1}X^iYZ^{t-1-i}\|_{max}\leq \limt t m n \|X^i\|_{max}\|Y\|_{max}\|Z^{t-1-i}\|_{max}\nonumber \\
\leq& \limt t m n T_{max} ( m^2 t^m \rho^{i-m}) (n^2 t^n
\rho^{t-i-1-n}) \leq \limt m^3 n^3 T_{max} t^{m+n+1}
\rho^{t-1-n-m}=0\nonumber
\end{align}
where $T_{max}=\|Y\|_{max}\|Q_X\|_{max}\|Q_X^{-1}\|_{max}
\|Q_Z\|_{max}\|Q_Z^{-1}\|_{max}$.

%%Reomoved by Yanhua, due to the redudency. 11-18-2011
%%\noindent(iii) Now we focus on the third equality.
%%\begin{align}
%%&\limt \sum_{i=0}^{t-1}\sum_{j=0}^iX^jYZ^{i-j}= \limt
%%\sum_{i=0}^{t-1}\sum_{j=0}^{t-1-i}X^jYZ^i
%%=  \limt \sum_{i=0}^{t-1} (I-X)^{-1}(I-X^{t-i})YZ^i\nonumber \\
%%= & \limt  (I-X)^{-1}Y( \sum_{i=0}^{t-1}Z^i) + \limt
%%\sum_{i=0}^{t-1} X^{t-i} YZ^i = (I-X)^{-1}Y(I-Z)^{-1}.\nonumber
%%\end{align}

%The last equality comes from the first two equalities we proved
%above.
\end{proof}
}
% Removed by Yanhua for technical report. The following is moved to the main content.
%\begin{lemmaa}
%\label{lem:unbalancedlimit}
%Let $G=(V,E,A)$ be an ergodic signed digraph and $P$ be the {\em
%signed   transition matrix} of $G$.
%If $G$ is strictly unbalanced,
%$$\limt P^t = \textbf{0}.$$
%\end{lemmaa}
%\begin{proof}
%Fix a pair of  nodes $i$ and $j$, by Proposition~\ref{pro:UnbalancedDig},
%there exists two path $R_1$ and $R_2$ with the same length $\ell(i)$
%which have different signs. Consider a random walk from $i$. Let $p_1$
%(resp. $p_2$)
%be the probability that the walk exactly follows $R_1$ (resp. $R_2$)
%in the first $\ell(i)$ steps. Let $R^{\ell(i)}_{i,k}$ be the set of
%all paths from $i$ to $k$ with length $\ell(i)$.
%Then, for a unit vector $I_i$ with $i$-th entry $1$ and $0$ otherwise,
%we have,
%
%$$||I_i^TP^{\ell(i)}||_1 =\sum_{k \in[n]} \left|\sum_{R\in R^{\ell(i)}_{i,k}}
%\mathrm{Prob}[R] sign(R)\right| \leq 1-\min(p_1,p_2) = \rho_i.$$
%
%Notice that the $L_1$ norm always shrinks. Consider all $\rho_i$ and
%$\ell(i)$. Let $\rho = \max_i \rho_i <1$ and $\ell = \max_i
%\ell(i)$. For any $i$, we have
%
%$$||I_i^TP^{\ell}||_1 <\rho\mbox{ and } \lim_{t\rightarrow \infty}
%I_i^TP^t ={\bf 0}$$
%
%This implies $\lim_{t\rightarrow \infty} P^t ={\bf
%  0}$ since its null space is $\mathbf{R}^{|V|}$.
%\end{proof}

%\section{deferred portion of Section 3}

%\section{Proof of Remark~\ref{rmk:Rmax}}\label{sec:remark}

\section{Illustration of exponential convergence time of $P^t$ on
ergodic digraph.}\label{sec:illustration}

\begin{figure}[htb]
\centering
\includegraphics[width=0.5 \textwidth]{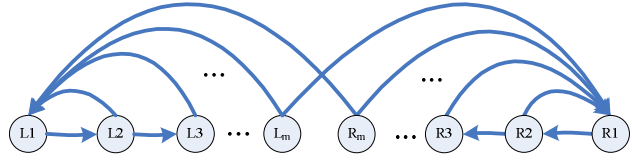} %\vspace{-3mm}
\caption{An example digraph with exponential convergence time. All
edges are with unit weights.}\label{Plot1} \vspace*{-0.3cm}
\end{figure}

{\color{black}{Given an unsigned ergodic digraph
$\bar{G}=(V,E,\bar{A})$, with transition probability matrix
$\bar{P}$, it has fixed stationary distribution $\pi$, i.e.,
$\pi^T=\pi^T\bar{P}$.

The convergence time (or mixing time) of a random walk Markov chain
on $G$ is the time until the Markov chain is ``close'' to its
stationary distribution $\pi$. To be precise, for an initial
distribution $x_0$, let $x_t^T=x_0^T\bar{P}^t$ be the distribution
at step $t$. The variation distance mixing time is defined as the
smallest $t$  such that for any subset $W\subseteq V$,
$$|(x_t^T-\pi^T)e_W|\leq \frac{1}{4},$$
where $e_W$ is the vector such that $e_W(i)=1$ if $i\in W$, and
$e_W(i)=0$ if $i\in V\setminus W$.

%The convergence time of a random walk on $G$ with initial
%distribution $x_0$ is the minimum $t>0$ such that $\|x_0^T \bar{P}^t
%- \pi^T\|_{\infty} \leq \frac{1}{4n}$, where $n=|V|$.
The convergence time is said to be exponentially large if there
exists $x_0$ such that the convergence time of the random walk
starting from $x_0$ is $2^{\Omega(n)}$, where $n=|V|$.
Lemma~\ref{lem:ExpConv} below illustrates that the convergence time
of random walk on ergodic digraphs could be exponentially large.}}

\begin{lemmaa}\label{lem:ExpConv}
There exist ergodic digraphs, such that the convergence time of the
random walks on these digraphs are exponentially large.
%namely, given an initial distribution $x_0=[x_0(i)]$, for $i\in V$,
%the total variation distance between $x_t^T=x_0^T\bar{P}^t$ and
%$\pi^T$ satisfies
%\begin{align}
%&\|x_t^T-\pi^T\|_{\infty}\leq \epsilon, \mbox{ for small
%$\epsilon>0$.}\label{eq:convxt}
%\end{align}
\end{lemmaa}

%Using Fig.~\ref{Plot1} as an example which is, we will show that
%with some initial distribution
%
%We will show using the that there exists
%
%Given a small positive real value $\epsilon>0$, the convergence time
%$t_0$ of $P$, is defined as for any distribution vector $x_0$,
%

%Given a small positive real value $\epsilon>0$, the convergence time
%$t_0$ of $P$, is defined as for any distribution vector $x_0$,
%\begin{align}
%t_0&:=\argmin_{t>0}\{\|x_0^TP^t-\pi^T\|_2\}\leq \epsilon\nonumber.
%\end{align}

% Given a small
%positive real value $\epsilon>0$, the convergence time $t_0$ of $P$,
%is defined as
%\begin{align}
%t_0&:=\argmin_{t>0}\{\|P^t-\textbf{1}\pi^T\|_2\}\leq
%\epsilon\nonumber.
%\end{align}
\begin{proof}

We prove this by construction. Fig.~\ref{Plot1} shows an example
digraph $G$, with $|V|=2m$ nodes. On the left hand side, there are
$m\geq 3$ nodes $L_1, L_2, \cdots, L_m$ connected by $m-1$ directed
edges from $L_1$ to $L_m$, and every node $L_i$ with $i>1$ has a
directed connection to the leftmost node $L_1$. The right hand side
nodes have symmetric connections as the left hand side. Moreover,
node $L_m$ and $R_m$ also have one more connection to $R_1$ and
$L_1$, respectively, which connect two components together. It is
clear that the graph is strongly connected and aperiodic
    (there exist cycles of length $2$ and $3$), and thus
    ergodic.
%% In this
%% example, every node has loops with both length $2$ and $3$, which
%% infers that $G$ is aperiodic. Since every node can reach any other
%% node, $G$ is strongly connected, thus ergodic.

Let $x_t(L_i)$ denote the probability that the random walk is at
node $L_i$ at step $t$, and $x(L_i)$ be its stationary distribution.
Similarly define $x_t(R_i)$ and $x(R_i)$ for node $R_i$. The graph
is symmetric, thus we have $x(L_i)=x(R_i)$ for $1\leq i\leq m$. Let
$x(L_1)=x(R_1)=\rho/4$, we have $x(L_i)=x(R_i)=\rho/2^i$ for
$i=2,3,\ldots, m$. Then, by solving
$\sum_{i=1}^{m}(x(L_i)+x(R_i))=1$, we obtain
$\rho=\frac{2^{m-1}}{3\cdot 2^{m-2}-1}$. It is easy to verify that
indeed the obtained $x$ is the stationary
    distribution of the random walks on the digraph.

Then, we consider the initial distribution as $x_0=[1,0,0,\ldots,
0]$, and the subset $W=\{R_1, \cdots, R_m\}$ including all $m$ nodes
on the right-hand side.
Let $x_t(W)=x_t^T \cdot e_W$ denote the total probability that the random walk is
     in some node in $W$ at step $t$.
The only edge from the left half to the right half is the edge from
    $L_m$ to $R_1$.
Thus all additions to $x_{t+1}(W)$ from $x_t(W)$ comes from this edge, namely
    $x_{t+1}(W) - x_{t}(W) \le x_t(L_m)/2$.
We now bound $x_t(L_m)$.
For $t\le m-1$, we know that $x_t(L_m) = 0$.
For $t\ge m$, we have
$$x_t(L_m)=x_{t-1}(L_{m-1})/2 = x_{t-2}(L_{m-2})/2^2 =
    \cdots = x_{t-m+2}(L_{2})/2^{m-2} \le 1/2^{m-2}.$$
Hence, we have
$$x_t(W) = \sum_{i=1}^{t}(x_i(W)-x_{i-1}(W))
    \le t \cdot x_t(L_m)/2 \le t/2^{m-1}.$$
Therefore, the smallest $t$ that satisfies
    $|(x_t^T-\pi^T)e_W| = |x_t(W) - 1/2| \leq 1/4$ is such that
    $x_t(W) \ge 1/4$, which implies that
    $t/2^{m-1}\ge 1/4$ and $t\ge 2^{m-3}$.
This completes the proof.
\end{proof}